\documentclass[11pt]{article}
\usepackage[T1]{fontenc}
\usepackage[margin=1in]{geometry}
\usepackage{setspace}

\usepackage{amsmath}
\usepackage{amssymb}
\usepackage{amsthm}
\usepackage{caption}
\usepackage{graphicx}
\usepackage{mathrsfs} 

\usepackage{dsfont}
\usepackage{comment}
\usepackage{color}
\usepackage{bbm}
\usepackage{hyperref}
\usepackage{bm}
\usepackage{xy}
\usepackage{enumitem}
\usepackage[
backend=biber,
style=authoryear,
]{biblatex}
\usepackage[ruled,vlined]{algorithm2e}

\usepackage{xr}
\makeatletter
\newcommand*{\addFileDependency}[1]{
  \typeout{(#1)}
  \@addtofilelist{#1}
  \IfFileExists{#1}{}{\typeout{No file #1.}}
}
\makeatother

\usepackage[colorinlistoftodos]{todonotes}

\usepackage{tikz}
\usetikzlibrary{decorations.markings}
\usetikzlibrary{decorations.pathmorphing}
\usetikzlibrary{arrows, snakes}
\usetikzlibrary{shapes.geometric, calc,decorations.pathreplacing}
\usetikzlibrary{backgrounds}
\usetikzlibrary{fit}
\usetikzlibrary{decorations.pathreplacing}
\input xy
\xyoption{all}

\tikzset{middlearrow/.style={
		decoration={markings,
			mark= at position 0.5 with {\arrow{#1}} ,
		},
		postaction={decorate}
	}
}

\theoremstyle{plain}
\newtheorem{theorem}{Theorem}[section]

\newtheorem{corollary}[theorem]{Corollary}

\newtheorem{lemma}[theorem]{Lemma}

\newtheorem{proposition}[theorem]{Proposition}

\newtheorem{definition}[theorem]{Definition}
\newtheorem{condition}[theorem]{Condition}

\theoremstyle{remark}
\newtheorem{remark}[theorem]{Remark}
\newtheorem{example}[theorem]{Example}

\numberwithin{equation}{section}

\providecommand{\keywords}[1]{\textbf{\textsf{Keywords: }} #1}

\newcommand\myprime{\mkern-3.5mu\raise0.6ex\hbox{$\scriptstyle\prime$}}

\newcommand{\conv}{\mathrm{conv}}

\def\CC{{C\nolinebreak[4]\hspace{-.05em}\raisebox{.4ex}{\tiny\bf ++}}}

\DeclareMathOperator*{\intt}{int}

\DeclareMathOperator*{\argmin}{arg\,min}
\DeclareMathOperator*{\Glim}{
\Gamma-\,lim}
\DeclareMathOperator*{\Klim}{K-\,lim}

\makeatletter
\newcommand*\bigcdot{\mathpalette\bigcdot@{.8}}
\newcommand*\bigcdot@[2]{\mathbin{\vcenter{\hbox{\scalebox{#2}{$\m@th#1\bullet$}}}}}
\makeatother

\title{Efficient convex PCA with applications to Wasserstein GPCA and ranked data}
\author{
  Steven Campbell\\
  Department of Statistics\\
  Columbia University\\
  \texttt{sc5314@columbia.edu}
  \and
  Ting-Kam Leonard Wong\\
  Department of Statistical Sciences\\
  University of Toronto\\
  \texttt{tkl.wong@utoronto.ca}
}
\date{\today}

\begin{document}

\maketitle

\begin{abstract}
Convex PCA, which was introduced in \cite{BGKL17}, modifies Euclidean PCA by restricting the data and the principal components to lie in a given convex subset of a Hilbert space. 
This setting arises naturally in many applications, including distributional data in the Wasserstein space of an interval, and ranked compositional data under the Aitchison geometry. Our contribution in this paper is threefold. First, we present several new theoretical results including consistency as well as continuity and differentiability of the objective function in the finite dimensional case. Second, we develop a numerical implementation of finite dimensional convex PCA when the convex set is polyhedral, and show that this provides a natural approximation of Wasserstein GPCA. Third, we illustrate our results with two financial applications, namely distributions of stock returns ranked by size and the capital distribution curve, both of which are of independent interest in stochastic portfolio theory. Supplementary materials for this article are available online.
\end{abstract}


\newpage

\section{Introduction} \label{sec:intro}

We study {\it convex principal component analysis} (CPCA) which extends the standard Euclidean PCA to data  constrained to lie in a convex subset of a Hilbert space. The precise mathematical formulation is described in Section \ref{sec:CPCA.general}. To give right away the main idea of CPCA,  we compare in  Figure \ref{fig:Low.Dim.Example} Euclidean and convex PCA using a simulated two dimensional data-set, where the data is restricted to lie in a convex cone. The orthogonal projection of a data point onto the Euclidean principal component (PC) may not lie in the convex set, making it difficult to interpret the result. By restricting the projected points to the convex set, we obtain (using our algorithm developed in Section \ref{sec:algorithm}) the convex PC. CPCA was first introduced in \cite{BGKL17} to formulate a notion of geodesic PCA (GPCA) on the Wasserstein space on an interval, where each data point is a probability distribution.  This approach is based on the fact that the $2$-Wasserstein space on an interval is isometric to a convex subset of an $L^2$ space (see Section \ref{sec:Wasserstein.isometry}). Apart from Wasserstein GPCA, convex PCA is also useful in other settings. In Section \ref{sec:rank.based.dist} we apply CPCA to a {\it ranked} data-set where each observation is a vector $(y_1, \ldots, y_d)$ satisfying the convex constraint $y_1 \geq \cdots \geq y_d$.

\subsection{Main contributions and organization of the paper}
Our contribution is threefold. First, we obtain new theoretical results about CPCA. 
Second, we present a numerical implementation of CPCA when the convex state space is polyhedral, and apply it to solve a finite dimensional approximation of Wasserstein GPCA. Third, we provide two novel financial applications of CPCA.

\begin{figure}[t]
\centering
\begin{center}
\includegraphics[width=0.7\textwidth]{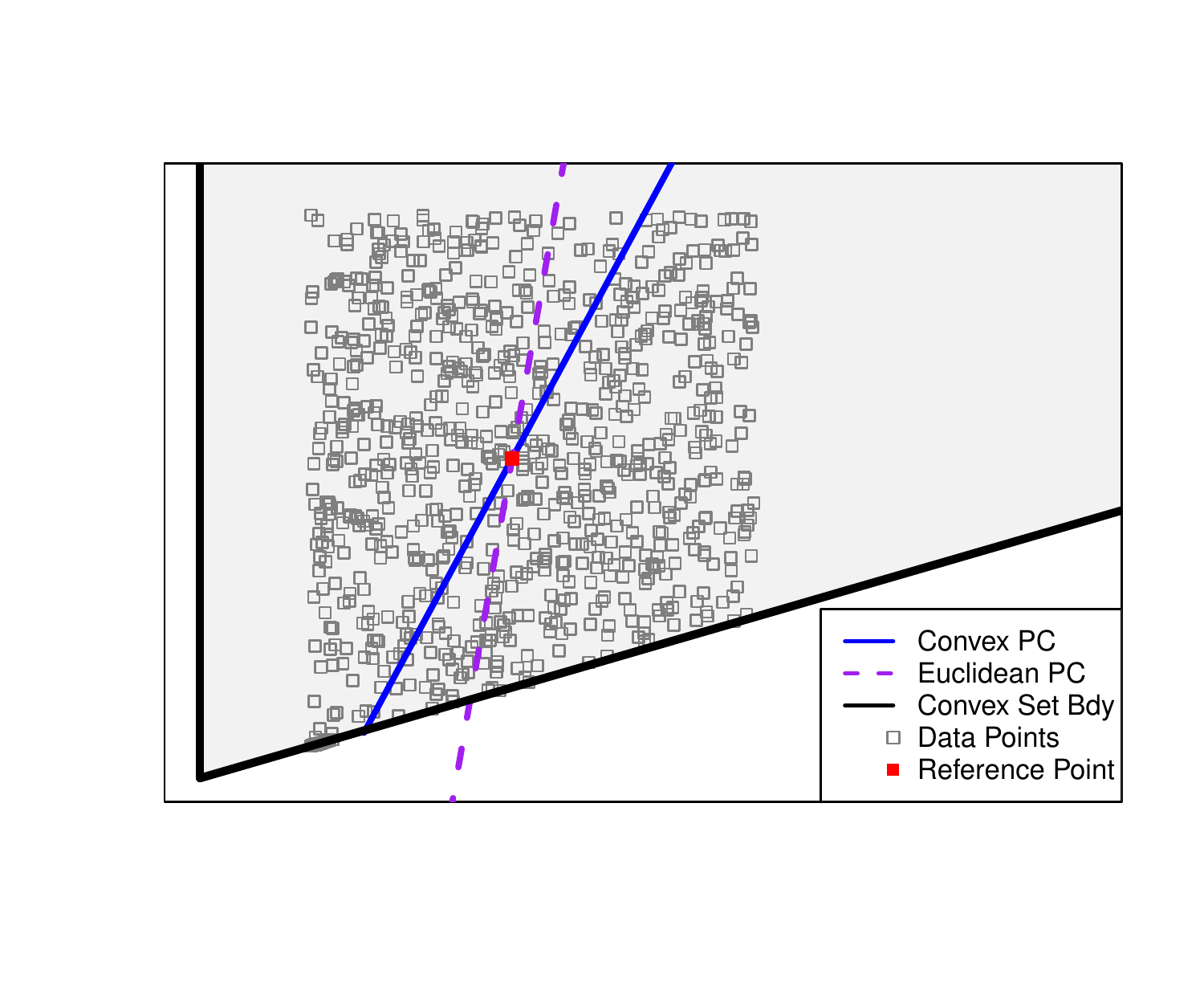}
\end{center}
\vspace{-2cm}
\caption{Euclidean and convex PCA for a simulated two dimensional data-set with values in a cone (shaded area). In Euclidean PCA, the orthogonal projection onto the principal component may lie outside of the convex set. Restricting the projection to the convex set leads to convex PCA.}
\label{fig:Low.Dim.Example}
\end{figure}
We begin by recalling in Section \ref{sec:CPCA.general} the formulation of convex PCA following the treatment of \cite{BGKL17}. Convex PCA has two formulations, namely global and nested. In this paper, we focus on the nested version for interpretability and computational purposes. We also define a notion of explained variation in this context. 

In Section \ref{sec:CPCA.Finite.Dim} we study theoretical properties of CPCA, mostly in the finite dimensional case. We consider (i) consistency in measure, (ii) consistency of finite dimensional approximations of the underlying Hilbert space, and (iii) analytical properties of the objective function \eqref{eqn:Vp} of an equivalent formulation of nested CPCA. These results place CPCA on a rigorous basis and pave the way for numerical implementation. 

Implementation of CPCA is the subject of Section \ref{sec:algorithm}. Our focus is on the case where the Hilbert space is finite dimensional and the convex set is polyhedral. We present an efficient algorithm along with a \CC \ implementation which  utilizes the application programming interface OpenMP for parallel computing.\footnote{See \url{https://www.openmp.org/}. Our implementation is available at: \url{https://github.com/stevenacampbell/ConvexPCA}.}

We study Wasserstein GPCA on an interval in Section \ref{sec:Application.GPCA}. After recalling the isometry between the Wasserstein space and a closed convex subset of an $L^2$ space, we propose a natural finite dimensional approximation of Wasserstein GPCA
, and note that it satisfies the consistency result in Section \ref{subsec:Consistency.Finite.Dim.Approx}. 

We present two financial applications, both of which are motivated by our work \parencite{campbell2022functional} on stochastic portfolio theory (SPT). SPT is a mathematical framework introduced by \cite{F02} for studying macroscopic properties of equity markets and portfolio selection. A key empirical observation of this theory is that the behaviours of stocks have a systematic dependence on their relative ranks (measured by market capitalization). While several stylized stochastic models based on interacting particle systems (see \cite{banner2005atlas},  \cite{itkin2021open} and the references therein) have been introduced to capture this phenomenon, they were mostly studied from the theoretical perspective. On the other hand, CPCA allows us to (i) choose a geometric structure compatible with the data, and (ii) quantify the rank dependence and principal directions to motivate future models. In Section \ref{sec:ranked.returns}, we apply Wasserstein GPCA to distributions of US stock returns ranked by size, and show that the first two convex principal components can be interpreted in terms of volatility and skewness. In Section \ref{sec:rank.based.dist} we consider the US capital distribution curve which captures market capitalization weights arranged in descending order. By viewing them as ranked compositional data under the {\it Aitchison geometry}, we show that the first convex principal component corresponds closely to {\it market diversity}, an entropy-like quantity which measures the concentration of market capital and correlates significantly with the relative performance of active portfolio managers. These properties have significant implications in the construction of macroscopic stochastic models of equity markets as well as portfolio selection.

Finally, in Section \ref{sec:conclusion} we summarize our contributions and discuss possible directions for future research. The Appendix contains technical proofs as well as further details of our numerical implementation, including a comparison with alternative approaches.

\subsection{Related literature} \label{sec:literature}
Statistical analysis of distributional data using optimal transport, and PCA in particular, has attracted a lot of attention in recent years. In \cite{BGKL17} the authors focused on developing the theory of CPCA and Wasserstein GPCA. Numerical implementation of Wasserstein GPCA was considered in \cite{PB22} using a quadratic B-spline approximation of the distribution function and in \cite{cazelles2018geodesic} using a proximal forward-backward splitting algorithm. Here, our approximation is based on discretizing the transport map and is exact when all data points are in the given finite dimensional subspace (we provide a detailed comparison in Appendix \ref{app:comparision}).   Alternatives to Wasserstein GPCA have also been proposed. For instance, \cite{cazelles2018geodesic} studied a ``log'' PCA on the Wasserstein space, and \cite{PB22} introduced projection methods to perform PCA and regression for one-dimensional distributional data. Other recent statistical works include Wasserstein autoregressive models for distributional time series \parencite{zhang2022wasserstein}, Wasserstein PCA for circular measures \parencite{beraha2023wasserstein}, regression models \parencite{chen2021wasserstein} and clustering \parencite{okano2024wasserstein}. Finally, we refer the reader to  \cite{PC19} and \cite{chewi2024statistical} for accessible introductions to optimal transport and its applications in statistics and machine learning.

\section{Convex PCA}\label{sec:CPCA.general}
Let $H$ be a separable real Hilbert space with inner product $\langle \cdot,\cdot \rangle$ and norm $\|\cdot\|$. For $x,y\in H$ and $E\subset H$, let $\mathsf{d}(x,y)=\|x-y\|$ and $\mathsf{d}(x,E)=\inf_{z\in E}\mathsf{d}(x,z)$. Let $X$ be a given nonempty closed convex subset of $H$ which will serve as the state space of the data. In some theoretical results we will further assume that $X$ is compact. If $A$ and $B$ are nonempty subsets of $X$, we let $\mathsf{h}(A, B)$
be the Hausdorff distance between $A$ and $B$ and note that it defines a metric on the space of nonempty and compact subsets of $X$. For $A \subset H$ we let $\dim A$ be its affine dimension. In what follows $k$ is always a non-negative integer.

\begin{definition} \label{def:CL.space}
We endow the following spaces with the Hausdorff distance $h$. 
\begin{enumerate}
\item[(i)] Let $CL(X)$ be the space of nonempty and closed subsets of $X$.
\item[(ii)] Let $CC_k(X)$ be the family of convex sets $C\in CL(X)$ such that $\dim C\leq k$.
\item[(ii)] Given a reference point $\overline{x} \in X$, we let $CC_{\overline{x},k}(X)$ be the family of sets $C\in CC_k(X)$ such that $\overline{x}\in C$.
\end{enumerate}
\end{definition}

By Proposition 3.3 of \cite{BGKL17}, if $X$ is compact then all three spaces in Definition \ref{def:CL.space} are compact. Also, by Proposition A.3 of the same reference, $\mathsf{d}(x, \cdot)$ is continuous on $CL(X)$. If $C \subset H$ is a (nonempty) closed convex set, we let $\Pi_C: X \rightarrow C$ be the metric projection which is well-defined by the Hilbert projection theorem. 

\medskip

We are almost ready to state the two versions of CPCA. Let $\mu$ be a Borel probability measure on $X$. In applications, $\mu$ is usually the empirical measure of the data, say $\mu = \frac{1}{N} \sum_{i = 1}^N \delta_{x_i}$, where $\delta_{x_i}$ is the point mass at the $i$-th data point. Let ${\bf x}$ be an $X$-valued random element whose distribution is $\mu$. We assume $\mu$ has finite second moment, i.e., $\mathbb{E}_{\mu}[\|x\|^2] < \infty$. We also let $\overline{x} \in X$ be a fixed reference point which is typically the mean $\mathbb{E}_{\mu}[{\bf x}]$ of $\mu$.  We stress that this probabilistic formulation is only used to state compactly the CPCA problem. We do not assume that in applications the data points are, say, independent samples from a fixed distribution.

\begin{definition} [Convex PCA] \label{def:CPCA}
Let $\mu$ be a Borel probability measure on $X$ with finite second moment, ${\bf x} \sim \mu$ and $\overline{x} \in X$. For $C \in CL(X)$, define
\begin{equation} \label{eqn:functional.J}
J(C; \mu) = \mathbb{E}_{\mu} [\mathsf{d}^2({\bf x}, C)].
\end{equation}
\begin{itemize}
\item[(i)] (Global CPCA) Given $k \geq 0$, a $(k, \overline{x}, \mu)$-global convex principal component (CPC) is a set $C_k$ such that
\begin{equation} \label{eqn:CPCA.global}
C_k\in\argmin_{C\in CC_{\overline{x},k}(X)} J(C;\mu).
\end{equation}
\item[(ii)] (Nested CPCA) Let $C_0 = \{ \overline{x} \}$. For $k \geq 1$, we say that $C_k$ is a $(k, \overline{x}, \mu)$-nested convex principal component (CPC)  if there exist $C_1, \ldots, C_{k-1}$ such that $C_0 \subset C_1 \subset \cdots \subset C_{k-1} \subset C_k$ and, for each $\ell = 1, \ldots, k$, 
\begin{equation} \label{eqn:CPCA.nested}
C_{\ell}\in\argmin_{C\in CC_{\overline{x},\ell}(X), C \supset C_{\ell-1} } J(C;\mu).
\end{equation}
\end{itemize}
\end{definition}

In the Euclidean setting (i.e., $X = H$), the global and nested PCA problems are equivalent. In CPCA, they are equivalent only when $k = 1$. Although solutions to the nested CPCA problem are generally not optimal in the global sense, they are easier to compute and interpret. Specifically, since the nested CPCs are by construction nested, i.e., $C_k \subset C_{k+1}$, we can quantify the contribution of each additional dimension. Hence, we focus on nested CPCA in later sections. Using the compactness and continuity results cited above, it can be shown (see Theorem 3.1 of \cite{BGKL17}) that when $X$ is compact both the global and nested CPCA problems have non-empty solution sets for any $k$. Also, as long as $\mu$ is not supported on a $(k-1)$-dimensional convex subset of $X$, in \eqref{eqn:CPCA.global} and \eqref{eqn:CPCA.nested} the objective values are unchanged if we restrict to sets of the form
\begin{equation} \label{eqn:C.special.form}
C_{k} = (\overline{x} + \mathrm{span}\{{\bf p}_1, \ldots, {\bf p}_k\}) \cap X,
\end{equation}
where $({\bf p}_1, \ldots, {\bf p}_k)$ is a sequence of orthonormal vectors in $H$ \parencite[Propositions 3.3 and 3.4]{BGKL17}. This representation will be used in the analysis in Sections \ref{sec:CPCA.Finite.Dim} and \ref{sec:algorithm}. In \eqref{eqn:C.special.form}, for an optimal set $\{{\bf p}_1^*,\dots,{\bf p}_j^*\}$, we call ${\bf p}_j^*$ the $j$-th {\it convex principal direction}.\footnote{By an abuse of terminology as in the Euclidean context, we also call it a $j$-th {\it convex principal component}.}

\medskip

A useful concept in PCA is the proportion of explained variation. In the Euclidean setting, it can be expressed in terms of the eigenvalues of the covariance matrix. Since the metric projection $\Pi_C$ is nonlinear when $C$ is not a subspace, eigenvalue decomposition is no longer available. Nevertheless, we may still define a similar notion.

\begin{definition} [Explained variation] \label{def:explained.variation}
The total variation of the data ${\bf x} \sim \mu$ with respect to a given reference point $\overline{x}\in X$ is defined by $
TV = \mathbb{E}_{\mu}\left[\mathsf{d}(\mathbf{x}, \overline{x})^2\right]$. To avoid triviality, we assume $TV > 0$. Let $C \in CC_{\overline{x}, k}(X)$. We define the proportion of variation explained by $C$ by the ratio
\begin{equation} \label{eqn:explaind.variation}
EV(C) = \frac{\mathbb{E}_{\mu}\left[ \mathsf{d}(\Pi_{C} {\bf x}, \overline{x})^2\right]}{TV}.
\end{equation}
\end{definition}

\begin{remark} { \ }
\begin{itemize}
\item[(i)] Clearly $EV(C) \leq 1$ and, if $\mu(C) = 1$, then $EV(C) = 1$. If $C_1 \subset C_2 \subset \cdots$, then $EV_k = EV(C_k)$ is non-decreasing in $k$. We interpret $(EV_k - EV_{k-1}) TV$ as the additional variation explained by the $k$-th component.
\item[(ii)] Definition \ref{def:explained.variation} reduces to the usual definition for Euclidean PCA if we let $X = H$, $\overline{x} = \mathbb{E}_{\mu}[{\bf x}]$ and $C$ be an affine subspace containing $\overline{x}$.
\item[(iii)] Since the definition of $EV(C)$ requires evaluating the projection map $\Pi_{C}$ for individual data points, it is generally costly to compute $EV(C)$ for  large data sets.
\end{itemize}
\end{remark}

\begin{proposition} \label{prop:explained.var.projection}
Let $X$ be compact, and let $C_k$, $k = 1, 2, \ldots$, be an increasing sequence of convex subsets of $X$. If $\overline{\bigcup_k C_k} = X$, then $\Pi_{C_k} \rightarrow \Pi_X$ pointwise on $X$. Consequently, we have $EV(C_k) \rightarrow 1$ as $k \rightarrow \infty$.
\end{proposition}
\begin{proof}
See Appendix \ref{sec:appendix.CPCA.general}.
\end{proof}



\section{Finite dimensional CPCA}\label{sec:CPCA.Finite.Dim}
While infinite dimensional Hilbert spaces arise naturally in theory, in implementations we almost always work with a finite dimensional space due to approximation or discretization. In this section we establish some theoretical properties of CPCA mostly in the finite dimensional case. Sections \ref{sec:CPCA.consistency} and \ref{subsec:Consistency.Finite.Dim.Approx} study consistency properties in measure and in the dimension of the ambient space. In Section \ref{sec:nested.analytical.properties} we consider an equivalent formulation of nested CPCA and prove continuity and differentiability of its objective function.

\subsection{Consistency in measure} \label{sec:CPCA.consistency}
By consistency in measure we mean the stability of the optimization problems \eqref{eqn:CPCA.global} and \eqref{eqn:CPCA.nested} as the input measure $\mu$ varies. For example, under i.i.d.~sampling the empirical distribution converges weakly to the population distribution as the sample size tends to infinity. In this subsection we provide consistency results in the spirit of Section 6 of \cite{BGKL17}, thereby fixing a technical gap in their Lemma 6.1. Using similar techniques, we study in Section \ref{subsec:Consistency.Finite.Dim.Approx} the stability of CPCA when the underlying Hilbert space $H$ is approximated by a sequence of finite dimensional subspaces.



We assume that $X \subset H$ is a compact convex set. Let $\mathcal{P}(X)$ be the space of Borel probability measures on $X$ and equip it with the topology of weak convergence. Since $X$ is compact, weak convergence is equivalent to convergence in the $p$-Wasserstein distance, for any $p \geq 1$ (see Theorem 6.9 of \cite{V08}). For later use, we recall that for $\mu, \nu \in \mathcal{P}(X)$, the {\it $p$-Wasserstein distance} is defined by
\begin{equation} \label{eqn:Wasserstein.distance}
W_p(\mu, \nu) = \left( \inf_{\pi \in \Pi(\mu, \nu)} \int \mathsf{d}(x, y)^p \mathrm{d} \pi(x, y) \right)^{1/p},
\end{equation}
where $\Pi(\mu, \nu)$ is the set of couplings of the pair $(\mu, \nu)$. 


Given a subset $A$ of an ambient space, let $\chi_{A}$ be its {\it indicator function} which is $0$ on $A$ and $+\infty$ on its complement. Then the global CPCA problem \eqref{eqn:CPCA.global} is equivalent to
\begin{equation}\label{eqn:Global.CPCA.Ind.Form}
\inf_{C\in CL(X)}\left\{J(C;\mu)+\chi_{CC_{\overline{x},k}(X)}(C)\right\}.
\end{equation}
Similarly, the nested CPCA problem \eqref{eqn:CPCA.nested} can be written as
\begin{equation}\label{eqn:Nested.CPCA.Ind.Form}
\inf_{C\in CL(X)}\left\{J(C;\mu)+\chi_{\{ C\in CC_{\overline{x},k}(X): \  C\supset C_{k-1}\}}(C)\right\},
\end{equation}
where $C_{k-1}$ is a given $(k-1, \overline{x}, \mu)$-nested CPC. Also, we will be using the notion of {\it $\Gamma$-convergence} (denoted by $\Gamma\textnormal{-}\lim$) whose definition is recalled in Appendix \ref{sec:appendix.external.def.and.results}. We first give a fairly standard but general result which serves as a framework for proving specific consistency statements.

\begin{theorem}[Abstract consistency]\label{thm:general.consistency}
Let $X \subset H$ be compact and convex. Let $\mathfrak{C}_n, \mathfrak{C} \subset CL(X)$ and $\mu_n, \mu \in \mathcal{P}(X)$. Suppose $\Gamma\textnormal{-}\lim_{n \rightarrow \infty} \chi_{\mathfrak{C}_n} = \chi_{\mathfrak{C}}$ and $\mu_n \rightarrow \mu$ in $\mathcal{P}(X)$. Then
\begin{equation} \label{eqn:general.consistency}
\lim_{n\to\infty}\inf_{C\in CL(X)}\{J(C;\mu_n)+\chi_{\mathfrak{C}_n}(C)\}=\inf_{C\in CL(X)}\{J(C;\mu)+\chi_{\mathfrak{C}}(C)\}.
\end{equation}
Moreover, if $C_n\in \argmin_{C\in CL(X)}\{J(C;\mu_n)+\chi_{\mathfrak{C}_n} (C)\}$ for $n\geq1$, then any accumulation point of $(C_n)_{n\geq1}$ (under the Hausdorff distance) attains the infimum on the right hand side of \eqref{eqn:general.consistency}.
\end{theorem}


To apply Theorem \ref{thm:general.consistency}, we need to verify the condition $\Gamma\textnormal{-}\lim_{n \rightarrow \infty} \chi_{\mathfrak{C}_n} = \chi_{\mathfrak{C}}$. When this holds, we have convergence of the objective value as well as convergence of the solution in a suitable sense. The proof of Lemma 6.1 in \cite{BGKL17}, which verifies this condition and underlies the consistency statements there, contains a technical gap which we were able to fix only when $H$ is finite dimensional. Thus, in the remainder of this subsection we assume that $H$ is finite dimensional. The proofs of the following results are given in Appendix \ref{sec.appendix.CPCA.Finite.Dim} which also contains more discussion about the technical issues involved.

\begin{theorem}[Consistency of global CPCA] \label{thm:Global.CPCA.measure.consistency}
Suppose $H$ is finite dimensional and $X$ has non-empty interior. Let $\overline{x} \in \intt X$, $\overline{x}_n \to \overline{x}$ in $X$ and $\mu_n \to \mu$ in $\mathcal{P}(X)$. For $k \geq 1$, let $\mathfrak{C}_n = CC_{\overline{x}_n, k}(X)$ and $\mathfrak{C} = CC_{\overline{x}, k}(X)$. Then $\Gamma\textnormal{-}\lim_{n \rightarrow \infty} \chi_{\mathfrak{C}_n} = \chi_{\mathfrak{C}}$ and hence the conclusions of Theorem \ref{thm:general.consistency} hold for the corresponding global CPCA problems formulated via \eqref{eqn:Global.CPCA.Ind.Form}.
\end{theorem}

We also prove an analogous consistency result for nested CPCA which is not explicitly treated in \cite{BGKL17}. It is formulated in terms of a subsequence because the convex PCs are not necessarily unique. 

\begin{theorem}[Consistency of nested CPCA]\label{thm:nested.CPCA.measure.consistency}
Suppose $H$, $X$, $\overline{x}$ and $k$ are as in Theorem \ref{thm:Global.CPCA.measure.consistency}. Suppose $\overline{x}_n \rightarrow \overline{x}$ in $X$, $\mu_n \rightarrow \mu$ in $\mathcal{P}(X)$ and that $\mathrm{dim}\left(\mathrm{supp}(\mu)\right)\geq k$. For each $n$, suppose $\{\overline{x}_n\} = C_{n, 0} \subset C_{n, 1} \subset \cdots \subset C_{n, k}$ is a sequence such that each $C_{n, j}$ is a $(j, \overline{x}_n, \mu_n)$-nested convex principal component. Then there is a subsequence $n'$ along which the following statements hold:
\begin{itemize}
    \item[(i)] For each $j$ we have $C_{n', j} \rightarrow C_j \in CC_{\overline{x}, j}(X)$.
    \item[(ii)] The sequence $(C_j)_{j = 0}^k$ of sets is increasing in $j$ and forms a sequence of nested convex principal components with respect to $\mu$ and $\overline{x}$.
    \item[(iii)] For each $n'$ and $j$, let $\mathfrak{C}_{n', j} = \{C \in CC_{\overline{x}_{n'}, j}(X): C \supset C_{n', j-1} \}$ and $\mathfrak{C}_{j} = \{C \in CC_{\overline{x}, j}(X): C \supset C_{j-1} \}$. Then $\Gamma\textnormal{-}\lim_{n' \rightarrow \infty} \chi_{\mathfrak{C}_{n', j}} = \chi_{\mathfrak{C}_j}$ and hence the conclusions of Theorem \ref{thm:general.consistency} hold for the corresponding nested CPCA problems formulated via \eqref{eqn:Nested.CPCA.Ind.Form}.
\end{itemize}
\end{theorem}

\subsection{Consistency of finite dimensional approximation}\label{subsec:Consistency.Finite.Dim.Approx}
Convex PCA and Wasserstein GPCA on the real line are generally infinite dimensional optimization problems. To implement them in practice, finite dimensional approximations are necessary. Using similar ideas as in Section \ref{sec:CPCA.consistency}, we formulate such an approximation and show its consistency as the dimension tends to infinity. 

Let $H$ be an infinite dimensional separable Hilbert space and $X \subset H$ be a compact convex set. Let $\{H_n\}_{n \geq 1}$ be an increasing sequence of finite dimensional vector subspaces of $H$, such that $\bigcup_{n = 1}^{\infty} H_n$ is dense in $H$. For example, we may consider an orthonormal basis $\{ \varphi_n \}_{n \geq 1}$ and let $H_n = \mathrm{span}\{ \varphi_1, \ldots, \varphi_n\}$. Let $\Pi_n: H \rightarrow H_n$ be the orthogonal projection onto $H_n$. We impose the simplifying assumption that $\Pi_n X \subset X$ for each $n$. This condition holds in our finite dimensional approximation of Wasserstein GPCA in Section \ref{sec:Application.GPCA}.

Let $\overline{x} \in X$ and $\mu \in \mathcal{P}(X)$. Given $\{H_n\}$, natural finite dimensional approximations of the CPCA problems can be formulated as follows. Let $X_n = \Pi_n(X)$ which is a compact convex subset of $H_n$. Let $\overline{x}_n = \Pi_n \overline{x} \in X_n$ and let $\mu_n = (\Pi_n)_{\#} \mu \in \mathcal{P}(X_n)$ be the {\it pushforward} of $\mu$ under $\Pi_n$. By definition, $\mu_n$ is defined by $\mu_n(A) = \mu(\Pi_n^{-1}(A))$, for $A \subset X_n$ measurable. It is clear that $\mu_n \rightarrow \mu$. Now we may define the finite dimensional approximations by
\begin{equation} \label{eqn:CPCA.finite.dimensional}
\begin{split}
&\inf_{C\in CL(X)}\left\{J(C;\mu_n)+\chi_{CC_{\overline{x}_n,k}(X_n)}(C)\right\}, \quad \text{(global)}\\
&\inf_{C\in CL(X)}\left\{J(C;\mu_n)+\chi_{\{ C\in CC_{\overline{x}_n,k}(X_n): \  C\supset C_{k-1}\}}(C)\right\}, \quad \text{(nested)}
\end{split}
\end{equation}
where $C_{k-1} \in CC_{\overline{x}_n,k}(X_n)$ is given.

We give a consistency result for the global CPCA which is equivalent to nested CPCA when $k = 1$. The proof is given in Appendix \ref{sec.appendix.CPCA.Finite.Dim}. While we expect that analogous results hold for the higher CPCs of the nested problem, the mathematical statements and proofs become much more messy and technical. To focus on the algorithms and applications, we do not pursue consistency further in this paper.

\begin{theorem}[Consistency of finite dimensional approximation for global CPCA]\label{thm:global.consistency.CPCA.approx}
Suppose $X$ is compact and $\Pi_n X \subset X$ for each $n$. For $k$ fixed, let $\mathfrak{C}_n = CC_{\overline{x}_n, k}(X_n)$ and $\mathfrak{C} = CC_{\overline{x}, k}(X)$. Then $\Gamma\textnormal{-}\lim_{n \rightarrow \infty} \chi_{\mathfrak{C}_n} = \chi_{\mathfrak{C}}$ and hence the conclusions of Theorem \ref{thm:general.consistency} hold.
\end{theorem}

\begin{remark}
It is tempting to try to combine the results of Theorem \ref{thm:Global.CPCA.measure.consistency} (consistency in measure of finite dimensional CPCA) and Theorem \ref{thm:global.consistency.CPCA.approx} (consistency as dimension tends to infinity) to obtain consistency in measure for infinite dimensional CPCA. However, this requires letting two parameters tend to infinity simultaneously and is beyond the scope of our current techniques. 
\end{remark}

\subsection{Analytical properties for nested CPCA} \label{sec:nested.analytical.properties}
Now we focus on nested CPCA where $H$ is finite dimensional. We establish continuity and differentiability of quantities involved in the optimization. This will allow us to study, in Section \ref{sec:algorithm}, numerical algorithms to implement CPCA in practice. We work under the following setting:

\begin{condition}[Conditions for finite dimensional nested CPCA] \label{cond:finite.dim.conditions} { \ } 
\begin{itemize}
\item[(i)] $H$ is finite dimensional and so is isomorphic to $\mathbb{R}^d$ for some $d \geq 1$.
\item[(ii)] There exist continuously differentiable convex functions $g_1, \ldots, g_m : H \rightarrow \mathbb{R}$ such that the convex set $X$ has the form
\begin{equation} \label{eqn:X.intersection}
X = \bigcap_{j = 1}^m \{ x \in H: g_j(x) \leq 0 \}.
\end{equation}
\item[(iii)] The distribution $\mu$ has the form $\mu = \frac{1}{N} \sum_{i = 1}^N \delta_{x_i}$.
\item[(iv)] The interior of $X$ is non-empty and the reference point $\overline{x} \in X$ satisfies $g_j(\overline{x}) < 0$ for all $j$. 
\end{itemize}
\end{condition}

In \eqref{eqn:X.intersection}, if each $g_j$ is affine then $X$ is polyhedral. If each $g_j$ is linear then $X$ is a polyhedral cone. Note that here we do {\it not} require $X$ to be compact. 

To formulate the main results, we use the observation (see \eqref{eqn:C.special.form}) that the nested CPCA problem can be solved by constructing a sequence of orthonormal vectors. For ${\bf p} \in H$, we let
\begin{equation} \label{eqn:Vp}
V({\bf p}) = J( (\overline{x} + \mathrm{span}\{{\bf p}\})\cap X; \mu) = \frac{1}{N} \sum_{i = 1}^N \left\{ \inf_{z_i \in (\overline{x} + \mathrm{span}\{{\bf p}\}) \cap X} \|x_i - z_i\|^2 \right\}.
\end{equation}
Define
\begin{equation} \label{eqn:CPCA.sequential.1}
{\bf p}_1^* \in \argmin_{{\bf p} \in H: \|{\bf p}\| = 1} V({\bf p}),
\end{equation}
and, for $k \geq 2$,
\begin{equation}  \label{eqn:CPCA.sequential.2}
{\bf p}_k^* \in \argmin_{{\bf p} \in H: {\bf p} \in \mathrm{span}\{ {\bf p}_1^*, \ldots, {\bf p}_{k-1}^*\}^{\perp}, \|p\| = 1} V({\bf p}).
\end{equation}
By Theorem \ref{thm:cont.pca} below, $V(\cdot)$ is continuous on $H \setminus \{0\}$. This guarantees the existence of ${\bf p}_k^*$ (as long as $\dim H \geq k$). By Proposition 3.4 of \cite{BGKL17}, the sets $C_k = (\overline{x} + \mathrm{span}\{ {\bf p}_1^*, \ldots, {\bf p}_k^*\}) \cap X$ form a sequence of nested convex principal components.


\begin{theorem}[Continuity]\label{thm:cont.pca}
Under Condition \ref{cond:finite.dim.conditions}, the function $V$ defined by \eqref{eqn:Vp} is continuous on $H \setminus \{0\}$. Moreover, each inner minimization problem in \eqref{eqn:Vp} has a unique optimizer $z_i^* = z_i^*({\bf p}) = \overline{x} + a_i^*({\bf p}) {\bf p}$ which is continuous in ${\bf p} \in H \setminus \{0\}$.
\end{theorem}

The proof is given in Appendix \ref{sec:pf.cont.pca}. Using a recent generalized envelope theorem proved in \cite{marimon2021envelope} and recalled in Appendix \ref{sec:envelope}, we also study the differentiability of $V$. 

\begin{theorem}[Differentiability] \label{thm:pca.diff}
Under Condition \ref{cond:finite.dim.conditions}, the directional derivative 
\[
V'({\bf p}; \hat{{\bf p}}) = \lim_{t \rightarrow 0^+} \frac{1}{t} (V({\bf p} + t \hat{{\bf p}}) - V({\bf p}))
\]
of $V$ exists for any ${\bf p} \in H \setminus \{0\}$ and $\hat{{\bf p}} \in H$. Moreover, if ${\bf p} \neq 0$ and for each $i$ at most one constraint is binding (i.e., $g_j(z_i^*) = 0$ for at most one $j$), then $V$ is differentiable at ${\bf p}$. 
\end{theorem}

For the proof see Appendix \ref{app:Thm.pca.dff} where we also provide expressions of the directional directive and gradient of $V$ in terms of a Lagrangian and its saddle points. In particular, if the Hausdorff dimension of $\bigcup_{i \neq j } \{g_i = 0 \} \cap \{g_j = 0\}$ is at most $d - 2$, then $V$ is differentiable almost everywhere on $H \setminus \{0\}$.

\section{Implementation for polyhedral domains}\label{sec:algorithm}
In this section we provide an efficient algorithm to solve nested CPCA when the Hilbert space is finite dimensional and the convex domain is polyhedral. The algorithm, which we implement with R and C++, can be applied to fairly large and high dimensional data-sets.\footnote{Our implementation is available at \url{https://github.com/stevenacampbell/ConvexPCA}.} 

As in Section \ref{sec:CPCA.Finite.Dim} we let the Hilbert space $H$ be finite dimensional, so $H$ is isomorphic to a Euclidean space. To simplify the notations we assume without loss of generality that $H = \mathbb{R}^d$ for some $d \geq 2$ (the case $d = 1$ is trivial). Elements of $\mathbb{R}^d$ are regarded as column vectors. We say that a closed convex set $X \subset \mathbb{R}^d$ is {\it polyhedral} if there exists an $m \times d$ matrix $A$ and $b \in \mathbb{R}^m$ such that 
\begin{equation} \label{eqn:convex.polyhedron}
X = \{x \in \mathbb{R}^d : Ax \geq b\}.
\end{equation}
If $b = 0$, we say that $X$ is a polyhedral convex {\it cone}. We remark that if $X \subset \mathbb{R}^d$ is a compact convex set with non-empty interior, then it can be approximated arbitrarily well by convex polyhedrons under the Hausdorff metric (see \cite{bronshtein1975approximation} and \cite{schneider1981approximation}). 

Let $X$ be as in \eqref{eqn:convex.polyhedron} and consider, for a given reference point $\overline{x} \in X$ and measure $\mu = \frac{1}{N} \sum_{i = 1}^N \delta_{x_i} \in \mathcal{P}_2(X)$, the nested CPCA problem which can be stated in the form (see \eqref{eqn:CPCA.sequential.1} and \eqref{eqn:CPCA.sequential.2})
\begin{equation} \label{eqn:nested.PCA.polyhedral1}
{\bf p}_1^* \in \argmin_{{\bf p} \in H: \|{\bf p}\| = 1} V({\bf p}), \quad
{\bf p}_k^* \in \argmin_{{\bf p} \in H: {\bf p} \in P_{k-1}^{\perp}, \|p\| = 1} V({\bf p}),
\end{equation}
where $P_{k-1} = \mathrm{span}\{ {\bf p}_1^*, \ldots, {\bf p}_{k-1}^*\}$ and 
\[
V({\bf p}) = \frac{1}{N} \sum_{i = 1}^N \left\{ \inf_{z_i \in (\overline{x} + \mathrm{span}\{{\bf p}\}) \cap X} \|x_i - z_i\|^2 \right\}, \]
as given by \eqref{eqn:Vp}. We assume that Condition \ref{cond:finite.dim.conditions} holds.

\begin{figure}[t!]
\centering
\begin{tikzpicture}[scale = 1.25]
\draw[fill=gray!30, gray!30] (0,0) -- (0.5, -3) -- (1.5, -3.5) -- (4, -3.) -- (5.5, -1) -- (5.4, -0.2) -- (0, 0);
\draw[thick, gray!70] (0.2, -1.2) -- (4.75, -2);
\node[circle, draw=black, fill = black, inner sep=0pt, minimum size=2.5pt] at (2.5, -1.6)  {};
\node[left] at (0.2, -1.2) {\tiny $(\overline{x} + \mathrm{span}\{ {\bf p} \}) \cap X$};
\node[below] at (2.5, -1.6) {\tiny $\overline{x}$};

\node[circle, draw=blue, fill = blue, inner sep=0pt, minimum size=2.5pt] at (5.35, -0.8)  {};
\node[right, blue] at (5.35, -0.8) {\tiny $x_i$};
\node[right, blue] at (4.75, -2) {\tiny $z_i^* = \overline{x} + a_i^* {\bf p}$};
\node[circle, draw=blue, fill = blue, inner sep=0pt, minimum size=2.5pt] at (4.75, -2)  {};
\draw[dashed, blue] (5.35, -0.8) -- (4.75, -2);

\draw[dashed, blue] (1.1, -1.355) -- (0.93, -2.355);
\node[circle, draw=blue, fill = blue, inner sep=0pt, minimum size=2.5pt] at (1.1, -1.355)  {};
\node[circle, draw=blue, fill = blue, inner sep=0pt, minimum size=2.5pt] at (0.93, -2.355)  {};
\node[below, blue] at (0.93, -2.355) {\tiny $x_j$};
\node[above, blue] at (1.1, -1.355) {\tiny $z_j^*$};
\node[below, black] at (3.1, -3.3) {\tiny $X = \{x : Ax \geq b\}$};
\end{tikzpicture}
\caption{Nested CPCA \eqref{eqn:nested.PCA.polyhedral1} on a polyhedral domain. Note that $z_i^*$ is not always the orthogonal projection of $x_i$ onto $\overline{x} + \mathrm{span}\{ {\bf p} \}$ because of the constraint $z_i^* \in X$.} \label{fig:polyhedral}
\end{figure}
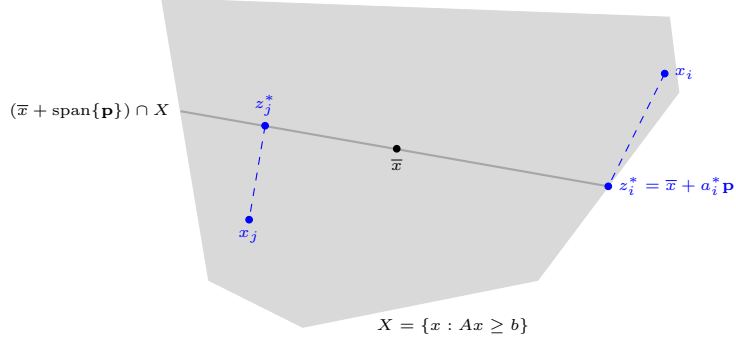

We begin by parameterizing the quantities involved. Write $z_i = \overline{x} + a_i {\bf p}$ where $a_i \in \mathbb{R}$ is chosen such that $z_i \in X$. Suppose ${\bf p}_1^*, \ldots, {\bf p}_{k-1}^*$ have been found and consider the second minimization problem in \eqref{eqn:nested.PCA.polyhedral1}. Given $P_{k-1}$, let $B_k = \begin{bmatrix} {\bf b}_1 & \cdots & {\bf b}_{d-k+1} \end{bmatrix}$ where 
$\{{\bf b}_1, \ldots, {\bf b}_{d-k+1}\}$ is an orthonormal basis of $P_{k-1}^{\perp}$. Then ${\bf p} \in P_{k-1}^{\perp}$ with $\| {\bf p} \| = 1$ if and only if ${\bf p}$ has the form ${\bf p} = B_k {\bf w}$ where ${\bf w}$ is a unit vector in $\mathbb{R}^{d-k+1}$. Thus \eqref{eqn:nested.PCA.polyhedral1} is equivalent to solving
\begin{equation} \label{eqn:nested.PCA.polyhedral2}
{\bf w}_k^* = \argmin_{{\bf w}: \| {\bf w} \| = 1} \tilde{V}({\bf w}) = \argmin_{{\bf w} : \| {\bf w} \| = 1} \left\{ \frac{1}{N} \sum_{i = 1}^N \inf_{a_i \in \mathbb{R}: z_i \in X} \|x_i - (\overline{x} + a_i B_k {\bf w}) \|^2  \right\},
\end{equation}
where the constraint $z_i \in X$ can be written in terms of $a_i$ and ${\bf w}$ as $A (\overline{x} + a_i B_k {\bf w}) \geq {\bf b}$. When $k = 1$, we may simply let $B_1 = I$ and the equivalence still holds. See Figure \ref{fig:polyhedral} for an illustration. Finally, we express the constraint $\|\mathbf{w}\| = 1$ using hyperspherical coordinates. In particular, write ${\bf w} = \boldsymbol{\omega}(\boldsymbol{\theta})$ where $\boldsymbol{\theta} = (\theta_1, \ldots, \theta_{d-k})$ and 

\[
\boldsymbol{\omega}(\boldsymbol{\theta})
=
\begin{bmatrix}
\sin(\theta_1) \\
\sin(\theta_1)\cos(\theta_2) \\
\vdots \\
\sin(\theta_1)\cdots\sin(\theta_{d-k-1})\cos(\theta_{d-k})\\
\sin(\theta_1)\cdots\sin(\theta_{d-k-1})\sin(\theta_{d-k})
\end{bmatrix}.
\]
Now \eqref{eqn:nested.PCA.polyhedral2} is an unconstrained optimization problem in $\boldsymbol{\theta} \in \mathbb{R}^{d-k}$.

We proceed to describe the algorithm, for which more details can be found in Appendix \ref{sec:appendix.alg.considerations}. For each $i$, the inner minimization problem in \eqref{eqn:nested.PCA.polyhedral1} is easy to solve since $L = (\overline{x} + \mathrm{span}\{ {\bf p} \}) \cap X$ can be identified with an interval. Specifically, we first compute the usual orthogonal projection, say $\tilde{z}_i^*$, of $x_i$ onto $L$. If $\tilde{z}_i^*$ is in $X$, then $z_i^* = \tilde{z}_i^*$ is the optimal point. Otherwise, $z_i^*$ is the boundary point of $L$ which is closest to $\tilde{z}_i^*$. This allows us to evaluate the objective $\tilde{V}(\boldsymbol{\omega}(\boldsymbol{\theta}))$. Algorithm \ref{alg:intersection.pt} in the appendix gives the precise procedure for finding the boundary points. For a fixed vector ${\bf p}$, the computational complexity of evaluating $V({\bf p})$ using our algorithm is at most $O(N d)+O(md)$. For a sparse (e.g. bi-diagonal) $A$, this can be sped up to $O(N d)+O(\max\{m,d\})$. If we pass to the unconstrained problem by replacing the input ${\bf p}$ with $B_k\boldsymbol{\omega}(\boldsymbol{\theta})$ we add another (at most) $O(d^2)$ operations. Thus, the dominant complexity will depend on the relative sizes of $N$, $m$, and $d$. In general, the cost of function evaluation is comparable to matrix-vector multiplication of the relevant dimensions.

Using the differentiability of $V$ (and hence $\tilde{V}$) in ${\bf p}$ (and hence in $\boldsymbol{\theta}$), we can use standard gradient approximation  and gradient-based optimization methods to solve \eqref{eqn:nested.PCA.polyhedral1}. As explained in Appendix \ref{sec:appendix.alg.considerations}, part of our implementation  relies on an intelligent initial guess for the solution. If we initialize a standard gradient descent in a neighborhood of the true minimum where $\boldsymbol{\theta}\mapsto \tilde{V}(\boldsymbol{\omega}(\boldsymbol{\theta}))$ is convex, the worst case complexity of a standard gradient descent with a target error of $\epsilon>0$ is $O\left(\max\{m,d,N\}d^2\epsilon^{-1}\right)$. If it is strongly convex, then the contribution $\epsilon^{-1}$ improves to $\log(\epsilon^{-1})$. Theoretically, these tailored rates arising from this custom implementation compare favourably to out-of-the-box interior point solvers which, even for convex problems, have typical computational complexity that is greater than cubic in the number of variables and constraints (see Sections 1.3.1 and 11.5.3 of \cite{BV04}).

\begin{example}
Consider the simulated data set shown in Figure \ref{fig:Low.Dim.Example}. Here $d = 2$ and $X$ is a convex cone indicated by the shaded region. The data is intentionally constrained to lie in a vertical strip so that the first Euclidean PC (dashed line) is almost vertical. Note that for data points near the lower left corner the Euclidean projections are outside $X$. We compute the first convex PC, $C_1$, using our algorithm. As shown in the figure, it is quite different from the Euclidean PC. Now the data points near the lower left corner are projected under $\Pi_{C_1}$ to the endpoint of the convex PC. We also compute the proportion of explained variation (Definition \ref{def:explained.variation}), which is 98\% for the convex PC. This is slightly lower than 99\% which is the variation explained by the Euclidean PC.
\end{example}

\section{Wasserstein GPCA}\label{sec:Application.GPCA}
In this section we consider Wasserstein GPCA which is the original motivation of CPCA in \cite{BGKL17}. We begin by recalling some basic notions of optimal transport and explaining why Wasserstein geodesic PCA is a special case of CPCA. Then we introduce a finite dimensional approximation to Wasserstein GPCA for which the algorithms of Section \ref{sec:algorithm} apply. As application, we analyze return distributions of stocks ordered by size.




\subsection{Wasserstein space on an interval and geodesic PCA} \label{sec:Wasserstein.isometry}
For technical purposes, we let $\Omega = [a, b]$ be a compact interval and leave the unbounded case for future research.\footnote{At a practical level, if the data set consists of compactly supported distributions (say empirical distributions), we can always choose a sufficiently large interval.} In this case $\mathcal{P}_2(\Omega) = \mathcal{P}(\Omega)$ but we will still write the former since we are using the $2$-Wasserstein metric. Let $\mathbb{P}_0 \in \mathcal{P}_2(\Omega)$ be a given reference measure which is absolutely continuous with respect to the Lebesgue measure on $\Omega$ and has full support. A simple but useful example is the uniform distribution on $\Omega$. For any $\mathbb{P} \in \mathcal{P}_2(\Omega)$, there exists a non-decreasing mapping $T_{\mathbb{P}}: \Omega \rightarrow \Omega$, which is unique $\mathbb{P}$-a.e., such that $(T_{\mathbb{P}})_{\#} \mathbb{P}_0 = \mathbb{P}$. Then, it is a well-known result in optimal transport that the Wasserstein distance between $\mathbb{P}, \mathbb{Q} \in \mathcal{P}_2(\Omega)$ is given by
\begin{equation} \label{eqn:interval.Wasserstein.metric}
W_2^2(\mathbb{P}, \mathbb{Q}) = \int_{\Omega} |T_{\mathbb{P}}(x) - T_{\mathbb{Q}}(x)|^2 \mathrm{d} \mathbb{P}_0(x).
\end{equation}
The mapping $\mathbb{P} \mapsto T_{\mathbb{P} }$ defines a bijection from $\mathcal{P}_2(\Omega)$ onto the space
\[
\mathcal{T}_{\mathbb{P}_0}(\Omega) = \{ T : \Omega \rightarrow \Omega: \text{$T$ is $\mathbb{P}_0$-a.e.~non-decreasing}\},
\]
which can be shown to be a compact convex subset of the separable Hilbert space $H = L_{\mathbb{P}_0}^2(\Omega)$. By \eqref{eqn:interval.Wasserstein.metric}, $\mathcal{P}_2(\Omega)$ is isometric to $\mathcal{T}_{\mathbb{P}_0}(\Omega)$ via this mapping. 
Based on this isometry, we {\it define} Wasserstein GPCA on $\mathcal{P}_2(\Omega)$ to be the CPCA on $X = \mathcal{T}_{\mathbb{P}_0}(\Omega) \subset H = L_{\mathbb{P}_0}^2(\Omega)$. Note that a line segment in $X$ corresponds via the isometry to a geodesic in the Wasserstein space $\mathcal{P}_2(\Omega)$ which is given in terms of McCann's {\it displacement interpolation} \parencite{M97}. In particular, the geometry does not depend on the choice of the absolutely continuous reference measure $\mathbb{P}_0$. We also define the {\it barycenter} $\overline{\mathbb{P}}$ of a collection $\mathbb{P}_1, \ldots, \mathbb{P}_N \in \mathcal{P}_2(\Omega)$ by $T_{\overline{\mathbb{P}}} =  \sum_{i = 1}^N \frac{1}{N} T_{\mathbb{P}_i}$. This coincides with the {\it Wasserstein barycenter} introduced in \cite{agueh2011barycenters}.

\subsection{Finite dimensional approximation} \label{sec:Wasserstein.finite.dimensional}
Consider a distributional data-set consisting of distributions $\mathbb{P}_1, \ldots, \mathbb{P}_N \in \mathcal{P}_2(\Omega)$ and a given point $\overline{\mathbb{P}} \in \mathcal{P}_2(\Omega)$ (corresponding to $\overline{x}$ in the general formulation). Let $T_1, \ldots, T_N, \overline{T}$ be the corresponding elements in $X = \mathcal{T}_{\mathbb{P}_0}(\Omega)$ which is the space of transport maps. We introduce a natural finite dimensional approximation of Wasserstein GPCA.

The idea is simple: we approximate the interval $\Omega = [a, b]$ using a refining sequence of partitions. For expositional simplicity and concreteness we focus on the {\it dyadic partition} $\mathcal{D}_n = \{a + (b - a)j 2^{-n}\}_{j = 0}^{2^n}$, but other partitions may be used for computational purposes. For $n \geq 1$, let $H_n \subset H$ be the finite dimensional subspace consisting of functions $T \in H$ such that $T$ is $\mathbb{P}_0$-a.e.~constant on each subinterval corresponding to $\mathcal{D}_n$. Clearly $H = \overline{\bigcup_n H_n}$ by a standard density argument. Let $\Pi_n: H \rightarrow H_n$ be the orthogonal projection onto $H_n$. By projecting $T_i$ and $\overline{T}$ onto $H_n$ (see \eqref{eqn:Wasserstein.L2.projection} below), we get along the lines of Section \ref{subsec:Consistency.Finite.Dim.Approx} a finite dimensional approximation of Wasserstein GPCA. In particular, since $T_i$ (which corresponds to the transport map from $\mathbb{P}_0$ to $\mathbb{P}_i$) is piecewise constant, we approximate each $\mathbb{P}_i$ by a {\it discrete} distribution supported on a set with at most $2^n$ elements. Note that if each data point $\mathbb{P}_i$ (more precisely each $T_i = T_{\mathbb{P}_i}$) is already in $H_n$, then our formulation is exact. The following lemma shows that our approximation satisfies the condition $\Pi_n X \subset X$ and hence the consistency result in Theorem \ref{thm:global.consistency.CPCA.approx} applies.

\begin{lemma}
For $n \geq 1$ we have $\Pi_n X \subset X$.
\end{lemma}
\begin{proof}
Given $n$, consider the orthogonal collection $(\varphi_{n,j} )_{j=1}^{2^n} \subset H_n$ given by $\varphi_{n, j} = 
\mathds{1}_{I_j}$, where $I_j = (a + (b - a)(j-1)2^{-n}, a + (b-a)j2^{-n}]$ is the $j$-th subinterval of $\mathcal{D}_n$, and $\mathds{1}_{I_j}$ is $1$ on $I_j$ and $0$ elsewhere.\footnote{Since $\mathbb{P}_0$ is absolutely continuous, the values on $\mathcal{D}_n$ do not matter.} Then for $T \in H$ we have
\begin{equation} \label{eqn:Wasserstein.L2.projection}
\Pi_nT =\sum_{j=1}^{2^{n}}\frac{\langle T, \varphi_{n, j}\rangle_{L_{\mathbb{P}_0}^2(\Omega)}}{\|\varphi_{n, j} \|_{L_{\mathbb{P}_0}^2(\Omega)}^2} \varphi_{n, j} = \sum_{j = 1}^{2^n} \left( \frac{1}{\mathbb{P}_0(I_j)} \int_{I_j} T \mathrm{d} \mathbb{P}_0 \right) \mathds{1}_{I_j}.
\end{equation}
Clearly, if $T \in X$ then $\Pi_n T$ is non-decreasing and maps $\Omega$ into $\Omega$. So $\Pi_n T \in X$ and $X$ is invariant under $\Pi_n$.
\end{proof}

Each $T \in H_n$ (has a version which) can be written in the form $
T = \sum_{j = 1}^{2^n} t_j \mathds{1}_{I_j}$, for some ${\bf t} = (t_1, \ldots, t_{2^n}) \in \mathbb{R}^{2^n}$. Then $T \in \Pi_n X$ if and only if $a \leq t_0 \leq t_1 \leq \cdots \leq t_{2^n} \leq b$. This allows us to identify $X_n = \Pi_n X$ with a convex {\it polyhedral} set of the form $\{ {\bf t} \in \mathbb{R}^{2^n}: A {\bf t} \geq {\bf b} \}$, so our algorithm developed in Section \ref{sec:algorithm} applies. It can be seen that the matrix $A$ is {\it bidiagonal}; this can be exploited to speed up some matrix computations. 

\begin{remark} Alternative approaches for solving Wasserstein GPCA (and relaxations thereof) have been proposed in \cite{cazelles2018geodesic} and \cite{PB22}. In the former paper, a proximal forward-backward splitting algorithm is employed to solve the GPCA problem on the data histograms. In the latter, a monotone B-spline approximation of the transport maps is used to first reduce the dimensionality of the data. The simplified GPCA problem is then solved directly using an interior point solver for nonlinear programs or approximately via a ``projected PCA''. That is, the unconstrained problem (i.e., over $H$) is solved and the ``projected'' components are taken to be the restrictions of the unconstrained components to (in our notation) $X$. We compare and illustrate these methods using publicly available data in Appendix \ref{app:comparision}.
\end{remark}

\subsection{Ranked stock return distributions} \label{sec:ranked.returns}
In both empirical examples (here and in Section \ref{sec:rank.based.dist}) of this paper we consider the US stock market using data provided by the 
\cite{CRSPdata}. Here, we analyze the daily (arithmetic) returns\footnote{Returns correspond to the \texttt{ret} variable in the CRSP database which is adjusted for stock splits and dividends.} of the largest $N=1000$ stocks by market capitalization from January 1962 to December 2021.

\begin{figure}[t!]
    \centering
    \begin{center}
        \includegraphics[width=0.7\textwidth]{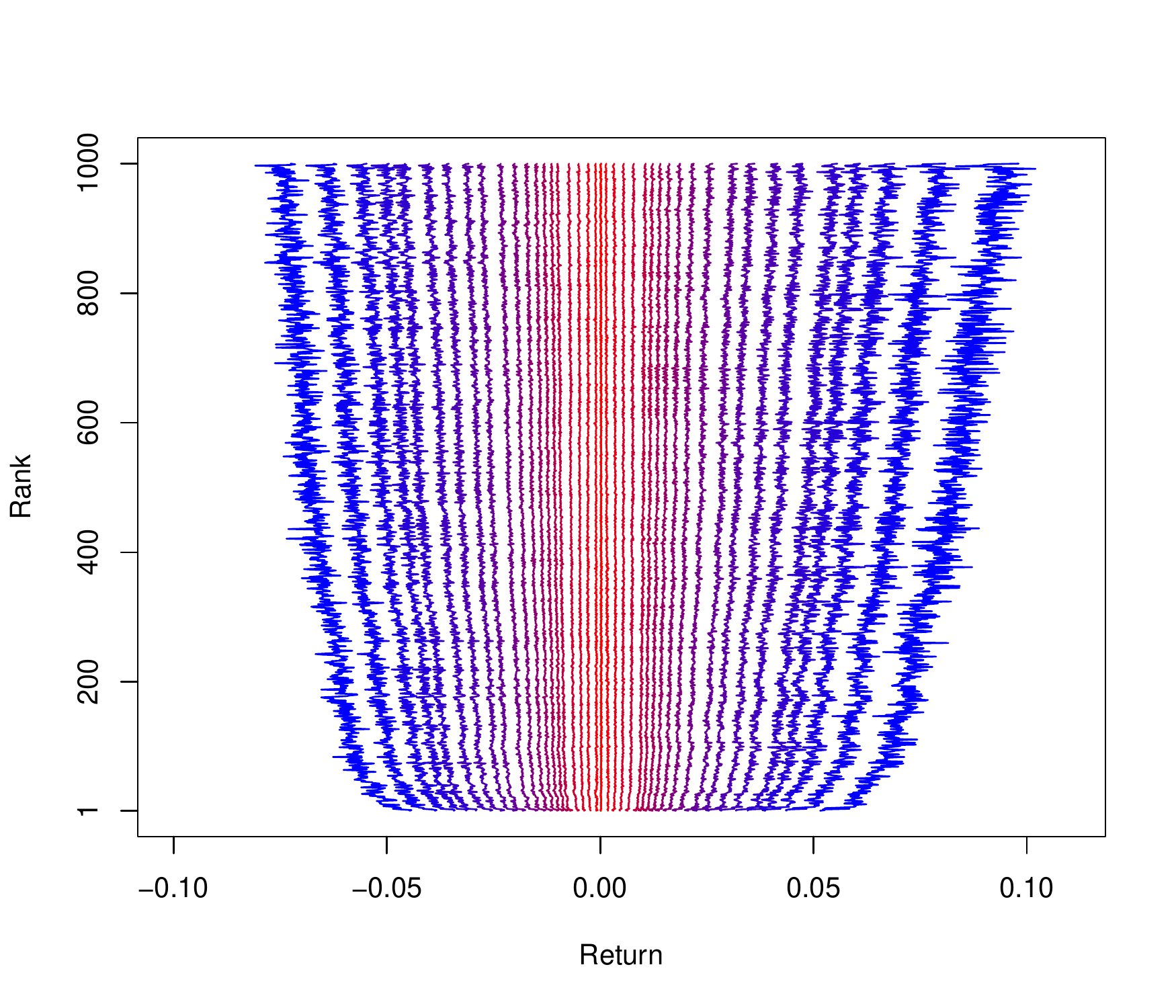}
    \end{center}
    \vspace{-0.5cm}
    \caption{Visualization of the data-set of ranked stock return distributions $\{\mathbb{P}_i\}_{i = 1}^{1000} \subset \mathcal{P}_2(\mathbb{R})$ in Section \ref{sec:ranked.returns}. 
    For several values of $q \in (0, 1)$,\protect\footnotemark \ we show how the $q$-quantile of $\mathbb{P}_i$ varies with the rank $i$. For example, the middle series (corresponding to $q = 0.5$) is mostly vertical, showing that the medians of the return distributions at all ranks are close to $0$. When $q$ is close to $0$ or $1$, the quantiles spread out as $i$ increases, i.e., smaller stocks tend to be more volatile. We also observe that the returns of smaller stocks are more positively skewed.
    }
    \label{fig:Return.Distributions}
\end{figure}
\footnotetext{We show here quantiles for $q$ and $1 - q$ in $\{.5$, $.55$, $.6$, $.65$, $.7$, $.75$, $.775$, $.8$, $.825$, $.875$, $.9$, $.925$, $.94$, $.955$, $.965$, $.975$, $.98$, $.985$, $.99$, $.994\}$. These values are chosen for visualization purposes.} 

At the start of each trading day, we rank the stocks by market capitalization and arrange the returns according to their ranks. This yields, for each rank $i = 1, \ldots, N$, an {\it empirical} return distribution, $\mathbb{P}_i = \frac{1}{M} \sum_{t = 1}^M \delta_{R_{it}} \in \mathcal{P}_2(\mathbb{R})$ where $R_{it}$ is the return of the stock which was at rank $i$ on day $t$, $t = 1, \ldots, M$. In stochastic portfolio theory it was observed that rank-based properties of stocks are useful in portfolio selection. For example, parameters (e.g.~growth rates) of individual stocks are difficult to estimate but the corresponding rank-based quantities are more often much more stable. Understanding of these return distributions is thus relevant in constructing realistic rank-based models of the stock market. The data-set $\{\mathbb{P}_i\}_{i = 1}^{1000}$ is visualized in Figure \ref{fig:Return.Distributions}. Observe that small stocks tend to be more volatile and their return distributions are more positively skewed.

We perform nested Wasserstein GPCA to analyze this data-set. We choose a sufficiently large interval $[a, b]$ which covers all returns, and use a dyadic partition ($\mathcal{D}_n$ with $n = 10$) to obtain a finite dimensional nested CPCA problem as described in Section \ref{sec:Wasserstein.finite.dimensional}. We let the reference distribution $\mathbb{P}_0$ be the uniform distribution on $[a, b]$, and let $\overline{\mathbb{P}}$ be the (Wasserstein) barycenter of the return distributions. In Figure \ref{fig:Results.Wasserstein.GPCA} we visualize the first two convex PCs (principal directions) which we also call the Wasserstein GPCs. In terms of the general notations of CPCA in Section \ref{sec:CPCA.general}, they correspond to the curves $t \mapsto \overline{x} + t {\bf p}_k^* \in X$ for $t \in [-\epsilon, \epsilon]$ and $k = 1, 2$. We also call this the {\it perturbation} of $\overline{x}$ along ${\bf p}_k^*$. As probability distributions, each curve has the form $(\mathrm{Id} + t v_k)_{\#} \overline{\mathbb{P}}$ for some vector field $v_k$ on the interval. Note that each quantile moves at constant velocity as the perturbation parameter $t$ varies. This is a feature of McCann's displacement interpolation, under which each ``particle'' travels along a constant velocity straight line.

\begin{figure}[t!]
    \centering
    \begin{center}
        \includegraphics[width=0.49\textwidth]{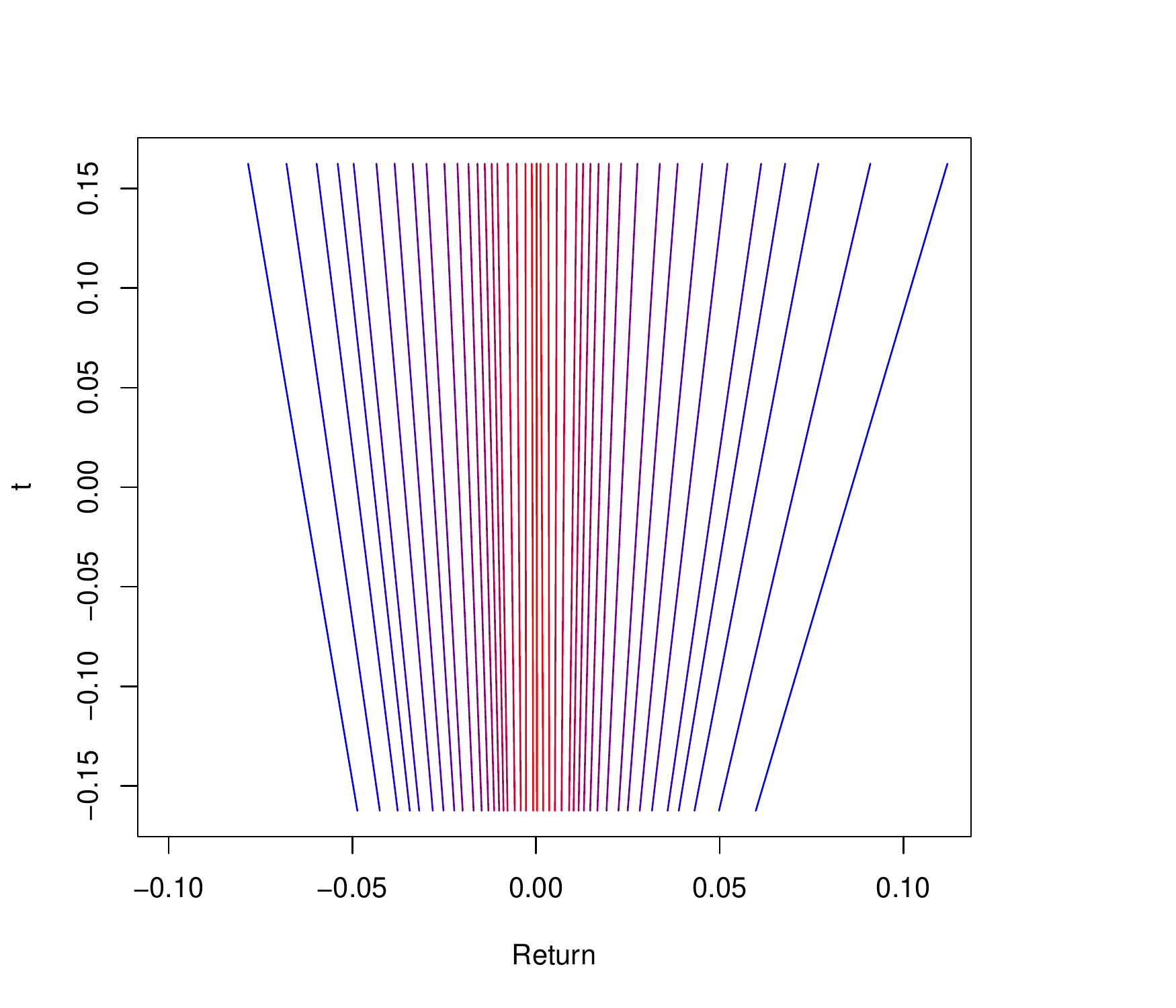}
        \includegraphics[width=0.49\textwidth]{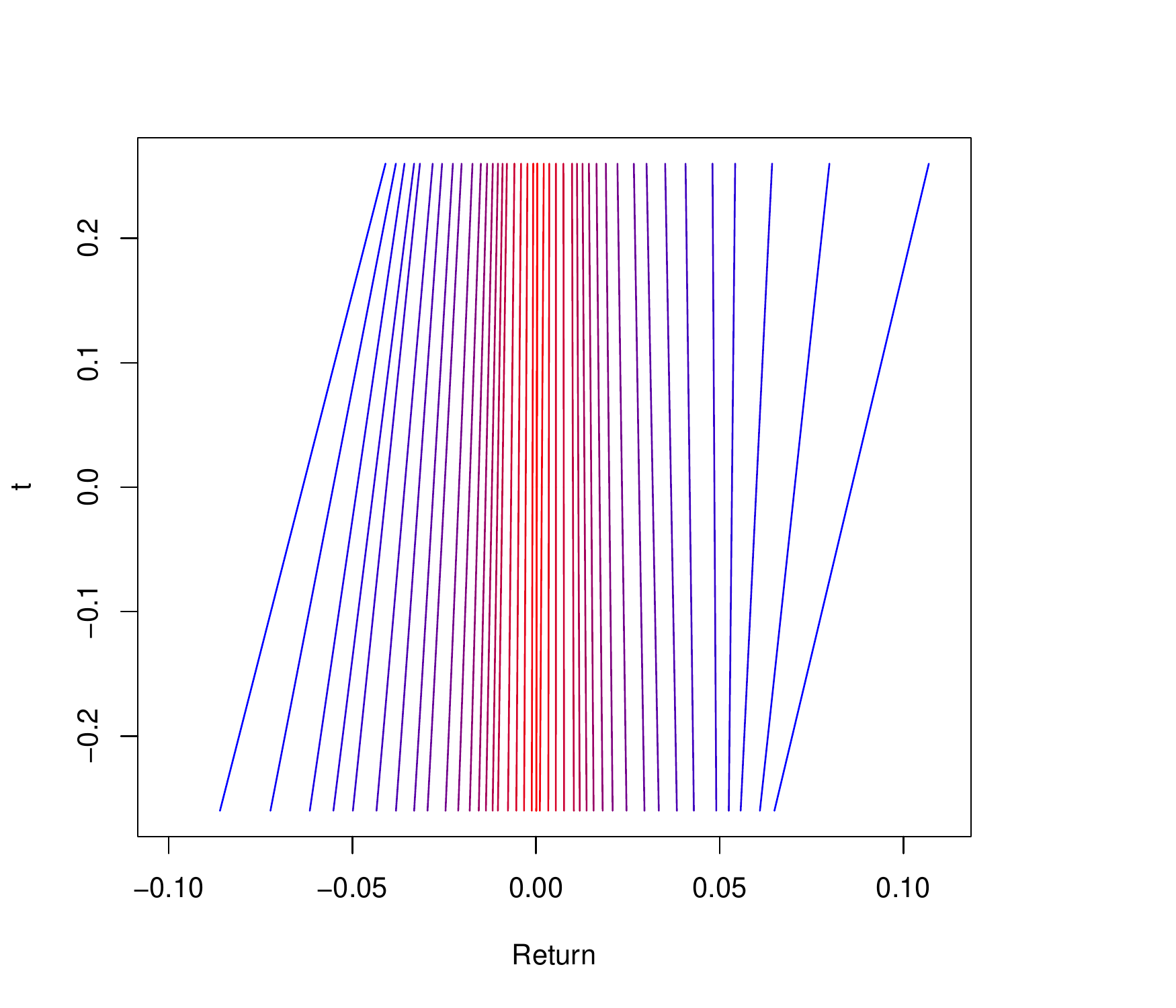}
    \end{center}
    \vspace{-0.5cm}
    \caption{Perturbations $\{(\mathrm{Id} + tv_k)_{\#}\overline{\mathbb{P}}\}_{\{t \in [-\epsilon, \epsilon]\}}$ along the first ($k = 1$, left) and second ($k = 2$, right) Wasserstein GPCs. We show the quantiles (using the $q$-values in Figure \ref{fig:Return.Distributions}) as the perturbation parameter $t$ varies. The first GPC corresponds to a change in volatility and is asymmetric in the two tails. The second GPC acts mostly on the tails. Here $\epsilon > 0$ is chosen to match the variation of the data.
    }
    \label{fig:Results.Wasserstein.GPCA}
\end{figure}

The first principal direction captures the volatility and the overall sknewness of the distribution, and explains alone $89.9\%$ of the variation (in the sense of Definition \ref{def:explained.variation}). The second principal direction can be interpreted in terms of additional skewness effects on the left and right {\it tails} of the distribution (note that the middle quantiles are almost unaffected), and contributes to a small increase of about $3.5\%$ in the explained variation. The projection of the data-set onto $C_2^* = (\overline{x} + \mathrm{span}\{ {\bf p}_1^*, {\bf p}_2^*\})\cap X$ (identified with a subset of $\mathbb{R}^2$) is shown in Figure \ref{fig:Results.Wasserstein.GPCA2} (left). In Figure \ref{fig:Results.Wasserstein.GPCA2} (right) we plot the initial velocity fields $v_1$ and $v_2$ as functions of the return. Consistent with above, we see that $v_1$ is an asymmetric size effect and $v_2$ acts mostly on the tails.

\begin{figure}[t!]
    \centering
    \begin{center}
                \includegraphics[width=0.49\textwidth]{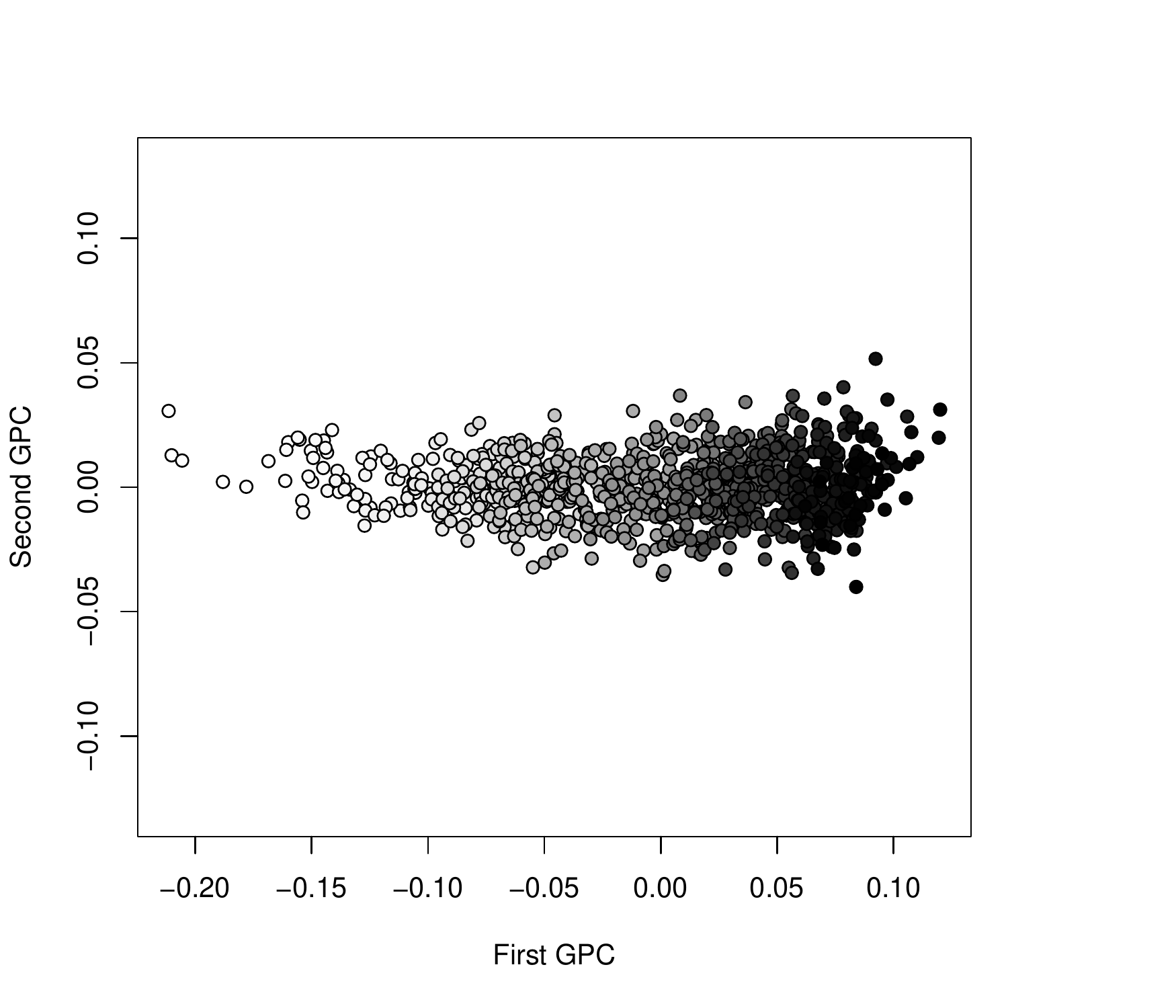}
        \includegraphics[width=0.49\textwidth]{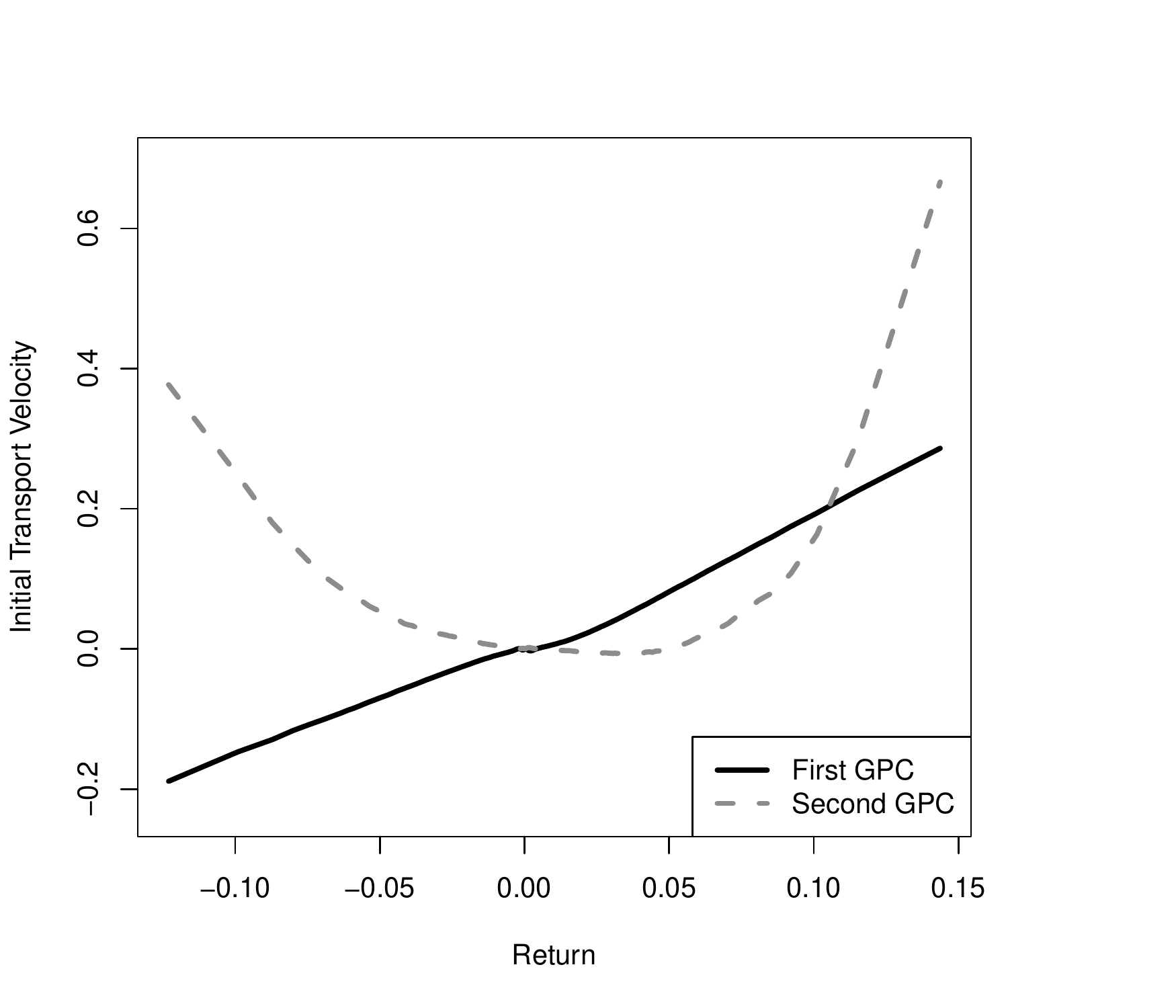}
    \end{center}
    \vspace{-0.5cm}
    \caption{Left: Projection of the ranked return distributions onto the first and second GPCs. Darker dots correspond to smaller stocks. Right: Initial velocity vector fields $v_1, v_2$ corresponding to the first and second GPCs.}
    \label{fig:Results.Wasserstein.GPCA2}
\end{figure}

\section{Capital distribution curve under the Aitchison geometry}\label{sec:rank.based.dist}

\subsection{Financial background}
Our second empirical example is about the {\it capital distribution curve} of the US stock market. We begin by defining this quantity and explaining its financial relevance. Let $m_i(t) > 0$ be the market capitalization of stock $i$ at time $t$. The {\it market weight} of stock $i$ (relative to a given universe with $n$ stocks) is given by
\begin{equation} \label{eqn:market.weight}
w_i(t) = \frac{m_i(t)}{m_1(t) + \cdots m_n(t)}, \quad i = 1, \ldots, n.
\end{equation}
The vector $w(t) = (w_1(t), \ldots, w_n(t))$ takes values in the {\it open unit simplex}
\begin{equation} \label{eqn:open.simplex}
\Delta_n = \{\mathbf{p} = (p_1, \ldots, p_n) \in (0, 1)^n : p_1 + \cdots + p_n = 1\}.
\end{equation}
The {\it capital distribution} at time $t$ is given by the vector of {\it ordered} market weights given by
\begin{equation} \label{eqn:ranked.market.weight}
w_{(\cdot)}(t) = (w_{(1)}(t), \ldots, w_{(n)}(t)),
\end{equation}
where $w_{(1)}(t) \geq \cdots \geq w_{(n)}(t)$ are the (decreasing) order statistics. Note that $w_{(\cdot)}(t)$ takes values in the {\it ordered} unit simplex $\Delta_{n, \geq} = \{\mathbf{p} \in \Delta_n : p_1 \geq \cdots \geq p_n\}$. In Figure \ref{fig:capdist2} (left) we plot the distribution curves (in log-log scale, i.e., $\log k \mapsto \log w_{(k)}(t)$) of the whole US universe, where one curve is drawn for (early January of) each year from 1962 to 2021. A general observation is that the shape of the capital distribution is remarkably stable over time. Nevertheless, fluctuations of the capital distribution curve have significant financial implications. In particular, consider what is called in stochastic portfolio theory the {\it market diversity} which refers to the concentration of market capital. We say that market diversity is low when most market capital is highly concentrated in the largest stocks, and market diversity is high if market capital is distributed more equally among the stocks. One way to measure market diversity is to use the function
\begin{equation} \label{eqn:diversity}
\mathsf{D}_{\lambda}({\bf p}) = \left( \sum_{i = 1}^n p_i^{\lambda} \right)^{1/\lambda}, \quad {\bf p} \in \Delta_n,
\end{equation}
where $\lambda \in (0, 1)$ is a tuning parameter.\footnote{Here we use $\lambda = 0.5$ which is commonly used in the SPT literature. Other symmetric and concave functions on $\Delta_n$ may be used, e.g. the Shannon entropy, and their qualitative behaviours are similar. Such functions are called {\it measures of diversity} in \cite{F02}.} It was shown empirically, see e.g.~\cite{AFG11,F02, FGH98}, that the change in market diversity correlates positively with the performance of active portfolio managers relative to the capitalization-weighted market portfolio. Also see \cite{campbell2022functional} and \cite{ruf2020impact} for recent empirical studies.

\begin{figure} [t!]
\includegraphics[width=0.49\textwidth]{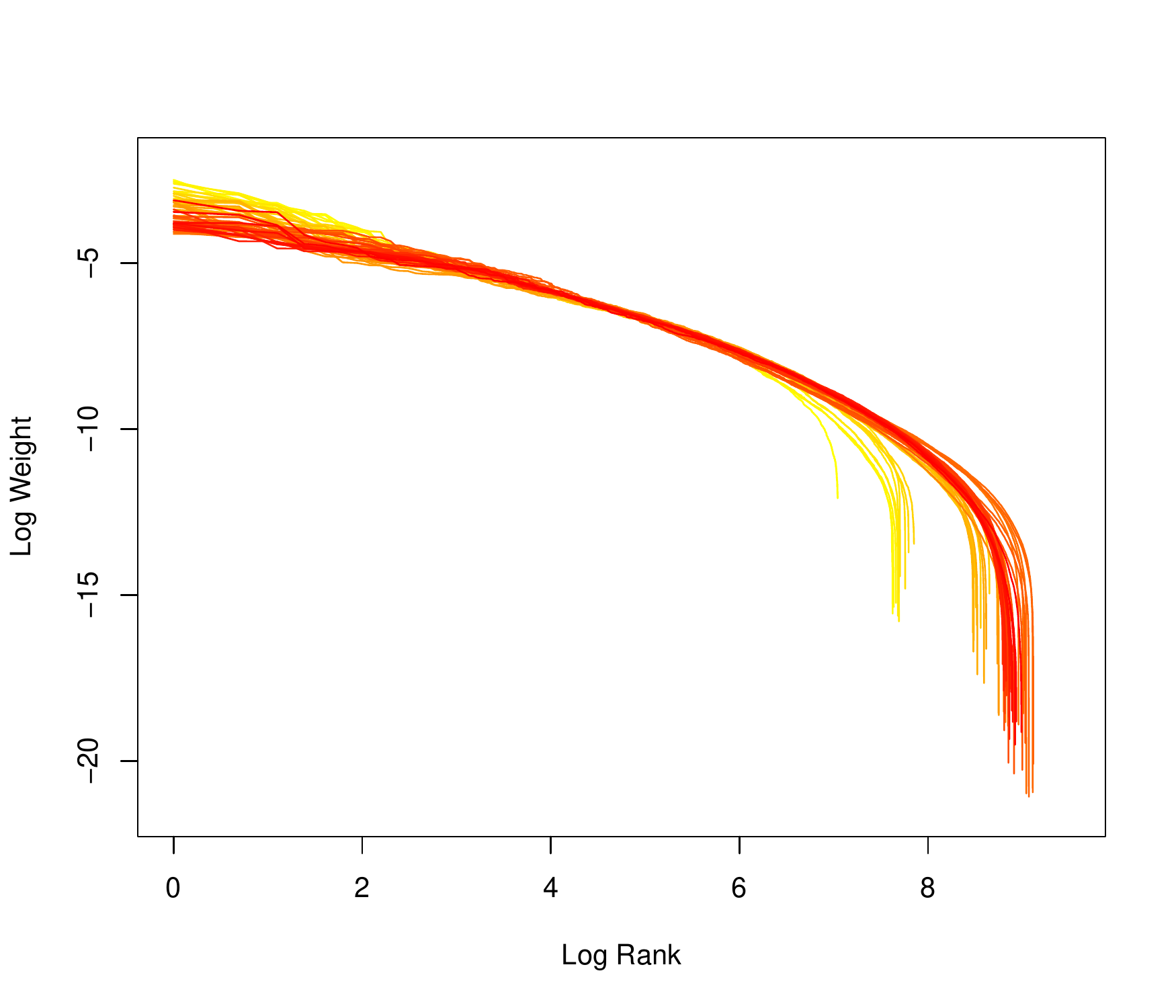}
\includegraphics[width=0.49\textwidth]{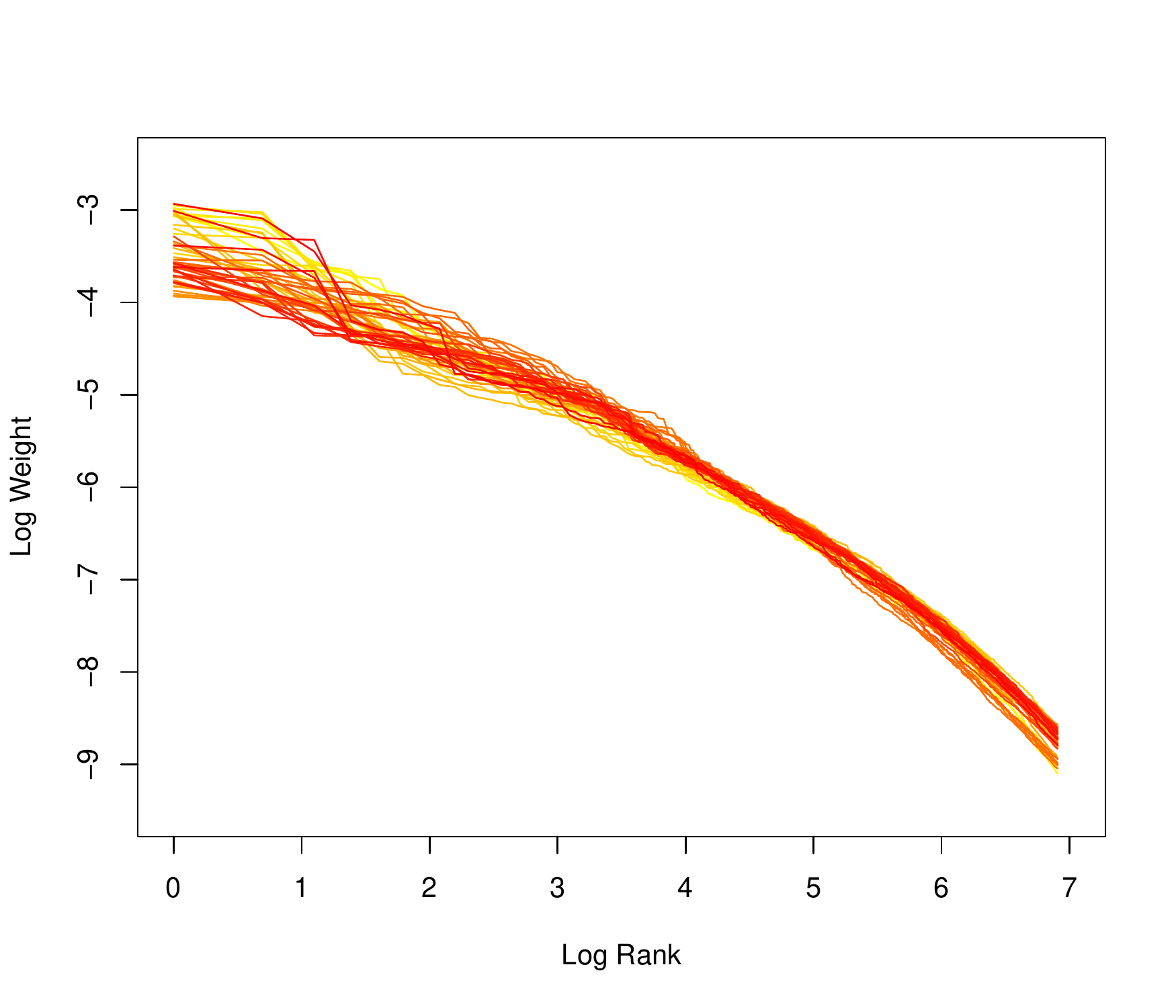}
\caption{Capital distribution curve of the CRSP universe from 1962 to 2021 (left) and restricted capital distribution curve of the largest 1000 stocks from 1973 to 2021 (right). The curves are coloured according to their chronological order with the most recent curves red, and the oldest curves yellow.} \label{fig:capdist2}
\end{figure}

As seen in Figure \ref{fig:capdist2} (left), the total number of stocks varies over time. To keep the dimension fixed, we consider in this section the daily capital distribution curves computed using the largest $n = 1000$ stocks (which vary each day) from January 1973 to December 2021. These curves are plotted in  Figure \ref{fig:capdist2} (right). In December 1972, data from the Nasdaq Stock Exchange was added to the CRSP universe and caused a material change in the capital distribution. Hence for consistency we restrict to dates after this addition. Over the chosen period the top $1000$ stocks generally capture more than $85\%$ of the total market capitalization. While each $w_{(\cdot)}(t)$ may be regarded as a probability distribution on the finite set $\mathcal{X} = \{1, \ldots, n\}$, the set $\mathcal{X}$ has no natural metric structure and so the Wasserstein metric may not be appropriate.\footnote{If we endow $\mathcal{X}$ with the structure of a weighted graph, we may consider optimal transport along the lines of \parencite{M11}. But CPCA cannot be applied, at least directly, to the resulting Wasserstein-like space which can be shown to be Riemannian.} 
The Euclidean geometry (on $\Delta_n$ as a subset of $\mathbb{R}^n$) is inappropriate for this data-set since in financial contexts comparing sizes by their {\it ratios} is often more meaningful than their absolute difference. With this in mind, we endow the simplex $\Delta_n$ with the {\it Aitchison geometry} in {\it compositional data analysis} \parencite{EDP06}, and use it to study fluctuations of the capital distribution curve using convex PCA. The Aitchison geometry, which we review in Section \ref{sec:Aitchison}, respects the relative (rather than absolute) scale of the components in the simplex and clearly distinguishes the boundary of the simplex from the interior. 

\subsection{Aitchison geometry} \label{sec:Aitchison}

For ${\bf x} \in (0, \infty)^n$, define its {\it closure} $\mathcal{C}[{\bf x}] \in \Delta_n$ by $\mathcal{C}[{\bf x}] = \left( \frac{x_1}{\sum_i x_i}, \ldots \frac{x_n}{\sum_i x_i} \right)$. Now define the following operations on $\Delta_n$:
\begin{itemize}
    \item[(i)] ({\it Perturbation}) For ${\bf p}, {\bf q} \in \Delta_n$, define $\mathbf{p} \oplus \mathbf{q} = \mathcal{C} [ (p_1q_1, \ldots, p_nq_n)]$.
    \item[(ii)] ({\it Powering}) For ${\bf p} \in \Delta_n$ and $\alpha \in \mathbb{R}$, define $\alpha \odot p = \mathcal{C} [ (p_1^{\alpha}, \ldots, p_n^{\alpha})]$.
\end{itemize}
Under these operations, $\Delta_n$ is an $(n - 1)$-dimensional real vector space. The identity element is the barycenter $\overline{\mathbf{e}} = \left(\frac{1}{n}, \ldots, \frac{1}{n} \right)$, and the additive inverse of $\mathbf{p}$ is $\mathbf{p}^{-1} = \ominus \mathbf{p} := \mathcal{C}[(p_1^{-1}, \ldots, p_n^{-1})]$. Next define the {\it Aitchison inner product} by
\begin{equation} \label{eqn:simplex.inner.product}
\langle \mathbf{p}, \mathbf{q} \rangle_A = \frac{1}{2n} \sum_{i, j = 1}^n \log \frac{p_i}{p_j} \log \frac{q_i}{q_j}.
\end{equation}	
With this $\Delta_n$ is a Hilbert space. It is isometric with $\mathbb{R}^{n-1}$ via the {\it isometric log ratio transform} given by
\begin{equation} \label{eqn:ilr}
\mathrm{ilr}(\mathbf{p}) = \left(\sqrt{\frac{i}{i+1}}\log\frac{(p_1\dots p_i)^{\frac{1}{i}}}{p_{i+1}}\right)_{1 \leq i \leq n - 1},
\end{equation}
which has the form $\mathrm{ilr}(\mathbf{p}) = (\langle {\bf p}, {\bf e}_i \rangle_A)_{1 \leq i \leq n - 1}$ for a specific orthonormal basis. (The result of the CPCA on the Aitchison simplex is independent of the choice of the basis.)


\begin{lemma} \label{lem:ordered.simplex.cone}
Under the Aitchison geometry, the ordered simplex $\Delta_{n, \geq}$ is a polyhedral convex cone. In particular, we have $\mathrm{ilr}(\Delta_{n, \geq}) = \{x \in \mathbb{R}^{n-1}: Ax \geq 0\}$, where 
\[A = (a_{ij})_{i,j=1}^{n-1}=\left(\sqrt{\frac{j+1}{j}} \delta_{i,j} -\sqrt{\frac{j}{j+1}} \delta_{i,j+1}\right)_{i,j=1}^{n-1},
\]
and $\delta_{i,j}$ is the Kronecker delta.
\end{lemma}
\begin{proof}
Omitted. Note that the matrix $A$ is bidiagonal which is helpful for implementation.
\end{proof}

\subsection{Convex principal components and market diversity}
To perform nested CPCA for the capital distribution curves, we first map the data to $\mathbb{R}^{n-1}$ via the isometric log ratio transform \eqref{eqn:ilr}. After performing CPCA on the resulting polyhedral convex cone (see Lemma \ref{lem:ordered.simplex.cone}), we can interpret the result on the (ordered) simplex by applying the inverse transformation.

\begin{figure}[t!]
    \centering
    \begin{center}
        \includegraphics[width=0.49\textwidth]{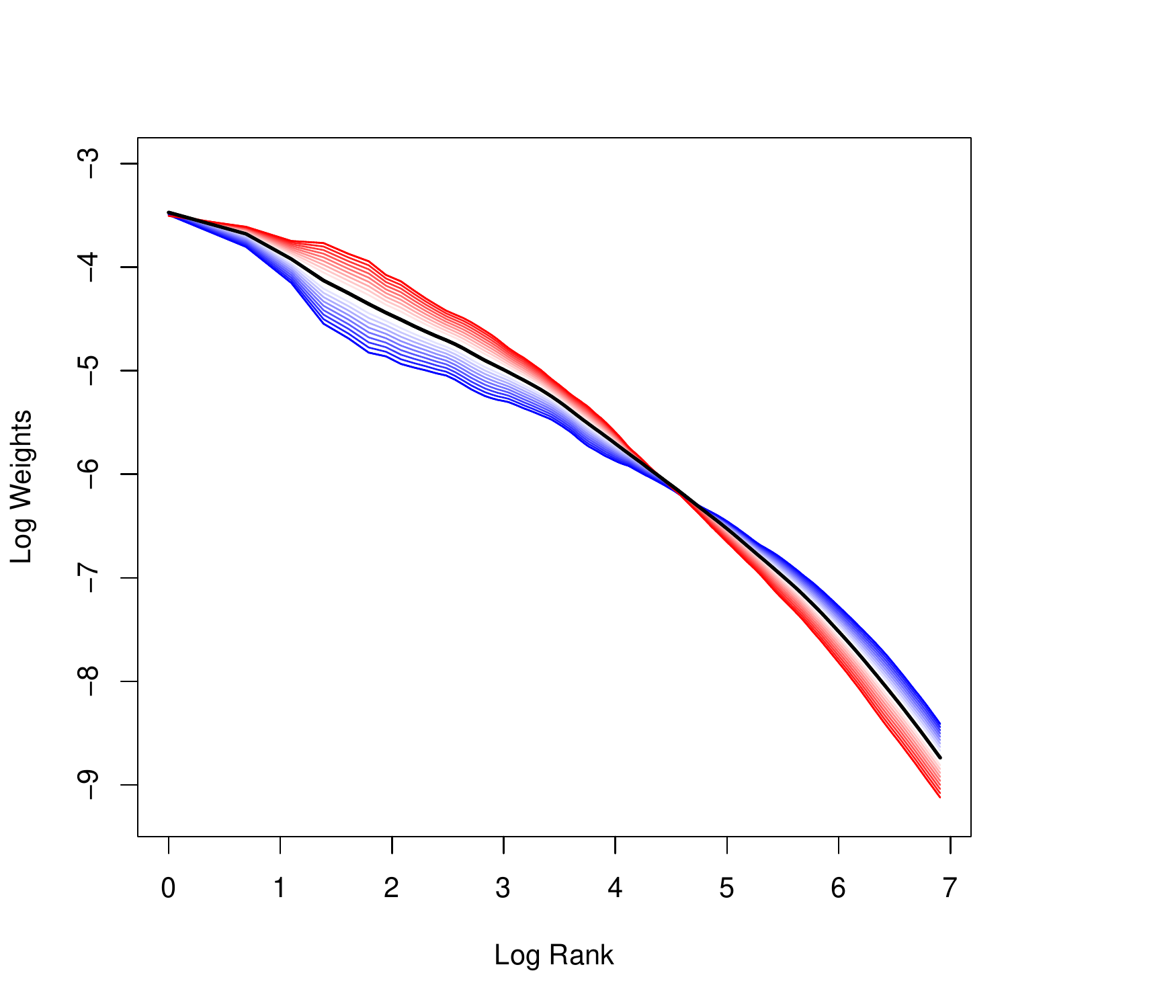}
        \includegraphics[width=0.49\textwidth]{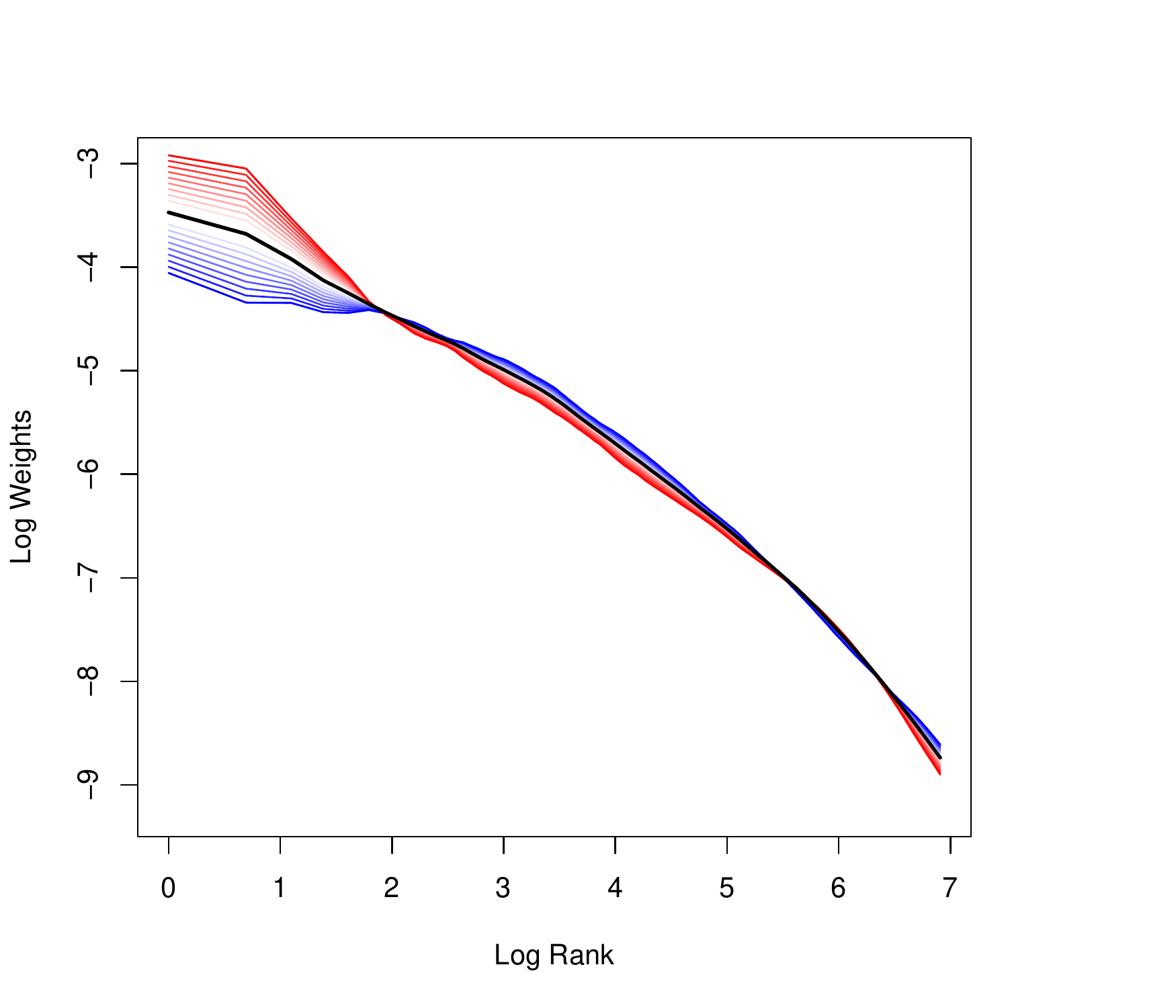}
    \end{center}
    \caption{Perturbation of the average capital distribution curve in the Aitchison geometry with respect to the first (left) and second (right) convex principal directions. The first principal direction corresponds roughly to a tilting of the curve (excluding the largest 2 or 3 stocks) about a point near the center. The second direction corresponds roughly to the relative size of the largest few stocks.}
    \label{fig:First.Two.CPCs.Cap.Dist}
\end{figure}

In Figure \ref{fig:First.Two.CPCs.Cap.Dist} we illustrate the perturbations about the average capital distribution curve (under the Aitchison geometry) along each of the first two convex principal directions\footnote{We note that for the illustrated time period, since the data is concentrated away from the boundary of $\Delta_{n,\geq}$, the convex principal directions do not differ significantly from the standard principal components in the Aitchison geometry.}. The proportion of explained variation is $62.5\%$ for the first principal direction and $78.8\%$ for the first two. Perturbations along the first principal direction correspond to tilting of the curve about a point near the center (in log-log scale) while fixing for the most part the largest few stock. On the other hand, perturbations in the second principal direction may be interpreted in terms of (correlated) fluctuations of the largest few stocks relative to the rest of the universe.

Remarkably, the first principal direction is closely related to the change in market diversity. In Figures \ref{fig:scatter.plots} and \ref{fig:ts.plots} (left) we compare (i) the projected coordinates of the capital distribution curves on the first principal direction and (ii) market diversity measured by $\mathsf{D}_{\lambda}$ (see \eqref{eqn:diversity}) with $\lambda = 0.5$. We observe that the two series are very strongly correlated. This finding gives a \textit{statistical} justification for market diversity and suggests that it is an excellent one-dimensional summary of the capital distribution curve. We also observe that there appears to be a change in the relationship during the COVID health crisis. The relationship in Figure \ref{fig:scatter.plots} seems to experience a parallel shift and the resulting points appear similar to those observed in the 1970s and 1980s. Turning to the second principal component (shown in Figures \ref{fig:scatter.plots} and \ref{fig:ts.plots} (right)), we observe a fairly significant positive correlation between the projected coordinate and the (relative and combined) weight of the largest two stocks. These findings suggest that the capital distribution curve is not as stationary as what Figure \ref{fig:capdist2} and the rank-based diffusion models might suggest.

Taken together, the results of our CPCA suggest that the dominant features of the capital curve are the change in market diversity and the idiosyncratic fluctuations of the largest stocks. In future research, we plan to use these insights to improve existing rank-based models and their calibration.

\begin{figure}
    \centering
    \begin{center}
        \includegraphics[width=0.49\textwidth]{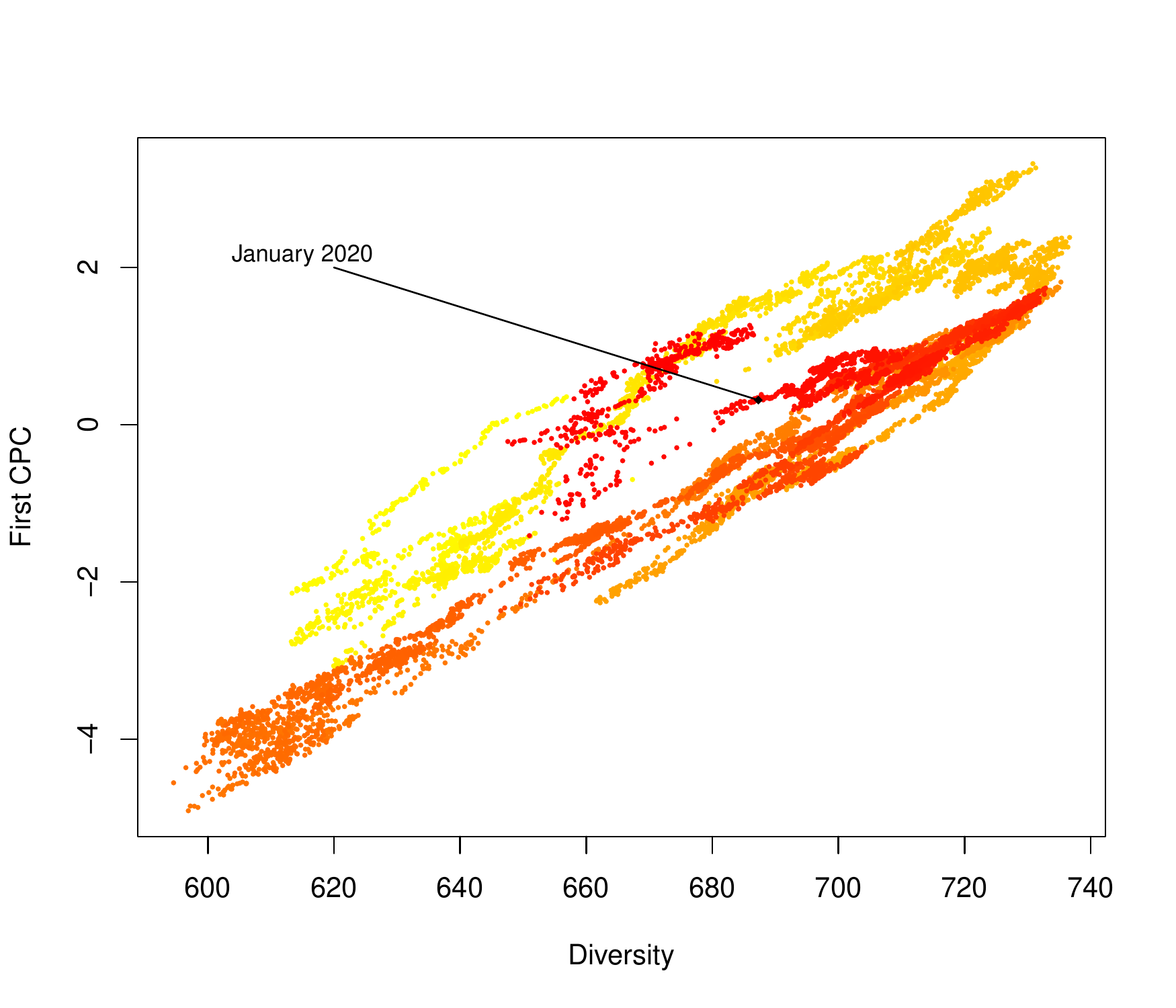}
        \includegraphics[width=0.49\textwidth]{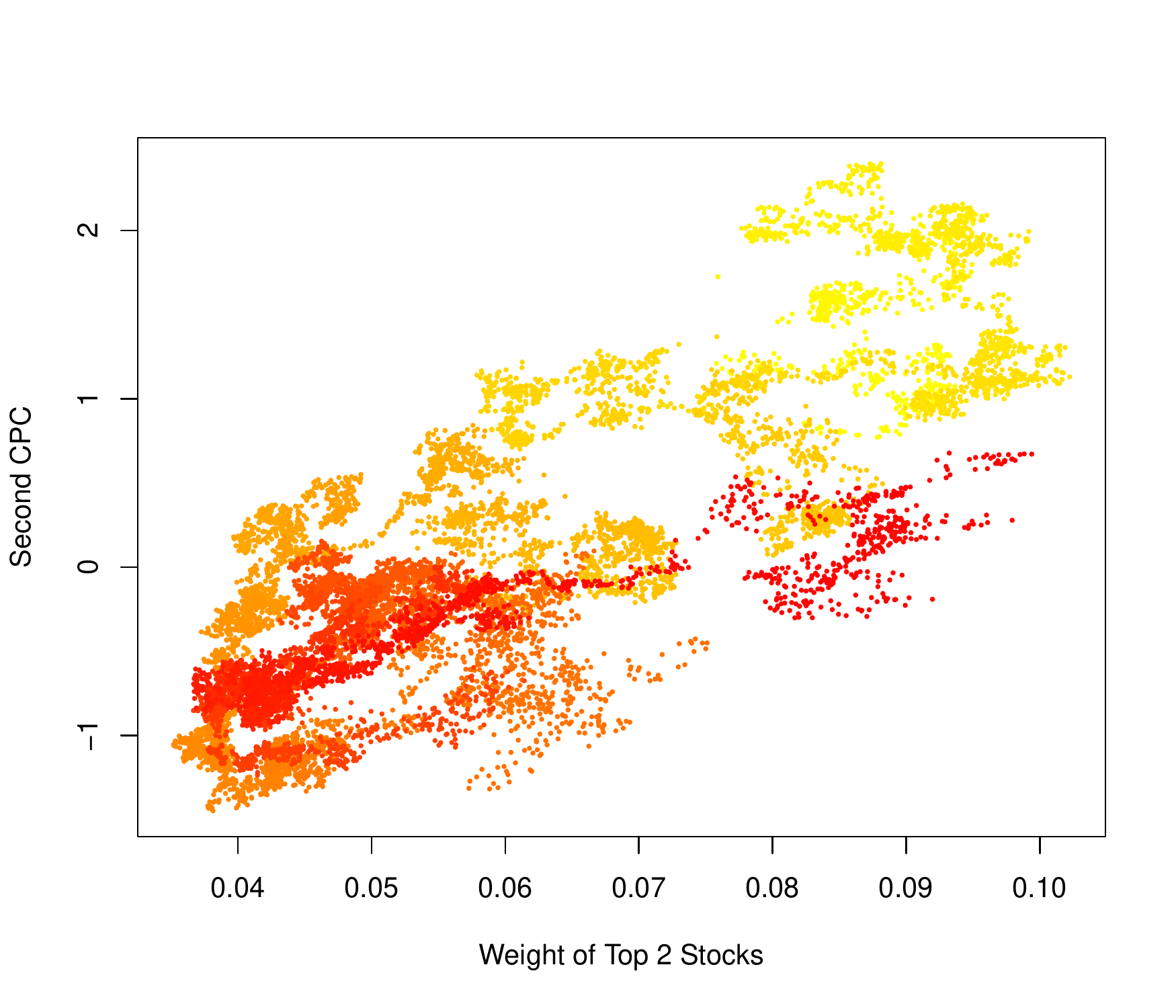}
    \end{center}
    \vspace{-0.5cm}
    \caption{Left: Projected coordinate with respect to the first convex principal direction vs market diversity. Right: Projected coordinate with respect to the second convex principal direction vs the combined market weight of the largest two stocks. The former variables have a strong positive correlation of about $0.91$ and the latter variables have a slightly weaker, but still significant, correlation of $0.8$. In both images, the color gradient illustrates the time of the data point (lighter (1973) -- darker (2021)).}
    
    \label{fig:scatter.plots}
\end{figure}

\begin{figure}[t!]
    \centering
    \begin{center}
        \includegraphics[width=0.49\textwidth]{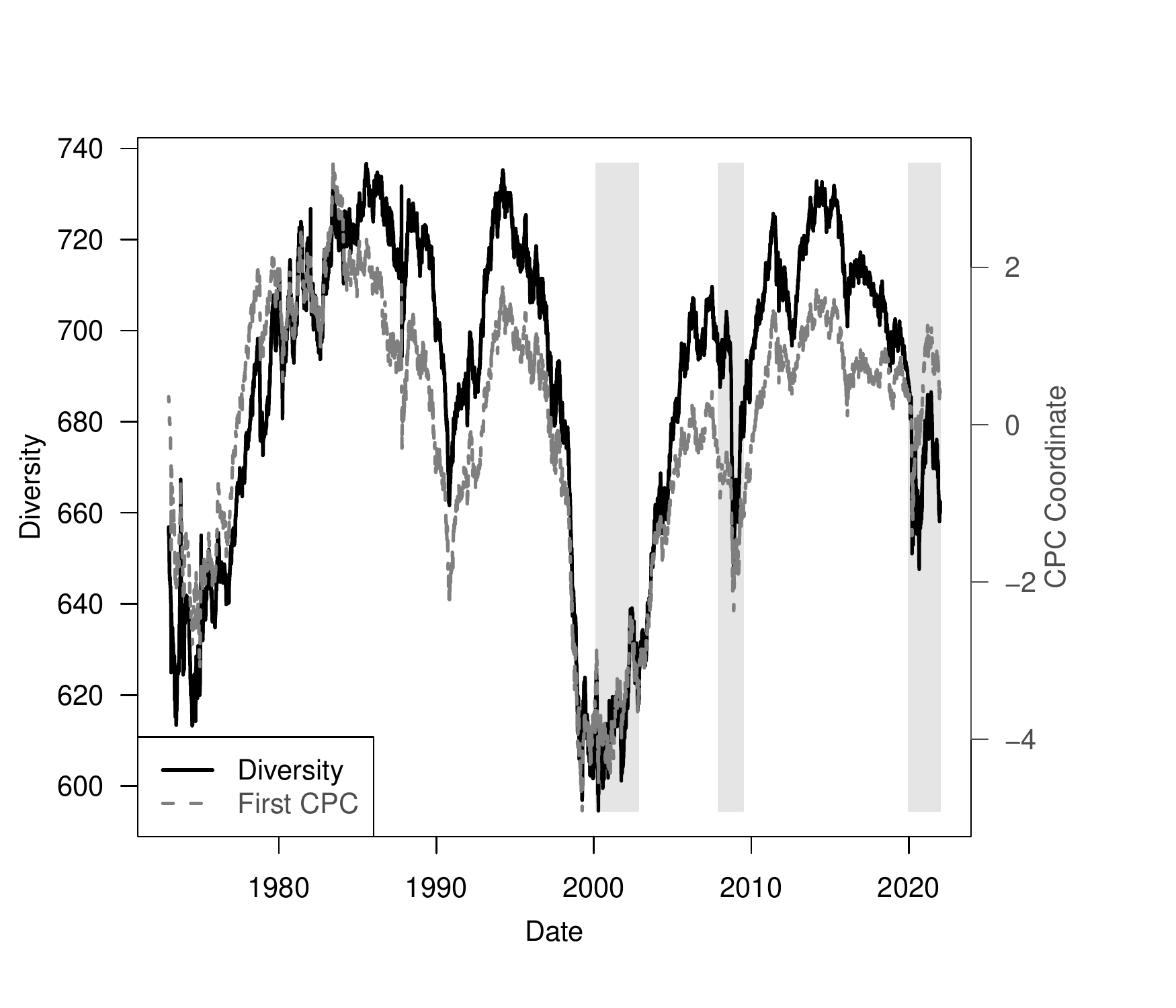}
        \includegraphics[width=0.49\textwidth]{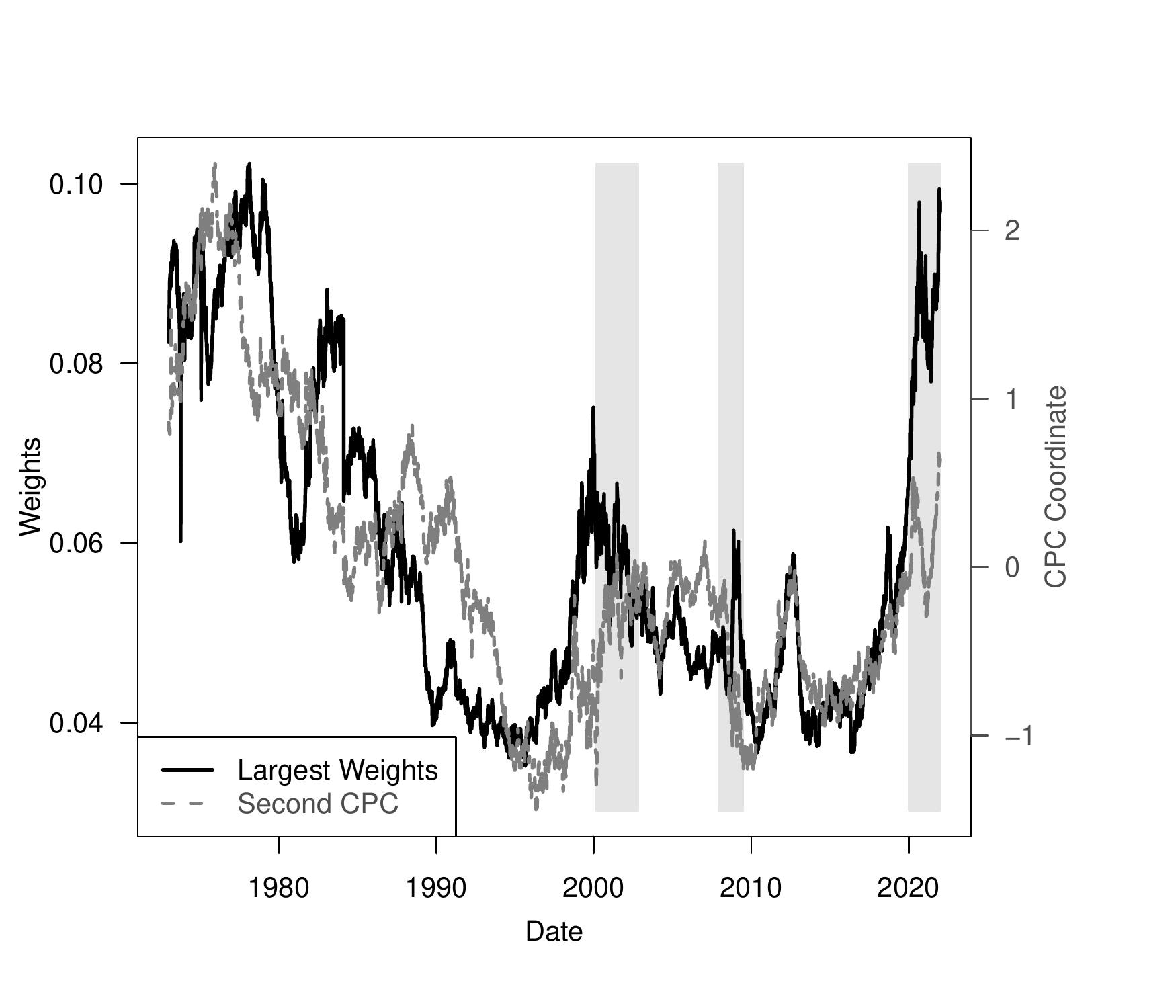}
    \end{center}
    \vspace{-0.5cm}
    \caption{Time series plot of the coordinate on the first convex principal direction with market diversity (left), and the coordinate on the second convex principal direction with the sum of the largest two market weights (right). The vertical grey bars correspond to the dot-com crash, the 2008 recession, and the COVID crisis.}
    \label{fig:ts.plots}
\end{figure}

\section{Conclusion and discussion} \label{sec:conclusion}
In this paper we contribute to both the theory and application of convex PCA. We establish several new theoretical result and develop an efficient algorithm for solving large CPCA problems that is amenable to parallelization and complements the approach of \cite{PB22}. This algorithm can be readily specialized to solve (a finite dimensional approximation of) Wasserstein GPCA on a finite interval. We also provide two interesting financial applications. 

We discuss here several directions for future research.
\begin{itemize}
\item[(i)] While we provide several consistency results in the finite dimensional case, infinite dimensional analogues of Theorems \ref{thm:Global.CPCA.measure.consistency} and \ref{thm:nested.CPCA.measure.consistency} remain open. In particular, we rely crucially on the compactness of the state space $X$ and that it has a nonempty interior. It appears that a different approach is needed for the general case.
\item[(ii)] Existing works on Wasserstein GPCA focus on the one-dimensional setting because the optimal transport maps are simply non-decreasing functions. In general, optimal transport maps are gradients of convex functions by Brenier's theorem, but they are much less tractable when $d \geq 2$. Also, computation of the optimal transport map is much more expensive in multiple dimensions. Another direction of interest is to formulate probabilistic models for Wasserstein GPCA and related algorithms.
\item[(iii)] Our empirical study highlights behaviours of the capital distribution curve and rank-based properties of stocks. Naturally, we hope that these stylized properties are preserved in stochastic models of the equity market. While the statistical properties of individual asset prices have been intensively studied in financial econometrics (see e.g.~\cite{cont2001empirical}), it is not easy to combine them to formulate realistic macroscopic models involving hundreds or thousands of stocks. The rank-based models in stochastic portfolio theory, which are based on systems of interacting It\^{o} processes, fail to capture short-term fluctuations that are relevant in portfolio selection, as they rely on the ergodic properties of Markov processes. On the other hand, dynamic factor models (see e.g.~\cite{bai2016econometric} and the references therein) capture the (possibly dynamic) correlation structure among asset returns but fail to reproduce macroscopic behaviours such as the stability of the capital distribution curve when compounded over time. The construction of realistic macroscopic models, with applications in portfolio management, is an important direction which we hope to address in future research.

\end{itemize}

\section*{Supplementary Materials}

\begin{description}
    \item[Implementation in R and \CC:] Code and README files are available at \url{https://github.com/stevenacampbell/ConvexPCA}. Examples are included that reproduce the analysis of this paper on simulated data. All code files necessary to replicate the numerical experiments of Appendix \ref{app:comparision} are also provided.

    \item[Supporting Appendix:] It provides proofs, algorithmic details and supplementary results.
\end{description}

\section*{Acknowledgment}
This work is partially supported by NSERC Grant RGPIN-2019-04419, a Seed Funding for Methodologists Grant from the Data Sciences Institute (DSI) at the University of Toronto, and an NSERC Alexander Graham Bell Canada Graduate Scholarship (Application No.~CGSD3-535625-2019). The authors report there are no competing interests to declare.

\newpage
\appendix

\counterwithin{figure}{section}
\addcontentsline{toc}{section}{}

\section*{Appendix}

\bigskip

The appendix consists of three parts. Section \ref{sec:app.A} provides proofs of our theoretical results. Section \ref{sec:appendix.alg.considerations} provides details of our algorithm and implementation. We also provide a comparison with alternative approaches. Finally, Section \ref{sec:appendix.external.def.and.results} collects miscellaneous definitions and results that are used in the proofs.

\section{Proofs}
\label{sec:app.A}


\subsection{Proofs from Section \ref{sec:CPCA.general}}\label{sec:appendix.CPCA.general}

\subsubsection{Proof of Proposition \ref{prop:explained.var.projection}}

\begin{lemma}\label{lem:incr.haus.conv}
Suppose $X$ is compact. Let $(C_k)_{k\geq1}$ be an increasing sequence of non-empty compact sets contained in $X$ and $C = \overline{\bigcup_{k}C_k}$. Then $\mathsf{h}(C_k,C)\to 0$ where $\mathsf{h}$ is the Hausdroff metric.
\end{lemma}
\begin{proof}
By \cite[Proposition 3.2]{BGKL17}, $(CL(X),\mathsf{h})$ is a compact metric space. Along a fixed subsequence $k'$ there is a further subsequence $k''$ along which $\mathsf{h}(C_{k''},C^*)\to 0$ for some compact $C^*\in CL(X)$. Since $X$ is compact, convergence under Hausdorff metric is equivalent to Kuratowski convergence (see Lemma \ref{lem:kur.haus}). By the definition of Kuratowski convergence (Definition \ref{def:kuratowski}), for any $x\in C^*$ there exists a sequence $(x_{k''})$ such that $x_{k''} \in C_{k''}$ for each $k''$ and $x_{k''}\to x$. It follows that $C^*\subseteq\overline{\bigcup_{k}C_k}$. On the other hand, it is easy to see that $\bigcup_k C_k \subset C^*$. Thus $C^* = \overline{\bigcup_{k} C_k}$. Since the subsequence $k'$ is arbitrary, we have that $\mathsf{h}(C_k,C)\to 0$.
\end{proof}

\begin{proof}[Proof of Proposition \ref{prop:explained.var.projection}]
By Lemma \ref{lem:incr.haus.conv}, we have $\mathsf{h}(C_k,X)\to 0$. Since $\Pi_{C_k}(x) \in C_k$ by construction, we have
\[
\mathsf{d}(\Pi_{C_k}(x), x) \leq \mathsf{h}(C_k, X) \rightarrow 0, \quad x \in X. 
\]
Hence $\Pi_{C_k} \rightarrow \Pi_X = \mathrm{Id}$ pointwise on $X$. By the bounded convergence theorem, we have
\[
EV(C_k) = \frac{1}{TV} \mathbb{E}_{\mu} \left[ \mathsf{d}(\Pi_{C_k} {\bf x}, \overline{x})^2  \right] \rightarrow \frac{1}{TV} \mathbb{E}_{\mu} \left[ \mathsf{d}(\Pi_X {\bf x}, \overline{x})^2  \right] = 1.
\]
\end{proof}
\subsection{Proofs from Section \ref{sec:CPCA.Finite.Dim}} \label{sec.appendix.CPCA.Finite.Dim}

\subsubsection{Proof of Theorem \ref{thm:general.consistency}} \label{sec:proof.general.consistency} 

\begin{lemma} \label{lem:J.uniform.convergence}
Suppose $X$ is compact. If $\mu_n\to\mu$ in $\mathcal{P}(X)$ then $J(\cdot ; \mu_n) \rightarrow J(\cdot ; \mu)$ uniformly on $CL(X)$.
\end{lemma}
\begin{proof}
Note that $\mathsf{d}(\cdot, C)$ is 1-Lipschitz for any $C\in CL(X)$. Since $X$ is compact, $\mathsf{d}(\cdot, C)$ is uniformly bounded above. Thus, there exists a constant $L > 0$, independent of $C \in CL(X)$, such that $\mathsf{d}^2(\cdot, C)$ is $L$-Lipschitz on $X$. By the Kantorovich-Rubenstein duality (see \cite[(6.3)]{V08}), we have, for any $C\in CL(X)$,
\[
|J(C;\mu_n)-J(C;\mu)| = \left|\mathbb{E}_{\mu_n}[\mathsf{d}^2(\mathbf{x},C)]-\mathbb{E}_{\mu}[\mathsf{d}^2(\mathbf{x},C)]\right|\leq LW_1(\mu_n,\mu). 
\]
This completes the proof since by compactness of $X$ weak convergence is equivalent to convergence in $W_1$.
\end{proof}

\begin{proof}[Proof of Theorem \ref{thm:general.consistency}]
We proceed analogously to \cite[Theorem 6.1]{BGKL17}. As $J(\cdot;\mu)$ is continuous, by \cite[Remark 4.8]{dal2012introduction}, the uniform convergence $J(\cdot;\mu_n)\to J(\cdot;\mu)$ (proved in the lemma above) implies continuous convergence (Definition \ref{def:cont.convergence}). Moreover, by \cite[Remark 4.9]{dal2012introduction}, continuous convergence implies $\Gamma$-convergence. It follows that
\[
\Glim_{n \rightarrow \infty} J(\cdot;\mu_n)=J(\cdot;\mu).  
\]
Now \parencite[Proposition 6.20]{dal2012introduction} gives, by the continuous convergence of $J(\cdot,\mu_n)$, the finiteness of $J(\cdot,\cdot)$ (by the boundedness of $\mathsf{d}(\cdot,\cdot)$), and the $\Gamma$-convergence of the indicators, that
\[
\Glim_{n \rightarrow \infty} \left\{ J(\cdot,\mu_n)+\chi_{\mathfrak{C}_n}(\cdot)\right\}=J(\cdot,\mu)+\chi_{\mathfrak{C}}(\cdot).
\]
In view of the compactness of $CL(X)$ and Proposition \ref{prop:consistency.gamma.conv}, this completes the proof. 
\end{proof}

\subsubsection{Proof of Theorem \ref{thm:Global.CPCA.measure.consistency}}

We begin by noting a technical gap in \cite[Lemma 6.1]{BGKL17}. 

\begin{remark}[Technical gap in \cite{BGKL17}]\label{rmk:technical.gap}
In \cite[Lemma 6.1]{BGKL17}, the authors constructed the translated sets $C_n:=C+x_n-x$ and use them in the proof of $\Gamma$-convergence in their Theorem 6.1. However, these sets may not be contained in $X$. Indeed, for a counterexample it suffices to take $X$ to be the unit square in $\mathbb{R}^2$ and $C$ to be one of its diagonals. Any translation of $C$ moves it, at least partially, outside of $X$. 
\end{remark}

To overcome this technical issue, we assume that $H$ is finite dimensional and that $X$ has a nonempty interior containing the reference element $\overline{x}$. With these additional assumptions we are able to establish Theorem \ref{thm:Global.CPCA.measure.consistency} and Theorem \ref{thm:nested.CPCA.measure.consistency}. 

\medskip

The following results are needed in the proof of Theorem \ref{thm:Global.CPCA.measure.consistency}. We say that sets $A,B\subset H$ are {\it weakly linearly separable} if there exists a non-zero linear functional $f: H \rightarrow \mathbb{R}$ and $c\in\mathbb{R}$ such that $f(a)\geq c\geq f(b)$ for all $a\in A$ and $b\in B$.

\begin{theorem}\label{thm:haus.intersection}
Let $H$ be finite dimensional and $K,L\subset H$ be non-empty, compact and convex sets that are not weakly linearly separable. 
If $K_i,L_i\subset H$ are also non-empty, compact and convex sets with $\mathsf{h}(K_i,K)\to0$, $\mathsf{h}(L_i,L)\to 0$ as $i\to\infty$, then $K_i\cap L_i\not=\emptyset$ for almost all $i$ and $\mathsf{h}(K_i\cap L_i,K\cap L)\to 0$ as $i\to\infty$.
\end{theorem}
\begin{proof}
This is a restatement of \cite[Theorem 1.8.10]{schneider2014convex} after noting that it holds for any finite dimensional real Hilbert space $H$. 
\end{proof}

\begin{lemma}\label{lem:weak.sep.sets.condition}
If $H$ is finite dimensional and $A,B\subset H$ are sets with $\mathrm{int}(A)\not=\emptyset$, $B\subset A$ and $B\cap \mathrm{int}(A)\not=\emptyset$, then $A$ and $B$ are not weakly linearly separable.
\end{lemma}

\begin{proof}
We proceed by contradiction. If $A$ and $B$ are weakly linearly separable, there exits a non-zero linear functional $f$ and $c\in\mathbb{R}$ such that $f(a)\geq c\geq f(b)$ for all $a \in A$ and $b \in B$. As $B\subset A$, we must have $f(b)=c$ for all $b\in B$. Choose $b_0\in B\cap \mathrm{int}(A)$ and an $\epsilon$-ball $B_\epsilon(b_0)\subset \mathrm{int}(A)$. Since $f$ is nonzero, there exists $a \in B_{\epsilon}(b_0)$ such that $f(a) < c$. Thus we get a contradiction.
\end{proof}

\begin{proof}[Proof of Theorem \ref{thm:Global.CPCA.measure.consistency}]
By Lemma \ref{lem:kur.gamma}  it suffices to show that
\[
CC_{\overline{x},k}(X)=\Klim_{n\to\infty}CC_{\overline{x}_n,k}(X).
\]
We prove this by verifying the properties in Definition \ref{def:kuratowski}.

(i) Let $C\in CC_{\overline{x},k}(X)$ be given. For $n \geq 1$, let $C_n'=C+\overline{x}_n-\overline{x} \subset H$. Clearly $\mathsf{h}(C,C_n')\leq \|\overline{x}-\overline{x}_n\|\to 0$. Define $C_n = C_n'\cap X \in CC_{\overline{x}_n,k}(X)$. We claim that $\mathsf{h}(C_n,C)\to 0$. Since $\overline{x}\in \mathrm{int}(X)$ we have $C\cap \mathrm{int}(X)\not=\emptyset$ as $\overline{x}\in C$. By Lemma \ref{lem:weak.sep.sets.condition}, $C$ cannot be (weakly) linearly separated from $X$. Taking $L_n=L=X$ in Theorem \ref{thm:haus.intersection} gives that $\mathsf{h}(C_n,C)=\mathsf{h}(C'_n\cap X,C\cap X)\to 0$.

(ii) The proof of this part is the same as that of  \cite[Lemma 6.1]{BGKL17}. We include it for completeness. Let $C$ be an accumulation point of $(C_n)_{n\geq1}$. As $\overline{x}_n\in C_n$ and $\overline{x}_n\to \overline{x}$ it follows that $\overline{x}\in C$ (as Kuratowski convergence is implied by Lemma \ref{lem:kur.haus}). Since $CC_k(X)$ is compact and $C_n\in CC_k(X)$, we have $C\in CC_k(X)$ which allows us to conclude that $C\in CC_{\overline{x},k}(X)$ as desired.
\end{proof}

\subsubsection{Proof of Theorem \ref{thm:nested.CPCA.measure.consistency}}
The main argument (given after Lemma \ref{lem:nested.consistency.inductive.step}) employs an induction on $j$. The following lemmas are used for the induction step. In particular, the key ingredient is Lemma \ref{lem:nested.consistency.inductive.step} corresponding to part (iii) of the theorem. 

We denote by $\mathrm{conv}(\ldots)$ the convex hull (of a set).

\begin{lemma}\label{lem:conv.hull.properties} Consider a Banach space $(\mathcal{X},\|\cdot\|)$ over $\mathbb{R}$.
\begin{enumerate}
    \item[(i)] If $A,B\subset \mathcal{X}$ are compact convex sets then so is $\mathrm{conv}(A\cup B)$.
    \item[(ii)] Let $A, B, A_n,B_n\subset \mathcal{X}$, $n\geq 1$, be compact convex sets and suppose there exists a compact set $C$ such that
    \[
    A,B,A_n,B_n,\mathrm{conv}(A_n\cup B_n),\mathrm{conv}(A\cup B)\subset    C,
    \]
    for all $n\geq1$. If $A_n\to A$ and $B_n\to B$ in Hausdorff distance then $\mathrm{conv}(A_n\cup B_n)\to \mathrm{conv}(A\cup B)$ in Hausdorff distance.
\end{enumerate}
\end{lemma}
\begin{proof}
(i) Clearly $A\cup B$ is compact. By \cite[Theorem 5.35]{guide2006infinite}, $\overline{\mathrm{conv}}(A\cup B)$ is compact. Clearly $\mathrm{conv}(A \cup A)\subset \overline{\mathrm{conv}}(A\cup B)$. To prove the reverse inclusion, let $x\in\overline{\mathrm{conv}}(A\cup B)$ be given. Then there exists a sequence $x_n\in\mathrm{conv}(A\cup B)$ such that $x_n\to x$. For each $n$, there exist $w_n\in[0,1]$ and $a_n\in A$, $b_n\in B$ such that $x_n=w_na_n+(1-w_n)b_n$. By compactness, there exists a subsequence $n'$ along which $w_{n'} \rightarrow w \in [0, 1]$, $a_{n'} \rightarrow a \in A$ and $b_{n'}\rightarrow b \in B$. Thus $x=wa+(1-w)b \in \conv(A \cup B)$.

(ii) First we note that $A_n\cup B_n\to A\cup B$ in Hausdorff distance. Since $C$ is compact, by Lemma \ref{lem:kur.haus} it suffices to show that $\Klim_{n \rightarrow \infty} \conv(A_n \cup B_n) = \conv(A \cup B)$. We now check the relevant conditions. (i) If $x\in \mathrm{conv}(A\cup B)$ then $x=wa+(1-w)b$ for some $w\in [0,1]$, $a\in A$ and $b\in B$. By Hausdorff convergence there exists $a_n\in A_n$, $b_n\in B_n$ such that $a_n\to a$ and $b_n\to b$. Let $x_n=wa_n+(1-w)b_n\in \mathrm{conv}(A_n\cup B_n)$. Then $x_n\to x$. (ii) Suppose $x_n\in\mathrm{conv}(A_n\cup B_n)$ for $n\geq1$ and $x$ is an accumulation point of $(x_n)_{n\geq1}$. Then there exists a subsequence such that $x_{n'}\to x$. Write $x_{n'}=w_{n'}a_{n'}+(1-w_{n'})b_{n'}$ where $w_{n'}\in [0,1]$, $a_{n'}\in A$ and $b_{n'}\in B$. Again by compactness there exists a further subsequence such that $w_{n''} \rightarrow w \in [0, 1]$, $a_{n''} \rightarrow a$ and $b_{n''} \rightarrow b$. So
\[
x = \lim_{n'' \rightarrow \infty} x_{n''} = \lim_{n'' \rightarrow \infty} (w_{n''}a_{n''}+(1-w_{n''})b_{n''}) = wa + (1 - w)b \in \conv(A \cup B).
\]
\end{proof}


\begin{lemma}\label{lem:full.dim.cpc}
Suppose $\mathrm{dim}(\mathrm{supp}(\mu)) \geq k$. Let $j \leq k$. If $C_j^*$ is either (i) a global or a (ii) nested $(j, \overline{x}, \mu)$-convex principal component, then $\dim C_j^* = j$.

\end{lemma}

\begin{proof}
We will prove this result by contradiction. We will treat (i) as the result for (ii) is analogous. Suppose that $\mathrm{dim}(C_j^*)< j \leq k$. Then, there must exist an $x\in\mathrm{supp}(\mu)$ such that $x\not\in \mathrm{aff}(C_j^*)$, where $\mathrm{aff}(C_j^*)$ is the affine hull of $C_j^*$. Additionally, we can take an open ball $B_\epsilon(x)$ of sufficiently small radius $\epsilon>0$ about $x$ such that $B_\epsilon(x)\cap C_j^*\not=0$ and $\mathrm{dist}(B_\epsilon(x),C_j^*)>2\epsilon$. As $x\in\mathrm{supp}(\mu)$ it follows that $\mu(B_\epsilon(x))>0$. Now consider $\tilde{C}=\mathrm{conv}(C^*_j\cup \{x\})$. It is clear that $\mathrm{dim}(\tilde{C})=\mathrm{dim}(C_j^*)+1\leq j$. Moreover, $\overline{x}\in\tilde{C}$, $\tilde{C}\subset X$ (as $\mathrm{supp}(\mu)\subset X$ by assumption), $\tilde{C}$ is convex, and $\tilde{C}$ is compact (see Lemma \ref{lem:conv.hull.properties}). Hence $\tilde{C}\in CC_{\overline{x},j}(X)$ and thus
\[ J(C_j^*,\mu)\leq J(\tilde{C},\mu)\]
by definition of $C_j^*$. However, for any $y\in B_\epsilon(x)$
\[
\mathsf{d}(y,\tilde{C})=\|\Pi_{\tilde{C}}y-y\|\leq \|\Pi_{\tilde{C}}y-x\|+\|x-y\|\leq 2\epsilon<\mathrm{dist}(B_\epsilon(x),C^*_j)\leq \mathsf{d}(y,C^*_j).
\]
Note here we have used that $\|\Pi_{\tilde{C}}y-x\|\leq \|y-x\|$ as $x\in \tilde{C}$. Also, as $C_j^*\subset \tilde{C}$ for any $z\in H$,  $\mathsf{d}(z,C_j^*)\geq \mathsf{d}(z,\tilde{C})$. We conclude that
\[
J(C_j^*;\mu)=\mathbb{E}_\mu[\mathsf{d}(\mathbf{x},C_j^*)^2]>\mathbb{E}_\mu[\mathsf{d}(\mathbf{x},\tilde{C})^2]=J(\tilde{C};\mu),
\]
which contradicts of the optimality of $C_j^*$.
\end{proof}

We introduce a linear transformation which will be useful in our argument. 

\begin{definition}\label{def:key.lin.trans}
Let $K$ be a compact and convex set in $H$ containing $0$. For $\mathrm{aff}(K)$ the affine span of $K$ and $v\in H$ a vector, we define the linear transformation $T(\cdot; K, v)=\Pi_{\mathrm{aff}(K)}(\cdot)+\Pi_{\mathrm{span}\{v\}}(\cdot)$.
\end{definition}

To prepare for the induction step, suppose Theorem \ref{thm:nested.CPCA.measure.consistency} holds up to some $j-1 \leq k$. Let $C\in\mathfrak{C}_{j}$ and take $C_{j-1}$ with $C_{j-1}\subset C$ according to the definition of $\mathfrak{C}_{j}$. Similarly, take $C_{n,j-1}\in CC_{\overline{x}_n,j-1}(X)$ as in the statement of the Theorem. Define $K_{n,j-1}=C_{n,j-1}-\overline{x}_n$, $K_{\infty,j-1}=C_{j-1}-\overline{x}$ and $K= C - \overline{x}$. Note that
\[
K_{n,j-1},K_{\infty,j-1},K \subset \mathfrak{X} := X-X.
\]
It is easy to see from the compactness of $X$ that $\mathfrak{X}$ is compact.

\begin{lemma}\label{lem:translated.sets.haus.conv}
If $\mathsf{h}(C_{n,j-1},C_{j-1})\to 0$ then $\mathsf{h}(K_{n,j-1},K_{\infty,j-1})\to 0$.
\end{lemma}
\begin{proof}
By Lemma \ref{lem:kur.haus} it suffices to check Kuratowski convergence (Definition \ref{def:kuratowski}) as $\mathfrak{X}$ is compact. The proof is then straightforward.
\end{proof}

We will use these sets to define linear transformations using Definition \ref{def:key.lin.trans}. The required vector $v$ will be constructed as follows. By Lemma \ref{lem:full.dim.cpc}, we have $\dim C_{j-1} = j-1$, and by definition, we have $\dim C_{n,j-1}\leq j-1$, $\dim C\leq j$. Consider the subspaces of $H$ given by $A_n:=\mathrm{aff}(K_{n,j-1})$ and $A_\infty:=\mathrm{aff}(K_{\infty,j-1})$. Here, and throughout, for two \textit{subspaces} $A,A'$ of $H$ we will define $A\oplus A'$ to be their direct sum:
\[
A\oplus A'=\{x\in H: x= a+a', \  a\in A, a'\in A'\}.
\]
Since $A_\infty\subset \mathrm{aff}(K)$ we have that either (i) $A_\infty=\mathrm{aff}(K)$, or (ii) as $\dim(K_{\infty,j-1})=j-1$, there exists a $b\in \mathrm{aff}(K)$ such that $b$ is orthogonal to $A_\infty$ and $\mathrm{aff}(K)=A_\infty\oplus\mathrm{span}\{b\}$. In this case, we fix such a $b$ to play the role of the defining vector $v$ in Definition \ref{def:key.lin.trans}. In the former case (i), we instead take $b=0$.

With this, we will now finally define
\begin{equation}\label{eqn:Tn}
    T_n:=T(\cdot;K_{n,j-1},b)=\Pi_{A_n}+\Pi_{\mathrm{span}\{b\}}
\end{equation} and 
\begin{equation}\label{eqn:Tinfty}
    T_\infty:=T(\cdot;K_{\infty, j-1},b)=\Pi_{A_\infty}+\Pi_{\mathrm{span}\{b\}}.
\end{equation}
Next, it is important when taking the inductive step that the $T_n$ are consistent with $T_\infty$. We show below that $\Pi_{A_n}\to \Pi_{A_\infty}$ pointwise, and therefore, $T_n\to T_\infty$ pointwise.

\begin{lemma}
If $\mathsf{h}(C_{n,j-1},C_{j-1})\to 0$ we have $\Pi_{A_n}(x)\to \Pi_{A_\infty}(x)$ for all $x\in H$.
\end{lemma}
\begin{proof}
We begin by showing the affine subspaces $(A_n)_{n\geq1}$ converge to $A_\infty$ in an appropriate sense. We recall that we must have $K_{n,j-1}\to K_{\infty,j-1}$ in the Kuratowski sense (see Lemmas \ref{lem:translated.sets.haus.conv} and \ref{lem:kur.haus}). Let $j_n:=\dim A_n\leq j-1$. Consider an arbitrary orthonormal basis $(\phi^{(n)}_\ell)_{\ell=1}^{j_n}$ of $A_n$ and extend it to an orthonormal basis of $H$, $(\phi^{(n)}_{\ell})_{\ell=1}^{d}$ (where $d=\dim H$). By definition every $y\in K_{n,j-1}$ can be written as
$y=\sum_{\ell=1}^{d}w_\ell\phi^{(n)}_\ell$
for some $(w_\ell)_{\ell=1}^{d}$ with $w_{\ell'}=0$ if $\ell'> j_n$. 

Let $x\in K_{\infty, j-1}$ be arbitrary. By the Kuratowski convergence (Definition \ref{def:kuratowski}) there exists a sequence $y_n\in K_{n,j-1}$ with $y_n\to x$. By the compactness of the unit sphere we may extract a subsequence of the bases $(\phi_\ell^{(n)})_{\ell=1}^d$ such that $\phi_\ell^{(n')}\to \bar{\phi}_\ell$ for some  $\bar{\phi}_\ell$ and all $\ell=1,\dots,d$. Write $y_{n'}=\sum_{\ell=1}^{d}w_\ell^{(n')}\phi^{(n')}_\ell$.
By the uniform boundedness of $(K_{n,j-1})_{n\geq0},K_{\infty,j-1}$ (as $K_{n,j-1},K_{\infty,j-1}\subset \mathfrak{X}$) and the Bolzano-Weierstrass Theorem we may choose a further subsequence so that $w_\ell^{(n'')}\to \bar{w}_\ell$ for some $\bar{w}_\ell\in\mathbb{R}$ and all $\ell=1,\dots,d$.

Using the convergence $y_n\to x$ we write
\[
y_{n''}=\sum_{\ell=1}^{d}w_\ell^{(n'')}\phi^{(n'')}_\ell\to \sum_{\ell=1}^{d}\bar{w}_\ell\bar{\phi}_\ell=x
\]
As indices $\ell> j_{n''}$ are such that $w^{(n'')}_j=0$ for all $n''$ we have that we may write  $x=\sum_{\ell=1}^{j-1}\bar{w}_\ell\bar{\phi}_\ell$ (recall $j_{n''}\leq j-1$).
Notice that the convergent subsequence of bases is independent of the choice of $y_n$ (it relates only to $K_{n,j-1}$) so the same basis subsequence can be chosen for arbitrary $x\in K_{\infty,j-1}$. Thus, the representation $x=\sum_{\ell=1}^{j-1}w_\ell\bar{\phi}_\ell$ for $(w_\ell)_{\ell=1}^{j-1}$ depending on $x$ holds for arbitrary $x\in K_{\infty, j-1}$. We conclude that $A_\infty\subset \mathrm{span}\{\bar{\phi}_1,\dots,\bar{\phi}_{j-1}\}$. On the other hand, the orthogonality of $(\bar{\phi}_\ell)_{\ell=1}^{j-1}$ is assured by the limit of an orthonormal sequence. Consequently, we have $\dim\left(\mathrm{span}\{\bar{\phi}_1,\dots,\bar{\phi}_{j-1}\}\right)=j-1$. As $\dim A_\infty =j-1$ we conclude 
$\mathrm{span}\{\bar{\phi}_1,\dots,\bar{\phi}_{j-1}\}=A_\infty$.



We will now collect these observations to prove the convergence of the projection. As the choice of basis is immaterial to the definition of the projection, the above tells us that for any subsequence $(\Pi_{A_{n'}})$ we may find a further subsequence $(\Pi_{A_{n''}})$ of projections that can be written in terms of some converging orthonormal bases $(\phi^{(n'')}_\ell)_{\ell=1}^{j-1}$ of $A_{n''}$ where $\dim(A_{n''})=j-1$ (we can pick the subsequence so that $\mathrm{span}\{\phi^{(n'')}_1,\dots,\phi^{(n'')}_{j-1}\}=A_{n''}$). Therefore,
\[\Pi_{A_{n''}}(x)\to \Pi_{A_\infty}(x), \ \ \ \forall x\in H.\]
As the subsequence was arbitrary, we are done.
\end{proof}

We pause here to briefly remark on the properties of our linear transformations.

\begin{remark}[Properties of $T_n$ and $T_\infty$]\label{rmk:properties.lin.transformations} We can notice that
\begin{enumerate}
    \item[(i)] $T_\infty= \Pi_{A_\infty}+\Pi_{\mathrm{span}\{b\}}=\Pi_{\mathrm{aff}(K)}$,
    \item[(ii)]  $T_n[K_{\infty,j-1}]\subset A_n$,
    \item[(iii)]$\dim(T_n[K])\leq j$, and;
    \item[(iv)] $\dim(T_n[K]\cup K_{n,j-1})\leq \dim(A_n\oplus \mathrm{span}\{b\})\leq j$.
\end{enumerate}
\end{remark}


In (i) and (ii) we have used that $A_\infty\perp \mathrm{span}\{b\}$. Before proceeding to our final analysis, we highlight two more important lemmas whose conclusions hold for our transformations. First, we have the following which can be proved as a straightforward consequence of the Uniform Boundedness Principle.
\begin{lemma}\label{lem:T.unif.on.cmpcts}
If $(T_n)_{n\geq1}$ are pointwise bounded continuous linear operators with $T_n:H\to H$, and $T_\infty$ is continuous with $T_n(x)\to T_\infty(x)$ for all $x\in H$ then $T_n\to T_\infty$ uniformly on compacts.
\end{lemma}
Second, we have the following standard result.
\begin{lemma}\label{lem:T.composition.conv}
If $(T_n)_{n\geq1}$ is such that $T_n\to T_\infty$ uniformly on $X\subset H$ and $x_n\in X$ is such that $x_n\to x\in X$ then $T_n(x_n)\to T_\infty(x)$.
\end{lemma}

\begin{lemma} \label{lem:nested.consistency.inductive.step}
Let $j\geq 1$, $C_{n,j-1}\in CC_{\overline{x}_{n},j-1}(X)$, and $C_{j-1}\in CC_{\overline{x},j-1}(X)$ with $\dim(C_{j-1})=j-1$. If $\mathsf{h}(C_{n,j-1},C_{j-1})\to 0$ then, $\Gamma\textnormal{-}\lim_{n\to\infty}\chi_{\mathfrak{C}_{n,j}}= \chi_{\mathfrak{C}_j}$.
\end{lemma}

\begin{proof}
As with the previous consistency results it suffices to show that $\mathfrak{C}_{j}=K\textnormal{-}\lim_{n\to\infty}\mathfrak{C}_{n,j}$ (see Lemma \ref{lem:kur.gamma}). We check the conditions in the definition (Definition \ref{def:kuratowski}).

(i) Suppose $C\in \mathfrak{C}_j$. Define $C''_{n}:=T_{n}[C-\overline{x}]+\overline{x}_n$ where $T_n$ and $T$ are determined through equations \eqref{eqn:Tn} and \eqref{eqn:Tinfty}. 
We see that $\dim(C''_n)\leq j$ and that the convexity and compactness of $C_n''$ is inherited by the linearity and continuity of $T_n$. Next, we will see that $C''_n\to C$ in Hausdorff distance. By the boundedness (in operator norm) of the family $(T_n)_{n\geq1}$, the set $\bigcup_nT_n[X-\overline{x}]+X$ is contained in a (large) bounded set in $H$. Since $H$ is finite dimensional, we can further take it to be compact. Thus, to get $\mathsf{h}(C_n'',C)\to 0 $ it suffices to show the Kuratowski convergence of $C_n''$ to $C$ (see Lemma \ref{lem:kur.haus}):

(i$'$) Let $x\in C$ and define $x_n=T_n(x-\overline{x})+\overline{x}_n\in C''_n$. We have, in view of Lemma \ref{lem:T.unif.on.cmpcts} and Remark \ref{rmk:properties.lin.transformations}(i), that $T_n\to \Pi_{\mathrm{aff}(C-\overline{x})}$ uniformly on $X-\overline{x}$. Consequently, as $\overline{x}_n\to \overline{x}$, the convergence $x_n\to x$ holds.

(ii$'$) Let $y_n\in C''_n$ and suppose $\tilde{y}$ is an accumulation point of $(y_n)_{n\geq1}$. We have there exists a subsequence such that $y_{n'}\to \tilde{y}$. By definition of $C''_n$ there exists an $x_{n'}\in C$ such that $y_{n'}=T_{n'}(x_{n'}-\overline{x})+\overline{x}_{n'}$. By compactness of $C$ extract a further subsequence such that $x_{n''}\to \tilde{x}\in C$. Applying the properties of $T_n$ (particularly Lemma \ref{lem:T.composition.conv}) we have
\begin{align*}\|\tilde{y}-\tilde{x}\|&=\lim_{n''\to\infty}\|y_{n''}-x_{n''}\|\\
&= \lim_{n''\to\infty}\|T_{n''}(x_{n''}-\overline{x})+\overline{x}_{n''}-x_{n''}\|=\|\tilde{x}-\overline{x}+\overline{x}-\tilde{x}\|=0.
\end{align*}
So, $\tilde{y}=\tilde{x}\in C$ and we conclude $h(C_n'',C)\to 0$. 

With this established, let $C'_n:=\mathrm{conv}(C''_n\cup C_{n,j-1})$. Note that we must have that $\dim(C'_n)\leq j$ (see Remark \ref{rmk:properties.lin.transformations}(iv) and argue by translation). By Lemma \ref{lem:conv.hull.properties} we have $C'_n\to C$ in Hausdorff distance as $C''_n\to C$ and $C_{n,j-1}\to C_{j-1}\subset C$. Here we also have $C'_n$ is convex and compact (again by Lemma \ref{lem:conv.hull.properties}), and $\overline{x}_n\in C'_n$. It remains to get a set that has the same properties, but is assured to be in $X$. Define $C_n=C'_n\cap X$. Applying the same reasoning as in the proof of Theorem \ref{thm:Global.CPCA.measure.consistency} we can make use of Theorem \ref{thm:haus.intersection} to get $\mathsf{h}(C_n, C)\to 0$. As $C_n\in \mathfrak{C}_{n,j}$ we are done with the first condition.

(ii) Let $C_{n}\in \mathfrak{C}_{n,j}$ for $n\geq0$. Suppose $C$ is an accumulation point of the $C_{n}$. As $\overline{x}_{n}\to \overline{x}$ we have $\overline{x}\in C$ (by Hausdorff convergence). We can conclude $C\in CC_j(X)$ by compactness since $C_{n}\in CC_j(X)$ for all $n$. Hence, $C\in CC_{\overline{x},j}(X)$. Finally, $C\supset C_{j-1}$ by \cite[Lemma A.4]{BGKL17} and so, $C\in \mathfrak{C}_{j}$. This completes the proof.
\end{proof}

We are now ready to prove the theorem.

\begin{proof}[Proof of Theorem \ref{thm:nested.CPCA.measure.consistency}]
Since global and nested CPCA are equivalent when $k = 1$, the case $j = 1$ is a consequence of Theorem \ref{thm:Global.CPCA.measure.consistency}.

We then argue by induction on $j$. Suppose (i)-(iii) hold up to some $j$. Then, along a subsequence $C_{n',\ell}\to C_\ell\in CC_{\overline{x},\ell}(X)$ for each $\ell \leq j$ and $C_0\subset C_1\subset \dots \subset C_j$ form a sequence of nested convex principal components. By Lemma \ref{lem:full.dim.cpc} the convex principal components are full dimensional. Specifically, $\dim(C_j)=j$ and so all the assumptions of Lemma \ref{lem:nested.consistency.inductive.step} are satisfied for $\mathfrak{C}_{n',j+1},\mathfrak{C}_{j+1}$. Thus, (iii) holds for $j+1$ and we can then extract a further subsequence such that (i) and (ii) also hold up to $j+1$. Thus by induction the theorem holds up to $k$.
\end{proof}


\subsubsection{Proof of Theorem \ref{thm:global.consistency.CPCA.approx}}

\begin{proof}[Proof of Theorem \ref{thm:global.consistency.CPCA.approx}]
Note again by Lemma \ref{lem:kur.gamma} it suffices to show
\[
CC_{\bar{x},k}(X)=\Klim_{n\to\infty} CC_{\bar{x}_n,k}(X_n).
\]
We will prove the Kuratowski convergence via the definition (Definition \ref{def:kuratowski}). 

(i) Let $C\in CC_{\overline{x},k}(X)$ be arbitrary and let $C_n = \Pi_n C$. It is easy to verify that $C_n \in CC_{\overline{x}_n,k}(X_n)$. To prove (i) it is sufficient to show that $\mathsf{h}(C_n,C)\to 0$.

We will appeal to the equivalence of Kuratowski and Hausdorff convergence in the compact metric space $(X,\mathsf{d})$ (see Lemma \ref{lem:kur.haus}). First, if $x\in C$ we have that $\Pi_nx \in C_n$ is such that $\|\Pi_nx-x\|\to 0$. Second, consider $y_n\in C_n$ for $n\geq 1$ and let $y$ be an accumulation point. There must exist a subsequence $(y_{n'})$ such that $\|y_{n'}-y\|\to 0$. By definition we have that each $y_n$ corresponds to at least one $x_n\in C$ so that $y_n=\Pi_n x_n$. By compactness of $C$ there exists a further subsequence $(x_{n''})$ and $x\in C$ such that $\|x_{n''}-x\|\to 0$. Consider also $\Pi_nx$. For all $\epsilon>0$ there exists an $N$ such that for $n''$ sufficiently large we have $\|\Pi_{n''}x-x\|<\epsilon/3$, $\|y_{n''}-y\|<\epsilon/3$, and $\|x_{n''}-x\|<\epsilon/3$. Since $\Pi_n$ is $1$-Lipschitz, we have, for $n''$ sufficiently large,
\begin{align*}
    \|y-x\|&\leq \|y-y_{n''}\|+\|y_{n''}-\Pi_{n''}x\|+\|\Pi_{n''}x-x\|\\
    &=\frac{2}{3}\epsilon+\|\Pi_{n''}x_{n''}-\Pi_{n''}x\|
    \leq \frac{2}{3}\epsilon+\|x_{n''}-x\|< \epsilon.
\end{align*}
Since $\epsilon>0$ is arbitrary, we have $y=x$. It follows that $\mathsf{h}(C_n,C)\to 0$.

(ii) Let $C_n\in CC_{\overline{x}_n,k}(X_n)$ for $n\geq1$. Suppose that $C$ is an accumulation point of $(C_n)_{n\geq1}$. Then $\mathsf{h}(C_{n'},C)\to0$ along some subsequence. As $C_{n'}\ni\overline{x}_{n'}=\Pi_{n'}\overline{x}\to \overline{x}$ we have that $\overline{x}\in C$ (again by appealing to Kuratowski convergence). On the other hand, note that $CC_k(X)$ is compact so as $C_{n'}\in CC_k(X_n)\subset CC_k(X)$ we have that $C\in CC_k(X)$.
We conclude $C\in CC_{\overline{x},k}(X)$ which completes the proof.
\end{proof}

\subsubsection{Proof of Theorem \ref{thm:cont.pca}}\label{sec:pf.cont.pca}
We present directly the main argument, and the lemmas needed will be proved in the discussion to follow.

\begin{proof}[Proof of Theorem \ref{thm:cont.pca}]

Write $z_i = \overline{x} + a_i {\bf p} \in X$ and consider
\[
V({\bf p}) =\frac{1}{N} \sum_{i = 1}^N \left\{ \inf_{a_i\in\mathbb{R}:  \overline{x} + a_i {\bf p}\in X} \|x_i - \overline{x}-a_i{\bf p}\|^2 \right\}.
\]
Since there are only finitely many data points, there exists $r > 0$ such that the optimal $z_i$ has norm at most $\sqrt{r}$, for all $i$ and independent of ${\bf p}$. Without loss of generality, we may introduce an additional convex constraint $g_{m+1} = \|x\|^2 - r \leq 0$, and we may choose $r$ such that $g_{m+1}(\overline{x})<0$. Let $X' = X\cap\{x\in H: g_{m+1}(x)\leq0\}$ be the new convex state space which does not affect the solution. Define the correspondence 
\begin{equation}\label{eqn:Phi.correspondence}
    \Phi(\mathbf{p}) =\bigcap_{j=1}^{m+1}\{a\in\mathbb{R}:  g_j(\overline{x}+a\mathbf{p})\leq 0\}=\{a\in\mathbb{R}:\overline{x}+a{\bf p}\in X'\}
\end{equation}
on $H\setminus \{0\}$. We now have
\begin{equation}\label{eqn:V.with.Phi}
V({\bf p}) =\frac{1}{N} \sum_{i = 1}^N \left\{ \inf_{a_i\in \Phi(\mathbf{p})} \|x_i - \overline{x}-a_i{\bf p}\|^2 \right\}.
\end{equation}
By Lemma \ref{lem:correspondence} below, $\Phi$ is continuous, non-empty and compact valued. This, alongside the continuity of $\|x_i-\overline{x}-a_i\mathbf{p}\|^2$ allows us to apply Berge's Maximum Theorem (see Theorem \ref{thm:berge}) to get that
\[
\psi_i(\mathbf{p})=\inf_{a_i\in\Phi(\mathbf{p})}\|x_i-\overline{x}-a_i\mathbf{p}\|^2
\]
is continuous on $H\setminus\{0\}$. This implies that $V(\cdot)$, as an average of the continuous functions $\psi_i(\cdot)$, is continuous. Berge's Maximum Theorem also tells us that that the set of optimizers $a^*_i(\mathbf{p})$ corresponding to each $\psi_i$ is upper-hemicontinuous, non-empty, and compact. By the strict convexity of $\|x_i-\overline{x}-a_i\mathbf{p}\|^2$ in $a_i\in\mathbb{R}$ we have a unique optimizer. 
Hence, we also have that $a^*_i(\mathbf{p})$ is a singleton. By Lemma \ref{lem:hemi.singleton}
this implies that it is continuous on $H\setminus\{0\}$ when viewed as a function.
\end{proof}

\begin{lemma} \label{lem:correspondence}
The correspondence $\Phi(\mathbf{p})$ from \eqref{eqn:Phi.correspondence}
is non-empty, continuous, convex, and compact-valued for every $\mathbf{p}\in H\setminus\{0\}$ under the assumptions of Theorem \ref{thm:cont.pca}. In particular, it is a compact interval.
\end{lemma}

\begin{proof}
Clearly, $\Phi$ is closed and convex as the intersection of the sub-level sets of continuously differentiable convex functions in $a$. It is non-empty since $0\in\Phi(\mathbf{p})$ by assumption on $\overline{x}$ and $r$. Moreover, by the last constraint $\Phi$ is bounded for each fixed $\mathbf{p}\in H\setminus \{0\}$. Thus, it is compact by the Heine-Borel Theorem. Since $\Phi$ is also convex and compact in $\mathbb{R}$ it must define a closed interval. It remains to show that $\Phi$ is continuous as a correspondence. To see this first consider the correspondence:
\begin{equation}\label{eqn:overline.underline.a}\phi(\mathbf{p})=\{\underline{a}(\mathbf{p}),\overline{a}(\mathbf{p})\},
\end{equation}
where $\underline{a}(\mathbf{p})$ solves for the intersection of the line  $\overline{x}+a\mathbf{p}$ with the feasible region boundary when $a<0$. That is
\[
\overline{x}+\underline{a}(\mathbf{p})\mathbf{p}\in \partial\{x\in H:  g_j(x)\leq 0,\  j=1,...,m+1\}
\]
and $\underline{a}(\mathbf{p})<0$. Define $\overline{a}(\mathbf{p})$ similarly for the case when $a>0$. Note that there are always two such points since by assumption the origin of the line, $\overline{x}$, is always in the interior of the feasible region and by the convexity of the region the line must exit (and hence intersect) the region at two points. By Lemma \ref{lem:a.cont} below the functions $\overline{a}(\cdot),\underline{a}(\cdot)$ are continuous in $\mathbf{p}\in H\setminus \{0\}$. As a result, by Theorem \ref{thm:cont.func.convex.hull.cont} $\phi$ is continuous as a correspondence and again by Theorem \ref{thm:cont.func.convex.hull.cont} since $\Phi$ is the convex hull of $\phi$ in $\mathbb{R}$, $\Phi$ is also continuous as a correspondence.
\end{proof}

To prove the continuity of $\overline{a}(\cdot),\underline{a}(\cdot)$ from \eqref{eqn:overline.underline.a} we will need some preliminary lemmas. The first is a technical result that states we can choose a full dimensional convex cone from the origin to intersect a portion of the boundary of a compact convex set with non-empty interior containing the origin. Heuristically, it says that we can shine a ``flashlight'' from the interior onto the boundary of a convex set.


\begin{lemma}[Flashlight lemma] \label{lem:flashlight}
Let $X$ be a compact, convex set in $H$ for $\dim(H)\geq 2$ with non-empty interior containing the origin. Suppose $x_0\in\partial X$ and consider a ball of radius $\delta$ about $x_0$, $B_\delta(x_0)$ that does not contain $0$. There exists a convex cone $K=\{\lambda v:\lambda\in\mathbb{R}_+,\ v\in V\subset H\}$ of rays from the origin with non-empty interior containing $x_0$ and (necessarily) intersecting $B_\delta(x_0)$ such that $\partial X\cap K \subseteq \partial X\cap B_\delta(x_0)$.
\end{lemma}

\begin{proof}
As $\dim(H)\geq 2$ by the homeomorphism $x/\|x\|$ we can restrict our attention to $\partial X=S^{d-1}=\{x\in H: \|x\|=1\}$. Consider $D:=\partial B_\delta(x_0)\cap S^{d-1}$ which is a $d-2$ dimensional sphere if $d>2$, and $2$ distict points otherwise. Then, $\conv(D)$ is an $d-1$ dimensional ball that lives in an affine subspace $A$ of $H$ that does not contain $0$ and satisfies $\dim(A)=d-1$. Let $K=\mathrm{cone}(\conv(D))=\bigcup_{\lambda\geq0}\lambda\conv(D)$ and notice that every ray in $K$ intersects $S^{d-1}$ at a point in $B_\delta(x_0)\cap S^{d-1}$. Furthermore, it is clear $K$ is full dimensional (as $\dim(A)=d-1$) and $x_0\in \mathrm{int}(K)$.
\end{proof}

The second lemma characterizes the (local) solution $a^*_k(\cdot)$ of the equation defining the points where a constraint $g_j(\overline{x}+a\mathbf{p})$ from $\Phi(\mathbf{p})$ in \eqref{eqn:Phi.correspondence} is binding.

\begin{lemma} \label{lem:local.sol} Let $\bar{g}_j(\mathbf{p},a)=g_j(\overline{x}+a\mathbf{p})$ and suppose there exists a non-zero solution $a^*$ to the equation
\[
\bar{g}_j(\mathbf{p}_0,a)=0
\]
at $\mathbf{p}_0\in H$ for some $j\in\{1,...,m+1\}$. Then under the assumptions of Theorem \ref{thm:cont.pca} there exist an open set $U_j$ containing $\mathbf{p}_0$ and a unique $C^1$ function $a_j^*:U_k\to\mathbb{R}$ such that $a^*_j(\mathbf{p}_0)=a^*$ and 
\[
\bar{g}_j(\mathbf{p},a^*_j(\mathbf{p}))=0, \quad \mathbf{p}\in U_k
\]
In particular, if $a^*>0$ (resp. $a^*<0$) there exists an open set $V_j\subset U_k$ containing $\mathbf{p}_0$ such that $a_j^*(\mathbf{p})>0$ (resp.~$a_j^*(\mathbf{p})<0$) on $V_k$.
\end{lemma}

\begin{proof}
We begin by showing that $\frac{\partial \bar{g}_j}{\partial a}(\mathbf{p}_0,a^*)\not=0$ by contradiction. Suppose that $\frac{\partial \bar{g}_j}{\partial a}(\mathbf{p}_0,a^*)=0$. Since $\bar{g}_k$ is convex in the coordinate corresponding to $a$ 
we have that $\bar{g}_j(\mathbf{p}_0,a^*)=0$ is a minimum of $\bar{g}_j(\mathbf{p}_0,a)$. However, by setting $a=0$ this minimality contradicts our assumption on $\overline{x}$ which says $\overline{g}_j(\mathbf{p}_0,0)<0$. Hence, we conclude that $\frac{\partial \bar{g}_j}{\partial a}(\mathbf{p}_0,a^*)\not=0$. With this result, since we have assumed the existence of a solution $a^*$ to the equation $\bar{g}_j(\mathbf{p}_0,a)=0$ we can apply the implicit function theorem to get the first result in the statement of the lemma. The second result follows taking $V_j =a_j^*{}^{-1}((0,\infty))$ as the preimage which is open by continuity and contains $\mathbf{p}_0$ by assumption on $a^*=a_j^*(\mathbf{p}_0)$.
\end{proof}

With this we can tackle the final result of this section.

\begin{lemma}\label{lem:a.cont}
The functions $\overline{a}(\mathbf{p})$ and $\underline{a}(\mathbf{p})$ from \eqref{eqn:overline.underline.a} are continuous in ${\bf p} \in H \setminus \{0\}$.
\end{lemma}

\begin{proof}
It suffices to consider $\overline{a}(\mathbf{p})$. Begin by fixing a vector $\mathbf{p}_0\in H\setminus\{0\}$. The point $\overline{x}+\overline{a}(\mathbf{p}_0)\mathbf{p}_0$ is on the boundary of $X'=\bigcap_{j=1}^{m+1}\{x\in H:  g_j(x)\leq 0\}$ 
and so by Lemma \ref{lem:constr.set.properties} this implies there is at least one $g_j$ such that
\[
g_j(\overline{x}+\overline{a}(\mathbf{p}_0)\mathbf{p}_0)=0
\]
with $\overline{a}(\mathbf{p}_0)>0$. By Lemma \ref{lem:local.sol}, for each $j$ for which this holds there exist an open set $U_j$ containing $\mathbf{p}_0$ and unique $C^1$ functions $a_j^*:U_j\to \mathbb{R}_+$ such that $a_j^*(\mathbf{p}_0)=\overline{a}(\mathbf{p}_0)$ and
\[
g_j(\overline{x}+a^*_j(\mathbf{p})\mathbf{p})=0, \quad \mathbf{p}\in U_j.
\]
With this established, let $\mathcal{B}\subset\{1,...,m+1\}$ be the set of indices corresponding to the binding/active constraints at $\mathbf{p}_0$ and let $\mathcal{I}$ be the set of inactive constraints at $\mathbf{p}_0$ (note $\mathcal{I}\cup\mathcal{B}=\{1,...,m+1\}$). 

For each $j\in \mathcal{I}$ consider the set
\[
I_j=g^{-1}_{j}((-\infty,0))
\]
Clearly $\overline{x}+\overline{a}(\mathbf{p}_0)\mathbf{p}_0\in I_k$ by assumption. Since the $g_k$ are continuous, each $I_k$ is open. Let:
\[
I=\bigcap_{j\in\mathcal{I}}I_j.
\]
As the intersection of finitely many non-empty open sets containing $\overline{x}+\overline{a}(\mathbf{p}_0)\mathbf{p}_0$, $I$ is open and non-empty. Moreover, if $x$ is in $I$ we have that $g_j(x)<0$ for all $k\in\mathcal{I}$. 

Now consider the functions $\mathbf{h}_j:U_j\to H$ for $j\in \mathcal{B}$ given by
\[
\mathbf{h}_j(\mathbf{p})=\overline{x}+a^*_j(\mathbf{p})\mathbf{p}.
\]
By the continuity of $a^*_j$ on $U_j$ this is continuous on $U_j$. Moreover since $\mathbf{p}_0$ is such that $\mathbf{h}_j(\mathbf{p}_0)=\overline{x}+\overline{a}(\mathbf{p}_0)\mathbf{p}_0\in I$ we have that $O_j=\mathbf{h}_j^{-1}(I)$ is open and non-empty containing $\mathbf{p}_0$. Let
\[
O = \bigcap_{j\in\mathcal{B}} O_j.
\]
Again, $O$ is open and contains $\mathbf{p}_0$. In particular, each of the functions $a^*_j$ for $j\in\mathcal{B}$ are continuous on $O$ and the values of $h_j(\mathbf{p})$ on $O$ correspond to boundary values of the level sets
\[
\{x\in H :g_j(x)\leq 0\}.
\]
Moreover, each  boundary value is contained in $I$ where no other constraint is active.

Choose $\delta$ sufficiently small so that the ball $B_\delta(\overline{x}+\overline{a}(\mathbf{p}_0)\mathbf{p}_0)$ is contained in $I$ and does not contain $\overline{x}$. Consider now our compact convex feasible set $X'$ and its boundary $\partial X'$.
%
Since $X'$ has non-empty interior by assumption, by translation we can apply Lemma \ref{lem:flashlight} to conclude there exists a cone $K$ of rays with origin $\overline{x}$ and non-empty interior containing $\overline{x}+\overline{a}(\mathbf{p}_0)\mathbf{p}_0$  such that 
\[
K\cap\partial X'\subset B_\delta(\overline{x}+\overline{a}(\mathbf{p}_0)\mathbf{p}_0)\cap \partial X'.
\]
Hence, since our open ball is in $I$ by assumption and every ray in $K$ intersects a boundary point (by the origin being $\overline{x}$ in our convex set), we have that every ray in $K$ intersects a boundary point of our feasible set in $I$. Since $K$ has non-empty interior and $\overline{x}+\overline{a}(\mathbf{p}_0)\mathbf{p}_0\in \overset{\circ}{K}$ we can choose $\eta>0$ small enough so that $B_\eta(\overline{x}+\overline{a}(\mathbf{p}_0)\mathbf{p}_0)\subset K$. Translating by $\overline{x}$ and scaling we obtain a ball of radius $\eta'$ about $\mathbf{p}_0$ such that every $\mathbf{p}\in B_{\eta'}(\mathbf{p}_0)$ defines a ray from $\overline{x}$ that intersects a boundary point in $I$. That is, the boundary point corresponding to such a $\mathbf{p}$ can only be associated with the boundary of a level set for some $g_k$ with $k\in\mathcal{B}$. 

Define $U = O\cap B_{\eta'}(\mathbf{p}_0)$. Clearly $U$ is open and contains $\mathbf{p}_0$. For every $\mathbf{p}$ in $U$, the intersection with the boundary of our feasible set $\overline{x}+\overline{a}(\mathbf{p})\mathbf{p}$ is a boundary point of one of the level sets of $g_j$ for $j\in\mathcal{B}$. Additionally, the intersection with each level set is given by a continuous function $a^*_j$ on $U$. Taken together we have that for each $\mathbf{p}\in U$,  $\overline{a}(\mathbf{p})=a^*_j(\mathbf{p})$ for some $j\in\mathcal{B}$.

Finally, to see continuity at $\mathbf{p}_0$ let $\epsilon>0$ and let $\mathbf{p}_1\in H\setminus \{0\}$. Let $\theta>0$ be such that $B_\theta(\mathbf{p}_0)\subset U$. Since the $a^*_j$ are continuous on $U$ for each $k$ we have that there exists a $\delta_j>0$ ($\delta_j<\theta$) such that if $\|\mathbf{p}_1-\mathbf{p}_0\|_2<\delta_j$ then $|a^*_j(\mathbf{p}_1)-a^*_j(\mathbf{p}_0)|<\epsilon$. Let $\delta^*=\min_{j\in \mathcal{B}}\delta_j$. Since for any $\mathbf{p}_1\in U$ there is a $j_1\in \mathcal{B}$ active we have that if $\|\mathbf{p}_1-\mathbf{p}_0\|_2<\delta^*$ then:
\[
|\overline{a}(\mathbf{p}_1)-\overline{a}(\mathbf{p}_0)|\leq \max_{j\in\mathcal{B}}|a^*_{j}(\mathbf{p}_1)-a^*_{j}(\mathbf{p}_0)|<\epsilon
\]
as required since the same $j_1$ is also active at $\mathbf{p}_0$. Hence $\overline{a}(\mathbf{p})$ must be continuous at $\mathbf{p}_0$ by definition of $U$. Now since the choice of $\mathbf{p}_0$ above was arbitrary in $H \setminus\{0\}$ we conclude that $\overline{a}:H\setminus\{0\}\to \mathbb{R}_+$ is continuous. 
\end{proof}

\subsubsection{Proof of Theorem \ref{thm:pca.diff}} \label{app:Thm.pca.dff}

\begin{proof}
As in the proof of Theorem \ref{thm:cont.pca} we can again restrict to a feasible set contained in a closed ball with sufficiently large radius. Let $\Phi(\mathbf{p})$ be similarly defined as in \eqref{eqn:Phi.correspondence}. For each $i$, let $J_i(a,\mathbf{p}) = \|x_i-\overline{x}-a\mathbf{p}\|^2$.

We look to apply Theorem \ref{thm:gen.env} to each subproblem in \eqref{eqn:V.with.Phi} for which we must verify the necessary assumptions. Note that as $H$ is finite dimensional Theorem \ref{thm:gen.env} (and Corollary \ref{cor:gen.env.diff}) applies to our setting once we verify the required conditions. For (i) $\mathcal{A}=\mathbb{R}$ is trivially convex. For (ii) it is clear that $J_i$ and $g_j(\overline{x}+a\mathbf{p})$,  $j=1,...,m+1$ are jointly continuous in $(a,\mathbf{p})$ by assumption. For (iii) and (iv) we showed in Lemma \ref{lem:correspondence} that $\Phi(\mathbf{p})$ is compact valued and continuous for $\mathbf{p}\in H\setminus\{0\}$. For (v) this follows by our assumption on $\overline{x}$ by setting $a=0$. For (vi) since $J_i$ and $g_j(\overline{x}+a\mathbf{p})$ are convex in $a$ for $j=1,...,m+1$ and for $a=0$, $g_j(\overline{x})<0$, $j=1,...,m+1$ we have that Slater's condition holds for every $\mathbf{p}\in H\setminus\{0\}$. This then implies that the saddle points of the Lagrangian are non-empty for every $\mathbf{p}\in H\setminus\{0\}$. Finally, (vii) holds since we have assumed the $g_j$ are continuously differentiable and it is clear this similarly holds true for the $J_i$. 
Taken together, we obtain the result of Theorem \ref{thm:gen.env} for our problem which gives \textit{directional} differentiability. Finally, to establish the points of differentiability, it suffices to note that the assumptions of Corollary \ref{cor:gen.env.diff} hold trivially for our problem when zero or one constraint is binding at the solution corresponding to the choice $\mathbf{p}\not=0$. In fact, this is the only case where the assumptions hold since our minimization variable $a$ is one-dimensional. Averaging over the sub-problems indexed by $i$ gives the result in the statement of the Theorem.

We close by collecting here the form of the directional derivative and gradient for our value function. For the set of saddle points $\mathcal{A}_i^*(\mathbf{p})\times\Lambda^*_i(\mathbf{p})\subset \mathbb{R}\times\mathbb{R}_+^{m+1}$ in Theorem \ref{thm:gen.env} corresponding to each inner problem $J_i$ we have that $V$ admits the directional derivative
\[
V'(\mathbf{p};\hat{\mathbf{p}})=\frac{1}{N}\sum_{i=1}^N\max_{a\in\mathcal{A}_i^*(\mathbf{p})}\min_{\boldsymbol{\lambda}\in\Lambda^*_i(\mathbf{p})}\left[\nabla_\mathbf{p} J_i(a,\mathbf{p})+\sum_{j=1}^{m+1}a\lambda_j\nabla g_j(\overline{x}+a\mathbf{p})\right]\hat{\mathbf{p}},
\]
where $\nabla_{\mathbf{p}}J_i(a,\mathbf{p})=-2a(x_i-\overline{x}-a\mathbf{p})$. Moreover, when $V$ is differentiable at some $\mathbf{p}_0\in H\setminus\{0\}$ we have for the unique choice $(a^*_i(\mathbf{p}_0),\boldsymbol{\lambda}^*_{i}(\mathbf{p}_0))\in \mathcal{A}_i^*(\mathbf{p})\times\Lambda^*_i(\mathbf{p})$ that the gradient takes the form
\[
\nabla_\mathbf{p}V(\mathbf{p}_0)=\frac{1}{N}\sum_{i=1}^N\left[\nabla_\mathbf{p}J_i(a^*_i(\mathbf{p}_0),\mathbf{p}_0)+\sum_{j=1}^{m+1}a^*_i(\mathbf{p}_0)\lambda^*_{i,j}(\mathbf{p}_0)\nabla g_j(\overline{x}+a^*_i(\mathbf{p}_0)\mathbf{p}_0)\right].
\]

\end{proof}

\section{Algorithmic Considerations}\label{sec:appendix.alg.considerations}

\setcounter{algocf}{0}
\renewcommand{\thealgocf}{\ref{sec:appendix.alg.considerations}.\arabic{algocf}}

We discuss here in greater detail the algorithm outlined in Section \ref{sec:algorithm}. In particular, we explain the evaluation of the value function and its gradient.

\subsection{Optimization}
We first touch on function evaluation. As mentioned in Section \ref{sec:algorithm}, for each $i$, the inner minimization problem in \eqref{eqn:nested.PCA.polyhedral1} is easy to solve since $L = (\overline{x} + \mathrm{span}\{ {\bf p} \}) \cap X$ can be identified with an interval $[t_0({\bf{p}}),t_1({\bf{p}})]$ for $-\infty\leq t_0({\bf{p}})<0<t_1({\bf{p}})\leq \infty$. The boundary point(s) of $L$ are given by $(z_0({\bf p}),z_1({\bf p}))=\left(\overline{x}+t_0({\bf p}){\bf p},\overline{x}+t_1({\bf p}){\bf p}\right)$. If we know the interval endpoints then we can use them to compute the inner minimization. 
First, we compute the usual orthogonal projection, say $\tilde{z}_i^*=\overline{x}+\tilde{t}_i^* {\bf{p}}$, of $x_i$ onto $L$. If $\tilde{t}_i^*\in [t_0({\bf{p}}),t_1({\bf{p}})]$ then $z_i^*=\tilde{z}_i^*$ is the optimal point. Otherwise, $z_i^*$ is the boundary point of $L$ which is closest to $\tilde{z}_i^*$. Hence, to evaluate the objective $V({\bf p})$ (equivalently, $\tilde{V}(\boldsymbol{\omega}(\boldsymbol{\theta}))$) we only need to efficiently find the boundary point(s) of $L$. We use that they are the closest intersections to $\overline{x}$ of the line $t\mapsto\overline{x}+t\mathbf{p}$ with the defining hyperplanes of $X$. The details are in Algorithm \ref{alg:intersection.pt}.

\begin{algorithm}[h]
\SetAlgoLined
\KwResult{Output boundary point coefficient(s) $t_0({\bf p}),t_1({\bf p})$.} 
Initialize $t_0({\bf p})=-\infty$, $t_1({\bf p})=\infty$\;
Let $(A_i)_{i=1}^m$ be the $i$th rows of $A$ and set for $i=1,\dots,m$ \[\alpha_i=\left((b_i-A_i^\top\overline{x})/A_i^\top {\bf p}\right)\mathds{1}_{A_i^\top {\bf p}\not=0}+\infty\mathds{1}_{A_i^\top {\bf p}=0}\] to be the intersection points (if they exist) of the line $t\mapsto\overline{x}+t\mathbf{p}$ with each hyperplane defined by the rows of $A$ and $\mathbf{b}$.\;
 \For{$i=1,...,m$}{
    \If{$\alpha_i>0 \ \mathbf{\mathrm{and}}  \ \alpha_i<t_1({\bf p})$}{$t_1({\bf p})=\alpha_i$\;
         }
         \If{$\alpha_i<0 \ \mathbf{\mathrm{and}}  \ \alpha_i>t_0({\bf p})$}{$t_0({\bf p})=\alpha_i$\;
         }
 }
 Return $(t_0({\bf p}),t_1({\bf p}))$\;
\caption{Find Boundary Point Coefficient(s)}
\label{alg:intersection.pt}
\end{algorithm}

Notice that the boundary points depend only on ${\bf p}$, and therefore are common for all the inner optimization problems in \eqref{eqn:nested.PCA.polyhedral1}. Hence, for problems with an extremely large data set, the inner problems can be solved in parallel. On the other hand, since for extremely high dimensional problems the gradient is similarly high dimensional, it is often convenient to parallelize the gradient estimation process (e.g.~by parallelizing independent central difference computations). For our purposes we implement the parallelization and function evaluation in \CC \ and rely on central differences for the gradient estimation. Since function evaluation is cheap using the aforementioned approach, this is a convenient way to recover the gradient.

Once this computational challenge is taken care of we can solve the resulting unconstrained optimization problem via a gradient descent approach or through other conventional solvers. Our implementation makes use of a variant of the Broyden–Fletcher–Goldfarb–Shanno (BFGS) algorithm in R available through the ``nloptr'' package \parencite{johnson2014nlopt}. Specifically, we take advantage of the ``Rcpp'' package \parencite{rcpp} to integrate our \CC \ implementation with R and perform the optimizations.

For algorithms of this type a relevant question is the choice of the initial point or ``guess'' for the solution. This is particularly important for large problems, as if we can be intelligent about our initialization, we can substantially improve the convergence time. For our approach, we choose an initial point for each successive principal component that is suggested through an iterated application of traditional PCA techniques. The explicit choice, and the rationale behind it, will be made precise in the following discussion.

\subsection{Choice of initialization point}\label{app:init.pt}

We begin this discussion with an equivalent representation of the convex PCA objective function.

\begin{proposition}[Equivalent formulations of PCA and CPCA Problems]\label{prop:equiv.rep.cpca}
Let $j\geq1$ and suppose $P_{j-1}$ is a $j-1$ dimensional subspace of $H$.
For
\[\hat{V}_j(\mathbf{p};X)=\frac{1}{N}\sum_{i=1}^N\inf_{a_i\in\mathbb{R}:\ \overline{x}+a_i\mathbf{p}\in X}\|\Pi_{P_{j-1}^\perp}(x_i-\overline{x})-a_i\mathbf{p}\|^2\]
and denoting the CPCA objective as $V(\mathbf{p};X)$ to make the dependence on $X$ explicit we have
\begin{enumerate}
    \item[(i)] $\argmin_{\mathbf{p}\in P_{j-1}^\perp:\|\mathbf{p}\|=1}V(\mathbf{p};X) =\argmin_{\mathbf{p}\in P_{j-1}^\perp:\|\mathbf{p}\|=1}\hat{V}_j(\mathbf{p};X)$, and;
    \item[(ii)] If there exists an $i\in\{1,\dots,N\}$ such that $x_i-\overline{x}\not\in P_{j-1}$ then \[\argmin_{\mathbf{p}\in P_{j-1}^\perp:\|\mathbf{p}\|=1}V(\mathbf{p};H) = \argmin_{\mathbf{p}\in H:\|\mathbf{p}\|=1}\hat{V}_j(\mathbf{p};H).\]
\end{enumerate}

\end{proposition}
\begin{proof}
See Section \ref{sec:prop.equiv.rep.cpca} below.
\end{proof}

Naturally, we are interested in the case where $P_{j-1}$ is given by the previous $j-1$ convex principal directions $(\mathbf{p}_{k}^*)_{k=0}^{j-1}$ (where $\mathbf{p}^*_0\equiv0$). The first statement says that to solve the nested CPCA problem one may equivalently \textit{first} project the data onto the orthogonal complement of the subspace spanned by the preceding principal directions and \textit{then} perform the minimization. The second statement says that the solution to traditional Euclidean PCA ($X=H$) with full dimensional data can be  recovered in the same way, but that we can further simplify to an \textit{unconstrained} minimization in the nested formulation of the problem. That is, we can simply perform a Euclidean PCA on the projected data. 

The key takeaway is given by the following corollary.

\begin{corollary}
Let $j\geq 1$ and suppose there exists an $i\in\{1,\dots,N\}$ such that $x_i-\overline{x}\not\in P_{j-1}$. If $\mathbf{p}^{*,H}_j\in\argmin_{\mathbf{p}\in H:\|\mathbf{p}\|=1}\hat{V}_{j}(\mathbf{p};H)$ then
\[\hat{V}_{j}(\mathbf{p}^{*,H}_j;H)=\inf_{\mathbf{p}\in P_{j-1}^\perp:\|\mathbf{p}\|=1}\hat{V}_{j}(\mathbf{p};H)\leq \inf_{\mathbf{p}\in P_{j-1}^\perp:\|\mathbf{p}\|=1}\hat{V}_{j}(\mathbf{p};X)\leq  \hat{V}_{j}(\mathbf{p}^{*,H}_j;X).\]
\end{corollary}
\begin{proof}
The first equality follows from Proposition \ref{prop:equiv.rep.cpca}(i) (with $X=H$) and (ii), and the choice of $\mathbf{p}^{*,H}_j$. The last inequality follows again from Proposition \ref{prop:equiv.rep.cpca}(ii) as we get $\mathbf{p}_j^{*,H}\in P_{j-1}^\perp$. The middle inequality then follows by adding the constraint that the projections are in $X$.
\end{proof}

This says that the Euclidean principal component can be used to bound the value of the convex PCA problem using the representation in Proposition \ref{prop:equiv.rep.cpca}. It is also clear that if none of the constraints defining $X$ are binding for all $i$ then the optimal values and minimizers of the two problems agree $\hat{V}_{j}(\mathbf{p}^{*,H}_j;H)=\hat{V}_{j}(\mathbf{p}^{*,H}_j;X)$. Note that this also illustrates that, as with standard Euclidean PCA problems, the minimizer for convex PCA problems need not be unique. Most importantly, as Proposition \ref{prop:equiv.rep.cpca}(i) says the minimizers of $\hat{V}_j$ and $V$ in $P_{j-1}^\perp$ coincide, this observation suggests that a Euclidean principal direction for the projected data is a good initial guess for the convex principal direction. Indeed, they will coincide if no constraints are binding and so we expect that $\mathbf{p}^{*,H}_j$ should be ``close'' to a true minimizer. Moreover, the direction $\mathbf{p}^{*,H}_j$ is cheap to compute as it is the solution to a regular PCA problem. As a result, in our optimization routine we initialize our search at some $\mathbf{p}^{*,H}_j\in\argmin_{\mathbf{p}\in H:\|\mathbf{p}\|=1}\hat{V}_{j}(\mathbf{p};H)$ to further accelerate convergence.

\subsection{Proof of Proposition \ref{prop:equiv.rep.cpca}} \label{sec:prop.equiv.rep.cpca}
To verify this we will use the following notation:
\[X=\bigcap_{j=1}^m\{x\in H: g_j(x)\leq0\}, \ \ \ \overline{X}=\bigcap_{j=1}^m\{x\in H: g_j(\overline{x}+x)\leq0\},\]
\[S_{\mathbf{p}}=\mathrm{span}\{\mathbf{p}\} \ \ \ \mathrm{and} \ \ \ \overline{S}_{\mathbf{p}}=\overline{x}+\mathrm{span}\{\mathbf{p}\}.\] Note that it is easy to see that for $x\in X$ we have
\begin{equation}\label{eqn:equiv.parameterizations}
\inf_{z\in \overline{S_{\mathbf{p}}}\cap X}\|x-z\|^2=\inf_{z\in S_{\mathbf{p}}\cap \overline{X}}\|x-\overline{x}-z\|^2=\inf_{a\in\mathbb{R}:\ \overline{x}+a\mathbf{p}\in X}\|x-\overline{x}-a\mathbf{p}\|^2.    
\end{equation}

\begin{lemma}[Iterated projection]\label{lem:iterated.proj}
We have $\Pi_{S_\mathbf{p}\cap \overline{X}}=\Pi_{S_\mathbf{p}\cap \overline{X}}\Pi_{S_\mathbf{p}}$.
\end{lemma}
\begin{proof}
Omitted.
\end{proof}

This will be used now to prove the following result which immediately gives us the first claim in our proposition. Namely, we now show that we may project the data while maintaining the set of minimizers.

\begin{lemma} \label{lem:proj.data}
For $j\geq 1$ we have
\begin{equation*}
\begin{split}
&\argmin_{\mathbf{p}\in P_{j-1}^\perp: \|\mathbf{p}\|=1}\frac{1}{N}\sum_{i=1}^N\inf_{z\in S_\mathbf{p}\cap \overline{X}}\|x_i-\overline{x}-z\|^2=\argmin_{\mathbf{p}\in P_{j-1}^\perp,:\|\mathbf{p}\|=1}\frac{1}{N}\sum_{i=1}^N\inf_{z\in S_\mathbf{p}\cap \overline{X}}\|\Pi_{P_{j-1}^\perp}(x_i-\overline{x})-z\|^2.
\end{split}
\end{equation*}
\end{lemma}

\begin{proof}
If $\mathbf{p}\in P_{j-1}^\perp$ we have, for any $i$,
\begin{align*}
    &\inf_{z\in S_\mathbf{p}\cap \overline{X}}\|x_i-\overline{x}-z\|^2\\
    &=\|x_i-\overline{x}-\Pi_{S_\mathbf{p}\cap \overline{X}}(x_i-\overline{x})\|^2\\
    &=\|\Pi_{P_{j-1}}(x_i-\overline{x})+\Pi_{P_{j-1}^\perp}(x_i-\overline{x})-\Pi_{S_\mathbf{p}\cap \overline{X}}(x_i-\overline{x})\|^2\\
    &=\|\Pi_{P_{j-1}^\perp}(x_i-\overline{x})-\Pi_{S_\mathbf{p}\cap \overline{X}}(x_i-\overline{x})\|^2+\|\Pi_{P_{j-1}}(x_i-\overline{x})\|^2\\
    &=\|\Pi_{P_{j-1}^\perp}(x_i-\overline{x})-\Pi_{S_\mathbf{p}\cap \overline{X}}\Pi_{P_{j-1}^\perp}(x_i-\overline{x})\|^2+\|\Pi_{P_{j-1}}(x_i-\overline{x})\|^2\\
    &=\inf_{z\in S_{\mathbf{p}}\cap \overline{X}}\|\Pi_{P_{j-1}^\perp}(x_i-\overline{x})-z\|^2+\|\Pi_{P_{j-1}}(x_i-\overline{x})\|^2.
\end{align*}
Note the fourth equality follows by Lemma \ref{lem:iterated.proj} as
\begin{align*}
    \Pi_{S_{\mathbf{p}}\cap \overline{X}}(x)&=\Pi_{S_{\mathbf{p}}\cap \overline{X}}\Pi_{S_\mathbf{p}}(\Pi_{P_{j-1}}x+\Pi_{P_{j-1}^\perp}x)=\Pi_{S_{\mathbf{p}}\cap \overline{X}}\Pi_{S_\mathbf{p}}\Pi_{P_{j-1}^\perp}x=\Pi_{S_{\mathbf{p}}\cap \overline{X}}\Pi_{P_{j-1}^\perp} x.
\end{align*}
Now, since the second term on the right hand side is constant in $\mathbf{p}$ the claim in the lemma then follows after averaging over the $x_i$.
\end{proof}

We now proceed to present a result whose statement implies the second claim of the proposition. 

\begin{lemma}\label{lem:drop.constr} For $j\geq 1$ we have
\begin{equation*}
\begin{split}
&\inf_{\mathbf{p}\in P_{j-1}^\perp: \|\mathbf{p}\|=1}\frac{1}{N}\sum_{i=1}^N\inf_{z\in S_\mathbf{p}}\|\Pi_{P_{j-1}^\perp}(x_i-\overline{x})-z\|^2\\
&=\inf_{\mathbf{p}\in H:\|\mathbf{p}\|=1}\frac{1}{N}\sum_{i=1}^N\inf_{z\in S_\mathbf{p}}\|\Pi_{P_{j-1}^\perp}(x_i-\overline{x})-z\|^2,
\end{split}
\end{equation*}
and if $(x_i-\overline{x})\not\in P_{j-1}$ for at least one $i$ then the minimizers coincide.
\end{lemma}

\begin{proof}
It suffices to show the relation ``$\leq$'' as the inequality ``$\geq$'' is immediate. We have using the definition of the projection for $\|\mathbf{p}\|=1$ that for any $i$:
\begin{align*}
    \inf_{z\in S_{\mathbf{p}}}\|\Pi_{P_{j-1}^\perp}(x_i-\overline{x})-z\|^2
    &=\big\|\Pi_{P_{j-1}^\perp}(x_i-\overline{x})-\langle\Pi_{P_{j-1}^\perp}(x_i-\overline{x}),\mathbf{p}\rangle\mathbf{p}\big\|^2 =: (\star).
\end{align*}
Now by the orthogonal decomposition for $\mathbf{p}$ and linearity of the inner product
\begin{align*}
    \langle\Pi_{P_{j-1}^\perp}(x_i-\overline{x}),\mathbf{p}\rangle&=\langle\Pi_{P_{j-1}^\perp}(x_i-\overline{x}),\Pi_{P_{j-1}^\perp}\mathbf{p}\rangle+\langle\Pi_{P_{j-1}^\perp}(x_i-\overline{x}),\Pi_{P_{j-1}}\mathbf{p}\rangle\\
    &=\langle\Pi_{P_{j-1}^\perp}(x_i-\overline{x}),\Pi_{P_{j-1}^\perp}\mathbf{p}\rangle,
\end{align*}
where the last equality follows by orthogonality. Hence we have:

\begin{align*}
    (\star)&=\big\|\Pi_{P_{j-1}^\perp}(x_i-\overline{x})-\langle\Pi_{P_{j-1}^\perp}(x_i-\overline{x}),\Pi_{P_{j-1}^\perp}\mathbf{p}\rangle\mathbf{p}\big\|^2\\
    &=\big\|\Pi_{P_{j-1}^\perp}(x_i-\overline{x})-\langle\Pi_{P_{j-1}^\perp}(x_i-\overline{x}),\Pi_{P_{j-1}^\perp}\mathbf{p}\rangle\Pi_{P_{j-1}^\perp}\mathbf{p}\big\|^2 \\
    & \quad \quad \quad \quad \quad \quad \quad + \big\|\langle\Pi_{P_{j-1}^\perp}(x_i-\overline{x}),\Pi_{P_{j-1}^\perp}\mathbf{p}\rangle\Pi_{P_{j-1}}\mathbf{p})\big\|^2\\
    &\geq \big\|\Pi_{P_{j-1}^\perp}(x_i-\overline{x})-\langle\Pi_{P_{j-1}^\perp}(x_i-\overline{x}),\Pi_{P_{j-1}^\perp}\mathbf{p}\rangle\Pi_{P_{j-1}^\perp}\mathbf{p}\big\|^2\\
    &\geq \inf_{z\in S_{\Pi_{P_{j-1}^\perp}\mathbf{p}}}\big\|\Pi_{P_{j-1}^\perp}(x_i-\overline{x})-z\big\|^2.\\
\end{align*}
Consequently for any $\mathbf{p}\in H$ with $\|\mathbf{p}\|=1$ we get:
\[
\frac{1}{N}\sum_{i=1}^N\inf_{z\in S_{\Pi_{P_{j-1}^\perp}\mathbf{p}}}\|\Pi_{P_{j-1}^\perp}(x_i-\overline{x})-z\|^2\leq\frac{1}{N}\sum_{i=1}^N\inf_{z\in S_\mathbf{p}}\|\Pi_{P_{j-1}^\perp}(x_i-\overline{x})-z\|^2.
\]
If $\Pi_{P_{j-1}^\perp}\mathbf{p}\not=0$ then $S_{\Pi_{P_{j-1}^\perp}\mathbf{p}}=S_{\Pi_{P_{j-1}^\perp}\mathbf{p}/\|\Pi_{P_{j-1}^\perp}\mathbf{p}\|}$. Otherwise, $\mathbf{p}\in P_{j-1}$ and equality holds in the above with $z=0$ which is admissible for any $S_{\mathbf{p}'}$ with $\mathbf{p}'\in P_{j-1}^\perp$. Hence, it is clear that we have
\[
\inf_{\mathbf{p}\in P_{j-1}^\perp, \|\mathbf{p}\|=1}\frac{1}{N}\sum_{i=1}^N\inf_{z\in S_\mathbf{p}}\|\Pi_{P_{j-1}^\perp}(x_i-\overline{x})-z\|^2\leq\inf_{\mathbf{p}:\|\mathbf{p}\|=1}\frac{1}{N}\sum_{i=1}^N\inf_{z\in S_\mathbf{p}}\|\Pi_{P_{j-1}^\perp}(x_i-\overline{x})-z\|^2,
\]
from which the claimed equality follows.

It remains to see that if $(x_i-\overline{x})\not\in P_{j-1}$ for at least one $i$ then the minimizers must coincide. It is clear that if $(x_i-\overline{x})\not\in P_{j-1}$ then $\Pi_{P_{j-1}^\perp}(x_i-\overline{x})\not=0$ and so if $\mathbf{p}\in P_{j-1}$ then it is not a minimizer of the RHS as it can be strictly improved upon. Moreover, if $\mathbf{p}\not\in P_{j-1}^\perp$ and $\mathbf{p}\not\in P_{j-1}$ then the chain of inequalities derived above is strict so $\mathbf{p}$ cannot be a minimizer of the RHS as $\|\Pi_{P_{j-1}^\perp}\mathbf{p}\|^{-1}\Pi_{P_{j-1}^\perp}\mathbf{p}\in P^\perp_{j-1}$ is an improvement. We conclude that any minimizer of the RHS must be in $P_{j-1}^\perp$ from which the claim follows.
\end{proof}

With these results in hand, we formally summarize the conclusions below to complete the proof of the proposition.

\begin{proof}[Proof of Proposition \ref{prop:equiv.rep.cpca}]
By a reparameterization of the infimums (see equation \ref{eqn:equiv.parameterizations}) in the definition of $V$ and $\hat{V}_j$, the statement (i) follows immediately from Lemma \ref{lem:proj.data} and (ii) similarly follows from Lemma \ref{lem:drop.constr}. 
\end{proof}

\subsection{Comparison with alternative approaches for Wasserstein GPCA}\label{app:comparision}

We compare two alternative approaches for the dimensionality reduction of distribution-valued data in the Wasserstein space over the interval with the methodology proposed in this paper. 

The first approach, due to \cite{cazelles2018geodesic}, employs a proximal forward-backward splitting method to solve the nonlinear (and constrained) nested GPCA problem which they call the ``iterative geodesic approach''. Their algorithm directly takes as an input the data histograms and they choose the base measure $\mathbb{P}_0$ to coincide with the Fr\'echet mean $\bar{\mathbb{P}}$. To the best of our knowledge, it represents the first numerically tractable implementation of GPCA in the literature.

The second approach, due to \cite{PB22} (which we denote P \& B), is a projected method that relaxes the constraints of (nested) GPCA. To reduce the computational complexity, the authors approximate the transport maps with monotone quadratic B-splines. Then, they perform an unconstrained optimization to search for the principal components of the data in the space of basis function coefficients. By restricting the resulting principal components to the convex constraint set (in our notation $X$), they obtain ``projected'' principal components. When the constraints are not binding, these exactly coincide with geodesic principal components on the approximated quantile functions (see Section \ref{sec:Application.GPCA}).

To compare these methods with ours (which we denote by C \& W), we simulate synthetic ranked return distributions that replicate qualitatively the features observed in Section \ref{sec:Application.GPCA} using CRSP data (which was used under license). Specifically, we simulate return data for $N=101$ stocks from the Atlas model of equity markets (see \cite{banner2005atlas,fernholz2013second,ichiba2011hybrid})\footnote{The code to simulate the data and reproduce the comparisons discussed here has been made available at \url{https://github.com/stevenacampbell/ConvexPCA}}. We choose the model parameters so that the expected return and volatility of a stock increases as a function of its capitalization rank. To analyze the resulting data using the approach of \cite{cazelles2018geodesic}, we export histograms of various dimensions and directly use their publicly available code\footnote{\url{https://github.com/ecazelles/2017-GPCA-vs-LogPCA-Wasserstein}} which is written in MATLAB. On the other hand, in order to compare the spline approximation proposed by \cite{PB22} with the piecewise constant approximation of Section \ref{sec:Application.GPCA} we reproduce their publicly available Python code\footnote{\url{https://github.com/mberaha/ProjectedWasserstein}} in $\tt{R}$.

\begin{table}[h]
\centering
\begin{minipage}{0.45\textwidth}
    \centering
    \begin{tabular}{|c|c|c|c|}
    \hline
    Dimension & P \& B & C \& W & Cazelles et al. \\ \hline
    $2^4$ & 0.08s & 0.01s & 28s \\ \hline
    $2^5$ & 0.14s & 0.01s & 38s \\ \hline
    $2^6$ & 0.42s & 0.03s & 52s \\ \hline
    $2^7$ & 1.71s & 0.66s & 142s \\ \hline
    $2^8$ & 7.67s & 1.45s & 334s \\ \hline
    $2^9$ & 57s & 4.94s & 563s \\ \hline
    \end{tabular}
    \caption{Computational times for 1st and 2nd principal components (geodesic or projected).}
    \label{tab:opt.times}
\end{minipage}%
\hfill
\begin{minipage}{0.45\textwidth}
    \centering
    \begin{tabular}{|c|c|c|}
    \hline
    Dimension & P \& B & C \& W \\ \hline
    $2^4$ & 0.0088 & 0.0297 \\ \hline
    $2^5$ & 0.0059 & 0.0191 \\ \hline
    $2^6$ & 0.0043 & 0.0125 \\ \hline
    $2^7$ & 0.0034 & 0.0084 \\ \hline
    $2^8$ & 0.0030 & 0.0058 \\ \hline
    $2^9$ & 0.0028 & 0.0042 \\ \hline
    \end{tabular}
    \caption{Reconstruction error using 1st and 2nd principal components (geodesic or projected).}
    \label{tab:rec.error}
\end{minipage}
\end{table}

A summary of the computation times\footnote{All of the results shown here are for \textit{serial} implementations of the algorithms run on a  laptop computer with an Intel Core i7-1185G7 processor and 32.0 GB of RAM.} for running each of the optimizations for the first two principal components as a function of the problem dimension (i.e. the number of histogram bins, basis splines, or piecewise constant basis functions) is given in Table \ref{tab:opt.times}. Included in these times for P \& B and C \& W is the cost of fitting the basis functions. Associated with each data point is a high resolution quantile function evaluated on a dyadic partition with $2^{12}$ points. We take the distributions that these characterize as the ``true'' distributional data. That is, in the notation of Section \ref{sec:Application.GPCA}, we let these quantile functions define $\mathbb{P}_1,\dots, \mathbb{P}_N$. Let $\hat{\mathbb{P}}_i$ for $i=1,\dots,N$ be the approximate measures that are recovered by first approximating the quantile functions and then projecting onto the first two geodesic (or projected) principal components. We define the \textit{reconstruction error} of the estimates as their average Wasserstein distance to the true measures, $n^{-1}\sum_{i=1}^NW_2(\mathbb{P}_i,\hat{\mathbb{P}}_i)$. These errors as a function of the problem dimension are summarized in Table \ref{tab:rec.error} for each method.

We provide several comments to contextualize these results. In Table \ref{tab:opt.times} we observe that the implementation of this paper has a distinct speed advantage as a function of the \textit{dimension}. We would like to note that while we have put considerable effort into improving the numerical efficiency of our algorithm, we have not made any effort to reduce the computational cost of the other approaches. Moreover, the computational time of our method is sensitive to the initial ``guess'' of the solution. In particular, in this paper we have \textit{not} suggested any tailored optimization routine like in \cite{cazelles2018geodesic} and use a gradient-descent style solver available in $\tt{R}$. Our goal was to make objective function and gradient evaluations sufficiently inexpensive so that the problem comes in reach of traditional numerical methods (see Tables \ref{tab:comp.time} and \ref{tab:comp.time.sph} for a comparison of the computational times with $N$ inner product evaluations of the same dimension in $\tt{R}$). This means that our implementation will suffer from all of the shortcomings associated with applying the chosen optimization routine to high dimensional non-linear functions. If our initial ``guess'' is far from the true minimum then gradient descent may require many iterations to converge up to numerical precision or converge to a local (rather than global) minimum. For this problem, the initialization point that we suggest in Appendix \ref{app:init.pt} is sufficiently close to the minimum that we do not need a very large number of iterations to converge.

\begin{table}[h]
\centering
\begin{minipage}{0.45\textwidth}
    \centering
    \begin{tabular}{|c|c|c|c|}
    \hline
    Dimension & $V(\cdot)$  & $\nabla V(\cdot)$ & $\langle \cdot, \cdot \rangle$ \\ \hline
    $2^4$ & 6$\mu$s & 139$\mu$s & 30.4$\mu$s\\ \hline
    $2^5$ & 13$\mu$s & 710$\mu$s &  34.9$\mu$s\\ \hline
    $2^6$ & 27$\mu$s & 3.18ms& 52.1$\mu$s \\ \hline
    $2^7$ & 43$\mu$s & 15.0ms & 67.7$\mu$s\\ \hline
    $2^8$ & 78$\mu$s & 44.9ms & 105.3$\mu$s\\ \hline
    $2^9$ & 215$\mu$s & 195ms & 184.4$\mu$s\\ \hline
    \end{tabular}
    \caption{Computational cost of function evaluation alongside $N=101$ (number of samples) repeated inner product evaluations.}
    \label{tab:comp.time}
\end{minipage}%
\hfill
\begin{minipage}{0.45\textwidth}
    \centering
    \begin{tabular}{|c|c|c|}
    \hline
    Dimension & $\tilde{V}(\boldsymbol{\omega}(\cdot))$  & $\nabla \tilde{V}(\boldsymbol{\omega}(\cdot))$  \\ \hline
    $2^4$ & 7.3$\mu$s & 183$\mu$s\\ \hline
    $2^5$ & 17.3$\mu$s & 902$\mu$s  \\ \hline
    $2^6$ & 32.5$\mu$s & 4.08ms \\ \hline
    $2^7$ & 57.5$\mu$s & 16.4ms \\ \hline
    $2^8$ & 116$\mu$s & 69.6ms \\ \hline
    $2^9$ & 282$\mu$s& 331ms \\ \hline
    \end{tabular}
    \caption{Computational cost of function evaluation in spherical coordinates with basis matrix $B_1=I$ (see \ref{eqn:nested.PCA.polyhedral2}).}
    \label{tab:comp.time.sph}
\end{minipage}
\end{table}

We also highlight that the unconstrained approach of \cite{PB22} actually obtains all of the projected principal components at once (rather than just the first two) from the eigendecomposition of a specific matrix that arises in their formulation. Moreover, as seen in Table \ref{tab:rec.error}, they need fewer basis functions to achieve the same level of reconstruction error. This is despite the fact that the convex PCA solution minimizes a quantity that is closely related to this error. There are two probable explanations for this phenomenon. First, the difference between the first two convex and projected PCs is not large relative to the approximation error. Second, the distributional data in this example is ``nice'' enough that the transport maps can be efficiently approximated by a low dimensional spline basis. In this sense, the basis chosen for the projected PCA is more efficient than the piecewise constant basis we use in the default implemenation of convex PCA.

On a related note, the approximations proposed here have the effect of reducing the burden imposed by the constraints. In GPCA, a constraint is binding if the quantile function is not \textit{strictly} monotone or takes a value at the boundary of the sample space $\Omega=[a,b]$. Lower resolution basis function approximations can have the effect of pulling the data away from these boundaries since they locally average the function values at nearby points in the domain. This, along with other structural features of the chosen basis, can have the effect of distorting the resulting principal components. In Figure \ref{fig:comparison.pcs} we provide an illustration of the projected principal components arising from a low dimensional spline approximation ($2^4$) with the geodesic principal components arising from a piecewise constant approximation of higher dimension ($2^7$). Despite having similar reconstruction errors, the second principal components still exhibit some visible differences. While not an issue here, if the true data is sufficiently concentrated near the boundary (as in the simple example illustrated in Figure \ref{fig:Low.Dim.Example}) the projected and geodesic components can differ materially. 

\begin{figure}[t!]
    \centering
    \begin{center}
        \includegraphics[width=0.49\textwidth]{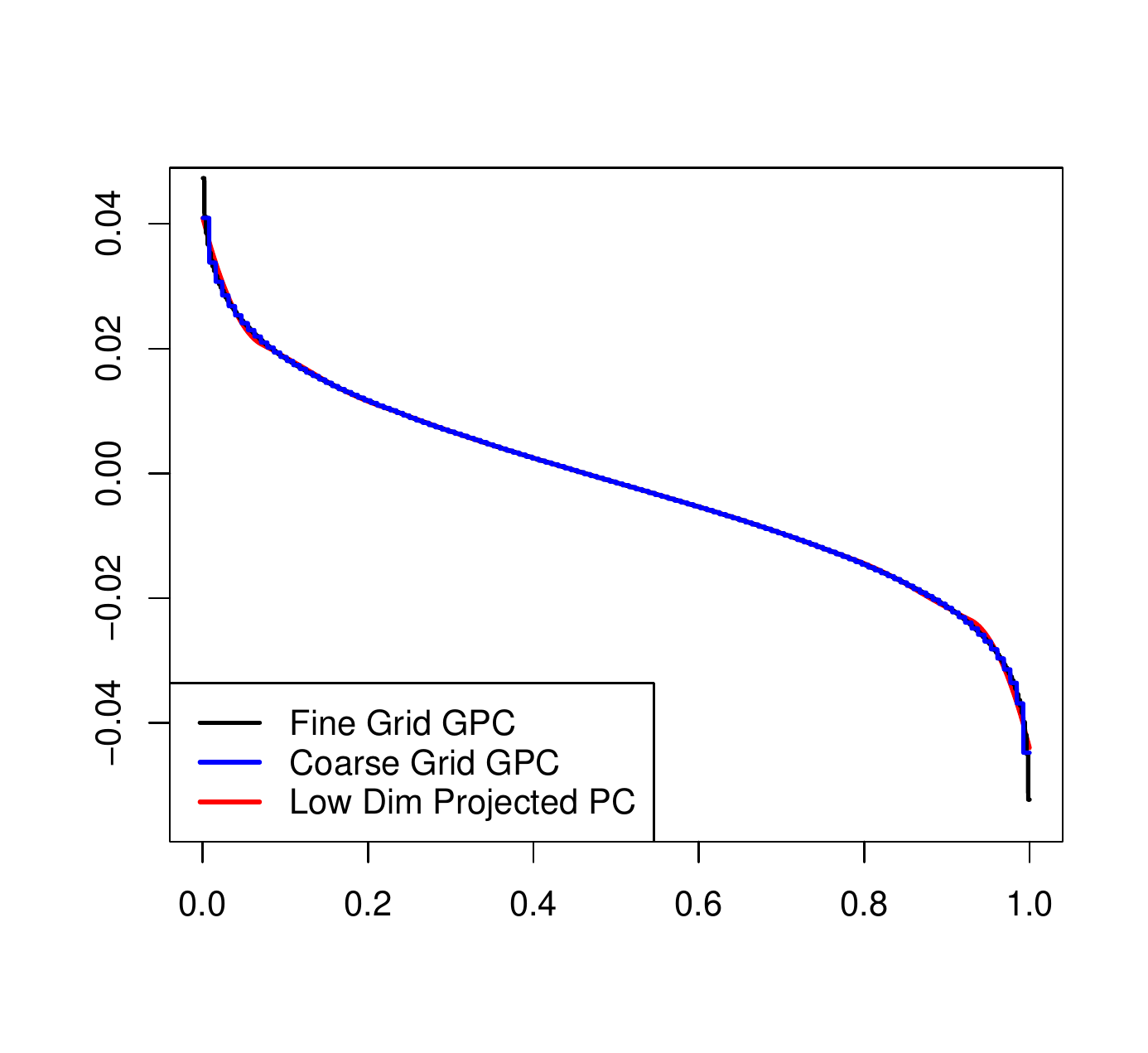}
        \includegraphics[width=0.49\textwidth]{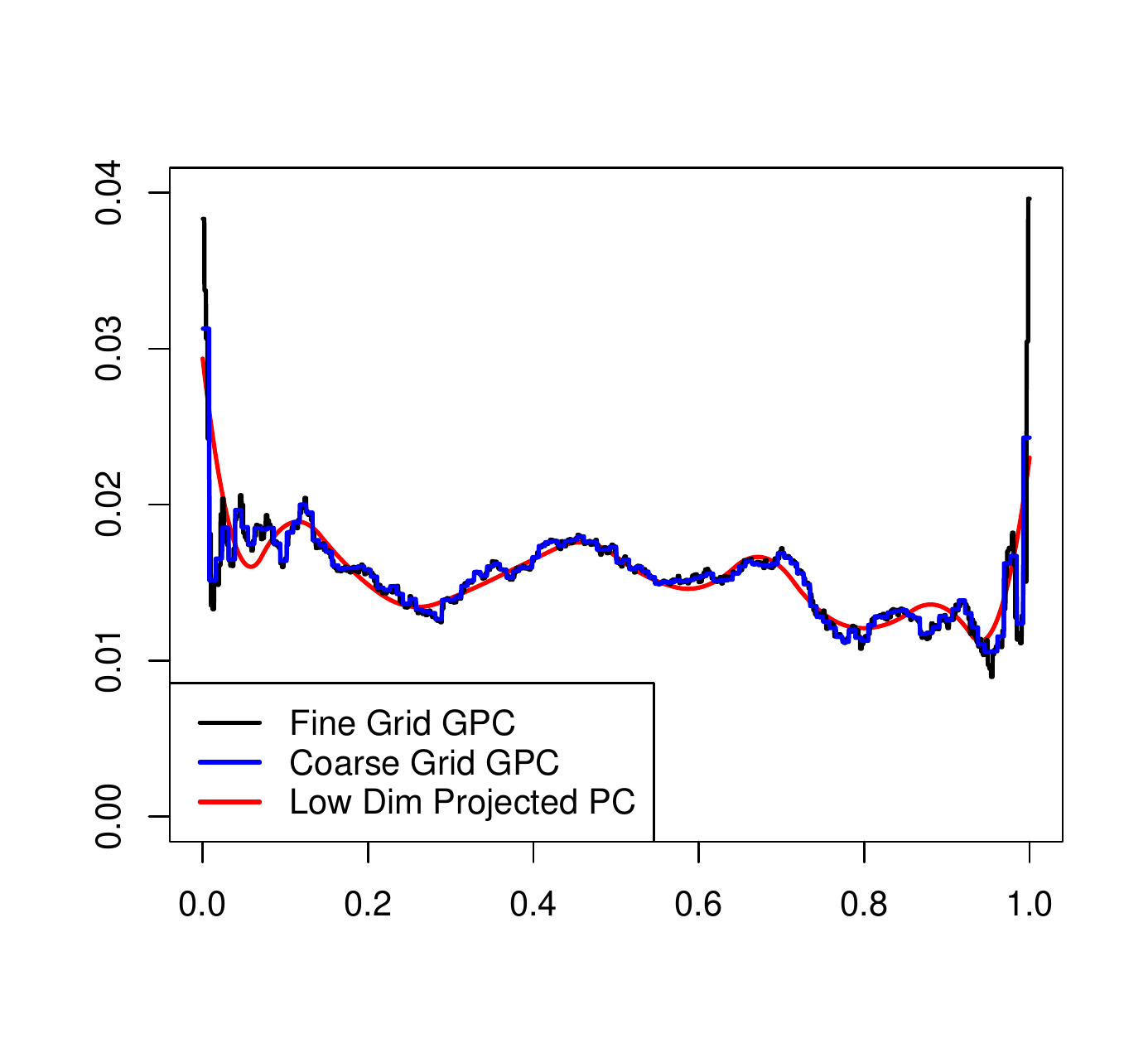}
    \end{center}
    \vspace{-0.5cm}
    \caption{First (left) and second (right) geodesic and projected PCs (represented as perturbation directions of the quantile functions on $[0,1]$). The fine grid GPC corresponds to a basis dimension of $2^9$ and serves as a higher resolution reference for the other estimates. The coarse grid GPC has basis dimension $2^7$ and the low dimensional projected PC corresponds to a spline basis of size $2^4$.}
    \label{fig:comparison.pcs}
\end{figure}


\section{Miscellaneous definitions and results}\label{sec:appendix.external.def.and.results}

\subsection{Convergence of sets and functions}
Throughout this subsection we let a metric space $(\mathcal{Y},\mathsf{d})$ be given. For $y \in \mathcal{Y}$, we let $\mathcal{N}(y)$ be the set of all open neighborhoods of $y$. Let $F_n$, $n \geq 1$, and $F$ be functions from $\mathcal{Y}$ to $\overline{\mathbb{R}}$.

\begin{definition}[Continuous convergence] \label{def:cont.convergence}\parencite[Definition 4.7]{dal2012introduction}
We say that $F_n$ continuously converges to $F$ if for every $y\in \mathcal{Y}$ and every neighborhood $V$ of $F(y)$ there exists an integer $k$ and $U\in \mathcal{N}(y)$, such that for $n\geq k$ and $y\in U$ we have $F_n(y)\in V$.
\end{definition}

\begin{definition}[$\Gamma$-convergence] \label{def:gamma.convergence}\parencite[Proposition 8.1]{dal2012introduction}
We say that $(F_n)_{n\geq0}$ $\Gamma$-converges to $F$, denoted $\Glim_{n\to\infty}F_n=F$, if the following holds for all $y \in \mathcal{Y}$.
\begin{enumerate}
    \item[(i)] $F(y)\leq\liminf_{n\to\infty}F_n(y_n)$ for any $y_n\in \mathcal{Y}$, $n\geq1$ with $y_n\to y$ and;
    \item[(ii)] There exists $y_n\in \mathcal{Y}$, $n\geq1$, with $y_n\to y$ such that $F(y)=\lim_{n\to\infty}F_n(y_n)$.
\end{enumerate}
\end{definition}

\begin{proposition}\label{prop:consistency.gamma.conv}\parencite[Theorems 7.8 and 7.23]{dal2012introduction}
Suppose that $\mathcal{Y}$ is compact and $\Glim_{n\to\infty} F_n=F$. Then $M(F) = \left\{y\in\mathcal{Y}:F(y)=\inf_{z\in \mathcal{Y}}F(z) \right\}$ 
is non-empty, $\lim_{n\to\infty}\inf_{y\in\mathcal{Y}}F_n(y)=\min_{y\in\mathcal{Y}} F(y)$, and if $y_n\in M(F_n)$, $n\geq 1$ then the accumulation points of $(y_n)_{n\geq 1}$ belong to $M(F)$.
\end{proposition}

\begin{definition}[Kuratowski Convergence]\label{def:kuratowski}\parencite[Remark 8.2.]{dal2012introduction} 
Let $A,A_n\subset \mathcal{Y}$, $n\geq1$. We say that the $A_n$ converges to $A$ in the sense of Kuratowski, denoted by $\Klim_{n\to\infty}A_n=A$, if
\begin{enumerate}
    \item[(i)] for all $y\in A$, there exists a constant $m\in\mathbb{N}$ and $y_n\in A_n$ for $n\geq m$ such that $y_n\to y$, and;
    \item[(ii)] for all $y_n\in A_n$, $n\geq 1$, and for any accumulation point $y$ of $(y_n)$, $y\in A$.
\end{enumerate}
\end{definition}

\begin{lemma}\label{lem:kur.haus}\parencite[Remark A.2]{BGKL17}
Convergence with respect to the Hausdorff distance implies convergence in the sense of Kuratowski. Moreover, if $\mathcal{Y}$ is compact both notions of convergence coincide.
\end{lemma}

\begin{lemma}\label{lem:kur.gamma}\parencite[Proposition 4.15]{attouch1984variational}
If $A_n,A\subset \mathcal{Y}$ for $n\geq 1$ then $\Klim_{n\to\infty}A_n=A$ if and only if $\Glim_{n\to\infty}\chi_{A_n}=\chi_A$.
\end{lemma}

\subsection{Correspondence and Berge's Maximum Theorem}
The following definitions and results, which can be found in \cite{guide2006infinite}, are used to prove Theorem \ref{thm:cont.pca}. Recall that a correspondence $\Phi:X\twoheadrightarrow Y$ maps each $x \in X$ to a subset of $Y$. Here we assume that $X$ and $Y$ are topological spaces. The graph of $\Phi$ is defined by
\[
\mathrm{Gr}(\Phi) = \{(x,y)\in X\times Y: y\in \Phi(x)\}.
\]

\begin{definition}
Let $\Phi:X\twoheadrightarrow Y$ be a correspondence and $x \in X$.
\begin{enumerate}
    \item[(i)] We say that $\Phi$ is upper hemicontinuous at $x$ if for every neighborhood $U$ of $\Phi(x)$, there is a neighborhood $V$ of $x$ such that $z\in V$ implies $\Phi(z)\subset U$.
    \item[(ii)] We say that $\Phi$ lower hemicontinuous at the point $x$ if for every open set $U$ that intersects  $\Phi(x)$, there is a neighborhood $V$ of $x$ such that $z\in V$ implies $\Phi(z)\cap U\not=\emptyset$.
\end{enumerate}
If $\Phi$ is both upper and lower hemicontinuous at $x$ we say it is continuous at $x$.
\end{definition}

\begin{lemma}\label{lem:hemi.singleton}
A singleton-valued correspondence is upper hemicontinuous if and only if it is lower hemicontinous, in which case it is continuous as a function.
\end{lemma}

\begin{theorem}\label{thm:cont.func.convex.hull.cont}
Given $f_1, \ldots, f_k :X \rightarrow  Y$, define a correspondence $\Phi:X\twoheadrightarrow Y$ by $\Phi(x)=\{f_1(x), \ldots ,f_k(x)\}$. If each $f_i$ is continuous at some point $x_0$ then:
    \begin{enumerate}
        \item[(i)] $\Phi$ is continuous at $x_0$.
        \item[(ii)] If $Y$ is a locally convex topological vector space, then the correspondence $x \mapsto \mathrm{conv}(\Phi(x))$ is also continuous at $x_0$.
    \end{enumerate}
\end{theorem}

\begin{theorem}[Berge's Maximum Theorem] \label{thm:berge} Let $\Phi:X\twoheadrightarrow Y$ be a continuous correspondence such that each $\Phi(x)$ is a nomempty compact subset of $Y$, and let $f:\mathrm{Gr}(\Phi) \to \mathbb{R}$ be continuous. Define the value function $m:X\to\mathbb{R}$ by
\[
m(x) = \max_{y\in\Phi(x)}f(x,y),
\]
and the correspondence $\mu:X\twoheadrightarrow Y$ of maximizers by 
\[\mu(x)=\{y\in\Phi(x):f(x,y)=m(x)\}.\]
Then: $m$ is continuous and $\mu$ has non-empty compact values. Moreover, if either $f$ has a continuous extension to all of $X\times Y$ or $Y$ is Hausdorff, then $\mu$ is upper hemicontinuous.
\end{theorem}

\subsection{Generalized envelope theorem} \label{sec:envelope}
The following result, taken from \cite[Theorem 1]{marimon2021envelope} is used in the proof of Theorem \ref{thm:pca.diff}. For convenience, we state the result using notations that are compatible with ours.

\begin{theorem}[Generalized envelope theorem]\label{thm:gen.env}
Consider the problem
\[
V(\mathbf{p})=\sup\left\{ J(\mathbf{a},\mathbf{p}): \mathbf{a}\in\mathcal{A} \text{ and } g_j(\mathbf{a},\mathbf{p})\geq 0 \text{ for $j = 1, \ldots, k$} \right\},
\]
where $\mathcal{A}\subset \mathbb{R}^m$, $\mathbf{p}\in\mathcal{P}\subset \mathbb{R}^n$ and $J$ and $g_i$ are real-vaued functions on $\mathbb{R}^m \times \mathbb{R}^n$. For $\boldsymbol{\lambda} = (\lambda_1, \ldots, \lambda_k) \in \mathbb{R}_+^k$ consider the Lagrangian
\[
\mathcal{L}(\mathbf{p},\mathbf{a},\boldsymbol{\lambda})=J(\mathbf{a},\mathbf{p})+ \sum_{j = 1}^k \lambda_j  \mathbf{g}(\mathbf{a},\mathbf{p}),
\]
and assume that the following regularity conditions hold.
\begin{enumerate}
        \item[(i)] The set $\mathcal{A}$ is convex.
        \item[(ii)] The functions $J$ and $g_j$, $j=1, \ldots ,k$, are continuous on $\mathbb{R}^m \times \mathbb{R}^n$.
        \item[(iii)] The constraint set $\Phi(\mathbf{p})=\{\mathbf{a}\in \mathcal{A}: g_j(\mathbf{a},\mathbf{p})\geq 0, \ j=1, \ldots ,k\}$ is compact for every $\mathbf{p}\in\mathcal{P}$.
        \item[(iv)] The correspondence $\Phi(\mathbf{p})$ is continuous.
        \item[(v)] For every $\mathbf{p}\in\mathcal{P}$ there exist $j$ and  $\hat{\mathbf{a}}_j\in\mathcal{A}$ such that $g_j(\hat{\mathbf{a}}_j,\mathbf{p})>0$ and $g_{\ell}(\hat{\mathbf{a}}_{\ell},\mathbf{p})\geq0$ for ${\ell}\not=j$.
        \item[(vi)] The set of saddle-points of $\mathcal{L}$, which has the form $\mathcal{A}^*(\mathbf{p})\times\Lambda^*(\mathbf{p})$, is non-empty for every $\mathbf{p}\in\mathcal{P}$.\footnote{For the definition of saddle point see \cite[Section 2]{marimon2021envelope}.}
        \item[(vii)] The gradients with respect to $\mathbf{p}$, $\nabla_\mathbf{p} J,\nabla_\mathbf{p} g_j$, exist and are jointly continuous in $(\mathbf{a},\mathbf{p})$.
\end{enumerate}
Then the directional derivatives of $V$ exist for ${\bf p} \in \intt(\mathcal{P})$ and are given by
\[
V'(\mathbf{p};\hat{\mathbf{p}})=\max_{\mathbf{a}\in\mathcal{A}^*(\mathbf{p})}\min_{\boldsymbol{\lambda}\in\Lambda^*(\mathbf{p})}\left[\nabla_\mathbf{p} J(\mathbf{a},\mathbf{p})+\sum_{j=1}^k\lambda_i\nabla_\mathbf{p}g_j(\mathbf{a},\mathbf{p})\right]\hat{\mathbf{p}},
\]
and the order of maximum and minimum does not matter.
\end{theorem}

\begin{corollary}\label{cor:gen.env.diff}
Suppose that in addition to the assumptions of Theorem \ref{thm:gen.env} the following conditions hold at a given $\mathbf{p}_0\in\mathrm{int}(\mathcal{P})$:
\begin{enumerate}
    \item[(i)] $J$ is strictly concave in $\mathbf{a}$ and $g_j$ is concave in $\mathbf{a}$, $j=1,...,k$.
    \item[(ii)] $J$ and $g_j$ are continuously differentiable in $\mathbf{a}$.
    \item[(iii)] The Linear Independence Constraint Qualification (LICQ) holds. That is, the vectors $\nabla_\mathbf{a}g_j(\mathbf{a}^*,\mathbf{p}_0)$ are linearly independent for binding constraints $i\in\mathcal{I}(\mathbf{a}^*,\mathbf{p}_0) =\{j:g_j(\mathbf{a}^*,\mathbf{p}_0)=0\}$.
\end{enumerate}
Then the saddle-point $(\mathbf{a}^*(\mathbf{p}_0),\boldsymbol{\lambda}^*(\mathbf{p}_0))$ is unique and the gradient of $V$ exists at $\mathbf{p}_0$ and is given by
\[
\nabla_\mathbf{p}V(\mathbf{p}_0)=\nabla_\mathbf{p}J(\mathbf{a}^*(\mathbf{p}_0),\mathbf{p}_0)+\sum_{j=1}^k\lambda^*_j(\mathbf{p}_0)\nabla_\mathbf{p}g_j(\mathbf{a}^*(\mathbf{p}_0),\mathbf{p}_0).
\]
\end{corollary}
\begin{proof}
Condition (i) implies the uniqueness of the saddle point solution $\mathbf{a}^*(\mathbf{p}_0)\in \mathcal{A}^*(\mathbf{p}_0)$ and conditions (ii) and (iii) imply the uniqueness of the multiplier $\lambda^*(\mathbf{p}_0)\in\Lambda^*(\mathbf{p}_0)$ (see the discussion in \cite{marimon2021envelope}). Consequently, the saddle point corresponding to $\mathbf{p}_0$ is unique and the representation for the gradient follows from Theorem \ref{thm:gen.env}.
\end{proof}

\subsection{Miscellaneous results}
The following lemma summarizes some properties of level sets of convex functions. Details can be found in \cite{rockafellar1970convex}.



\begin{lemma} \label{lem:constr.set.properties}
Let $g : \mathbb{R}^n \rightarrow \mathbb{R} \cup \{\infty\}$ be convex. Then $L =\{\mathbf{x}\in\mathbb{R}^n: g(\mathbf{x})< 0\}$ is open and convex. If $L$ is non-empty then the closure of $L$ is given by $\overline{L}=\{\mathbf{x}\in\mathbb{R}^n: g(\mathbf{x})\leq 0\}$ and $L = \intt(\overline{L})$. In particular, we have $\partial L=\partial \bar{L}=\{\mathbf{x}\in\mathbb{R}^n: g(\mathbf{x})= 0\}$.

Now let $g_1, \ldots, g_k$ be convex, and let $L_j = \{ {\bf x} \in \mathbb{R}^n : g_j({\bf x}) < 0\}$. Suppose $C = \bigcap_{j = 1}^k L_j$ is nonempty. Then  $\partial C \subset \bigcup_{j = 1}^k \partial L_j$.
\end{lemma}

\newpage
\begingroup
\printbibliography
\endgroup
\end{document}